\documentclass[11pt]{amsart}

\usepackage{latexsym}
\usepackage{amssymb}
\usepackage{amsmath}
\usepackage{color}
\usepackage{bbm}
\usepackage{hyperref}
\usepackage{lmodern}

\allowdisplaybreaks 

\usepackage{graphicx}
\usepackage{tcolorbox}
\usepackage{tabularx}
\usepackage{tikz-cd}

\newtheorem{theorem}{Theorem}[section]
\newtheorem{lemma}[theorem]{Lemma}
\newtheorem{proposition}[theorem]{Proposition}
\newtheorem{corollary}[theorem]{Corollary}

\theoremstyle{definition}
\newtheorem{defn}[theorem]{Definition}

\newtheorem{example}[theorem]{Example}
\newtheorem{remark}[theorem]{Remark}

\newtheorem{question}[theorem]{Question}

\DeclareMathOperator{\id}{id}
\DeclareMathOperator{\tr}{Tr}
\DeclareMathOperator{\Gg}{\mathcal{G}}
\DeclareMathOperator{\G}{\mathbb{G}}
\DeclareMathOperator{\Hbb}{\mathbb{H}}

\DeclareMathOperator{\GH}{\mathbb{G}\times \mathbb{H}}

\newcommand\nph{\varphi}

\newcommand{\cl}[1]{\mathcal{#1}}
\newcommand{\bb}[1]{\mathbb{#1}}
\DeclareMathOperator{\Tr}{Tr}

\newcommand{\mc}[1]{\mathcal{#1}}
\newcommand{\la}{\langle}
\newcommand{\ra}{\rangle}
\newcommand{\hs}{\hskip2pt}
\newcommand{\ten}{\otimes}

\newcommand{\vphi}{\varphi}

\providecommand{\wt}[1]{\widetilde{#1}}
\providecommand{\norm}[1]{\lVert#1\rVert}

\newcommand{\N}{\mathbb{N}}
\newcommand{\C}{\mathbb{C}}

\begin{document}

\title
{An operator system approach to self-testing}


\author[J. Crann]{Jason Crann$^{1}$}
\address{$^1$School of Mathematics \& Statistics, Carleton University, Ottawa, ON, Canada H1S 5B6}
\email{jasoncrann@cunet.carleton.ca}

\author[I. G. Todorov]{Ivan G. Todorov$^{2}$}
\address{$^{2}$School of Mathematical Sciences, University of Delaware, 501 Ewing Hall, Newark, DE 19716, USA}
\email{todorov@udel.edu}

\author[L. Turowska]{Lyudmila Turowska$^{3}$}
\address{$^{3}$Department of Mathematical Sciences, Chalmers University of Technology and the University of Gothenburg, Gothenburg SE-412 96, Sweden}
\email{turowska@chalmers.se}

\begin{abstract} 
We develop a general framework for self-testing, in which bipartite 
correlations are described by states on the commuting tensor product
of a pair of operator systems. We propose a definition of a local isometry 
between bipartite quantum systems in the commuting operator model, and 
define self-testing and abstract self-testing in the latter generality. 
We show that self-tests are in the general case 
always abstract self-tests and that, in 
some cases, the converse is also true. 
We apply our framework in a variety of 
instances, including to correlations with quantum inputs and outputs, 
quantum commuting correlations for the CHSH game, synchronous correlations, 
contextuality scenarios, quantum colourings
and Schur quantum channels. 
\end{abstract}

\date{22 June 2025}

\maketitle

\tableofcontents



\section{Introduction}\label{s_intro}

About two decades ago, a remarkable feature of quantum computation was uncovered
in \cite{myao}: quantum devices can be certified
either by a user or by an independent third party by only observing 
their outputs when completing specific quantum computational tasks. 
This feature became known as \emph{self-testing}, and was intensively studied in the 
subsequent years, in particular in connection with device-independent quantum 
cryptography (see e.g. \cite{bsca-1,bsca-2,pabgms, ruv}). 
The notion was used in the celebrated solution of 
the Tsirelson Problem in \cite{jnvwy}, and 
a number of recent advances on the topic, from a variety of perspectives, has been 
achieved in \cite{broadbent, cgjv, mps, mpw, pszz, zhao}. 
A survey of the concept and associated results can be found in \cite{supic-bowles}. 

In its simplest form, namely, the (bipartite) Bell scenario, self-testing rests on the correlation between the behaviours of two non-communicating 
parties, Alice and Bob, performing a quantum experiment. 
Each of the two parties holds a quantum system, modeled by two finite dimensional 
Hilbert spaces $H_A$ and $H_B$, and share an entangled state, 
modeled by a unit vector $\xi$ in the tensor product $H_A\otimes H_B$. 
Alice (resp. Bob) has access to measurement devices, indexed by 
a finite set $X$ (resp. $Y$), modeled by 
POVM's, say $(E_{x,a})_{a\in A}$, $x\in X$,  
(resp. $(F_{y,b})_{b\in B}$, $y\in Y$), where $A$ (resp. $B$) 
is the set of outputs for Alice (resp. Bob). 
The tuple $M = (H_A,H_B,(E_{x,a})_{a\in A},(F_{y,b})_{b\in B},\xi)$
yields a \emph{correlation}, that is, a family $p_M$ of conditional 
probability distributions over $A\times B$, given by
\begin{equation}\label{eq_first}
p_M(a,b|x,y) = \langle (E_{x,a}\otimes F_{y,b})\xi,\xi\rangle, \ \ \ 
x\in X, y\in Y, a\in A, b\in B.
\end{equation}
We say that the tuple $M$ is a model of $p_M$;
naturally, different models may yield the same correlation $p_M$. 
In the Bell scenario, self-testing consists in the fact that, in some cases, the correlation $p_M$
determines uniquely (up to a natural equivalence) the model $M$. 
Setups for self-testing, distinct from the Bell scenario, have also been 
considered; for example, it was shown in \cite{brv-etal} that
probabilistic assignments of some contextuality scenarios can be 
self-tested, thus enlarging the capabilities of self-testing 
beyond classical bipartite experiments.

Recent operator algebraic approaches to self-testing in the 
Bell scenario were made in \cite{mps}, \cite{pszz} and \cite{zhao}.
They rely on the fact that, given the index sets $X$ and $A$ for Alice, 
there exist an operator system $\cl S_{X,A}$ 
(resp. a C*-algebra $\cl A_{X,A}$) that 
is universal for families of POVM's $(E_{x,a})_{a\in A}$, 
indexed over $X$, in that the latter 
completely determine the unital completely positive maps 
(resp. unital *-homomorphisms) on the former. 
This led to an abstract view on self-testing, developed in 
\cite{pszz} and pursued further in \cite{zhao}, which resides in 
the states of the commuting tensor product $\cl S_{X,A}\otimes_{\rm c} \cl S_{Y,B}$ 
\cite{KPTT11} admitting a unique extension to a state on the 
maximal tensor product $\cl A_{X,A}\otimes_{\max} \cl A_{Y,B}$. 
It was shown in \cite{pszz} that, in some cases, self-testing and abstract 
self-testing coincide. 

In a parallel development, the \emph{commuting operator model}
of quantum mechanics has been gaining importance over 
the past decade  
as a genuinely different (due to \cite{jnvwy}) and useful model 
in quantum information theory; see, for example, \cite{cklt, cltt, psstw, pt}. 
Abstract self-testing was studied in the commuting operator model 
in \cite{pszz}, but no advances have so far been made on providing an  
operational definition of self-testing in the latter context. 

On the other hand, 
in connections with non-local games with quantum inputs and quantum outputs, 
no-signalling correlations were extended to the quantum setting in 
\cite{tt-QNS}, and further used as strategies for such games, 
notably for games arising from quantum graphs, in \cite{bhtt-JFA}
and \cite{bhtt-Adv}. The question of a 
suitable framework for self-testing of quantum-to-quantum, or
classical-to-quantum, correlations has been, however, open. 

The purpose of the present work is to fill both of these gaps. 
We provide a far-reaching framework for self-testing, which includes 
as a special case self-testing in the Bell scenario for 
classical bipartitie no-signalling correlations, as well as  
known instances of self-testing for a single contextuality scenario, 
simultaneously developing an operational concept of self-testing 
in the commuting operator framework. 
Our setup includes self-testing of all correlation types of current interest, 
and is developed in a general case, where the correlations between Alice and Bob 
arise from an arbitrary pair of operator systems. 
We apply our framework in a variety of 
new instances: quantum commuting self-testing for the CHSH game, 
self-testing of classical correlations arising from 
representations of the Clifford relations, of synchronous correlations
of quantum type, of probabilistic assignments over a pair of 
contextuality scenarios, of perfect strategies for quantum colourings
of classical complete graphs, as well as of quantum channels 
arising from Schur multipliers. 

We describe the content of the paper in more detail. 
In self-testing of no-signalling correlations of quantum type, 
a fundamental role for the identification 
of the uniqueness of the model $M$ (see (\ref{eq_first})) 
is played by \emph{local isometries}; 
in this case the latter have the form $V_A\otimes V_B$, where $V_A : H_A\to K_A$ and 
$V_B : H_B\to K_B$ are isometries between the Alice's and Bob's systems, 
respectively. 
Local isometries allow to push forward correlation models
via operations on each of the systems of Alice and Bob, separately. 
In the commuting operator model, the tensor splitting 
$H_A\otimes H_B$ of the joint quantum system is not available; instead, 
the latter is modeled by a single Hilbert space $H$. 
In Section \ref{s_locisom} we therefore develop a concept of 
a local isometry between quantum commuting bipartite systems, which leads to a natural relation of dominance of one system by another.
Along with establishing some basic properties, related to composing and 
tensoring, we show in Theorem \ref{preorder} that the dominance relation is a 
pre-order.

In Section \ref{s_models}, we propose a notion of a model 
of a correlation for quantum commuting systems. 
The quantum commuting correlations between Alice and Bob are represented 
via states on the commuting tensor product $\cl S_A\otimes_{\rm c} \cl S_B$ of
operator systems $\cl S_A$ and $\cl S_B$, where the former captures 
the local degrees of freedom of Alice, while the latter -- those of Bob. 
A quantum commuting model over the pair $(\cl S_A,\cl S_B)$
is defined as a tuple
$S=(_{\cl A}H_{\cl B},\varphi_A, \varphi_B, \xi)$, 
where $\cl A$ and $\cl B$ are the von Neumann algebras of local observables 
for Alice and Bob, respectively, acting on a Hilbert space $H$,
$\xi\in H$ is a pure state, and
$\varphi_A : \cl S_A\to \cl B(H)$ and $\varphi_B : \cl S_B\to \cl B(H)$
are unital completely positive maps with commuting ranges that are contained in 
the corresponding observable algebras. 
The model $S$ determines a correlation $f_S$, given by 
\begin{equation}\label{eq_initde}
f_S(u) = \langle (\nph_A\cdot\nph_B)(u)\xi,\xi\rangle, \ \ \ 
u\in \cl S_A\otimes_{\rm c} \cl S_B.
\end{equation}
Local isometries now allow to define a \emph{local dilation} relation 
$S\preceq \tilde{S}$ between models $S$ and $\tilde{S}$
(of the same correlation); we show that the latter is a pre-order. 
In fact, it is natural to work in what appears to be a greater generality
of the \emph{approximate local dilation} order $S\preceq_{\rm a} \tilde{S}$,
for which we require the existence of a sequence of local isometries that 
attain $\tilde{S}$ from $S$ in the limit, and 
which continues to imply that $f_S = f_{\tilde{S}}$. 
(Approximate) local dilations between models lead
naturally to the notion of (weak) self-testing of a given 
quantum commuting correlation by requiring 
the existence of a maximal (with respect to the
(approximate) local dilation pre-order) model for the given correlation.
The main result of Section \ref{s_models} is Theorem \ref{th_stisabs}, where we 
show that, under some additional hypotheses, if the correlation $f$ is a 
weak self-test then it is an abstract self-test; 
the result extends \cite[Proposition 4.10]{pszz} to the commuting operator framework. 

In Section \ref{s_abstractselft}, we study further abstract self-tests, 
providing a characterisation thereof in the general case. 
In Theorem \ref{th_rev}, we extend \cite[Theorem 3.5]{pszz}, 
showing that self-testing is equivalent to abstract self-testing 
for correlations arising from tensor products $\nph_A\otimes\nph_B$
in the place of products $\nph_A\cdot\nph_B$ in (\ref{eq_initde}), 
in the case where the observable algebras of Alice and Bob are of type I. 

Section \ref{s_app} is dedicated to applications and examples. 
In their majority, they are concerned with \emph{finitary} quantum systems, 
that is, based on a pair of finite dimensional operator systems. 
In Subsection \ref{ss_QNS}, we show how the proposed framework includes 
self-testing of QNS correlations introduced in \cite{tt-QNS}. 
Subsection \ref{ss_POVMNS} places the self-testing of 
POVM correlations \cite{pszz} in our setup, and explores the 
relation between the latter and QNS self-testing. En route, 
we record some new observations about the relation between 
the universal operator systems for POVM self-testing and
QNS self-testing. Example \ref{CHSH_Cq_Qq} demonstrates that 
self-tests for classical no-signalling correlations cannot be 
automatically lifted to self-tests for quantum no-signalling correlations. 
Subsection \ref{ss_CSHS-qcst} shows that the optimal quantum strategy 
for the CHSH game persists in being a self-test within the class of 
quantum commuting correlations. 

In Subsection \ref{ss_clifford}, we continue the study of self-testing 
for no-signalling correlations of quantum type. 
We single out a class such correlations (which we call Clifford correlations), 
arising from representations of the Clifford relations, and 
show that synchronous Clifford correlations are abstract self-tests. 
We employ the NPA hierarchy \cite{npa}, showing that the 
latter fact leads to self-tests among the correlations with 
constraints on their order-two moments. 

Subsection \ref{ss_probassicont} shows how our framework can host 
self-testing for a pair of contextuality scenarios \cite{acin-etal}. 
We note that the framework of contextuality scenarios is designed to include as a special case
no-signalling correlations; however, no-signalling is usually studied alongside 
contextuality. 
We show that, even in the case where one of the scenarios 
consists of a single outcome, 
this setting may be non-trivial, by casting \cite[Main Theorem]{brv-etal} 
in our language. 

In Subection \ref{s:qg}, we show one yet new special case hosted within
our framework, namely, classical-to-quantum no-signalling correlations. 
We show that perfect strategies for the quantum colouring game of classical 
complete graphs are self-tests. 
Finally, Subsection \ref{s:schur} provides self-testing examples among 
Schur quantum channels.

The relation between abstract self-testing -- requiring a unique extension of a state to a tensor product of C*-covers -- and self-testing -- defined through the existence of a maximal dilation -- is reminiscent of the relation between the unique extension property (UEP) for unital completely positive maps on operator systems and maximality in dilation order \cite{arv}, wherein C*-envelopes play a crucial role. In fact, this similarity was one of our motivations for developing an operator system approach to self-testing. In order to elucidate this similarity, in Section \ref{s_conncstarenv} we show that many of the special 
cases studied in the literature as well as our new self-testing examples fall in the class of states 
that can be factored through the tensor product of C*-envelopes of the 
ground operator systems. We also show that every stochastic operator 
matrix admits a unitary dilation, allowing us to identify the 
C*-envelope of the universal operator system of QNS correlations. We leave a deeper investigation of the connections with UEP to future work.

In Section \ref{s_questions}, we comment on some questions 
arising from our work. 

\medskip

\noindent 
{\bf Notation.} 
We let $\cl X\otimes \cl Y$ be the algebraic tensor product of
vector spaces $\cl X$ and $\cl Y$, unless the latter are Hilbert spaces, 
in which case the notation will be reserved for their Hilbertian tensor product. 
Given Hilbert spaces $H$ and $K$, we let $\cl B(H,K)$ be the 
(Banach) space of all bounded linear operators from $H$ into $K$; 
we set $\cl B(H) = \cl B(H,H)$. 
If $\cl{A}$ is a unital C*-algebra, we write $1_\cl{A}$ for its unit, and let 
$1_H = I_H = 1_{\cl B(H)}$. 
The opposite C*-algebra $\cl{A}^o$ of $\cl{A}$ has the same 
set-theoretic, linear and involutive structure as $\cl{A}$ and, writing its elements as $a^o$, where $a\in \cl{A}$, its multiplication is 
given by letting $a_1^o a_2^o = (a_2 a_1)^o$, $a_1,a_2\in \cl{A}$. 
The maximal (resp. minimal) tesnor product of C*-algebras $\cl{A}$ and 
$\cl{B}$ will be denoted by $\cl{A} \otimes_{\max} \cl{B}$
(resp. $\cl{A} \otimes_{\min} \cl{B}$). If $\cl{A}$ and $\cl{B}$ are von Neumann algebras, 
their spatial weak* closed tensor product will be denoted by 
$\cl{A}\bar\otimes \cl{B}$. 

For a finite set $X$, we let $\{e_x\}_{x\in X}$ denote the canonical 
orthonormal basis of $\bb{C}^X$. We write $M_X$ for the algebra of all 
$X\times X$ complex matrices, and let $\{\epsilon_{x,x'}\}_{x,x'\in X}$
be the canonical matrix unit system in $M_X$. We let ${\rm Tr}$ 
(resp. ${\rm tr}$) be the trace (resp. normalised trace) on $M_X$; we write 
${\rm Tr}_{|X|}$ (resp. ${\rm tr}_{|X|}$) 
to improve clarity as needed.


\section{Local isometries}
\label{s_locisom}

In this section, we define a notion of local isometry 
in the commuting operator framework of quantum mechanics and 
establish some of its properties that will be used in the sequel. 
We first recall the tensor product model of composite quantum systems; 
it serves as 
a motivating example which we aim to generalise. 

Let $H_A$, $K_A$, $H_B$ and $K_B$ be Hilbert spaces; 
they give rise to two (bipartite) 
quantum mechanical systems with observable
algebras
$\cl B(H_A\ten H_B) = \cl B(H_A)\bar{\otimes}\cl B(H_B)$ and 
$\cl B(K_A\ten K_B) = \cl B(K_A)\bar{\otimes}\cl B(K_B)$, respectively. 
An isometry $V : H_A\ten H_B\to K_A\ten K_B$ is called 
\emph{local} if 
$V = V_A\ten V_B$ for some isometries
$V_A : H_A\to K_A$ and $V_B : H_B\to K_B$.
We may view $V_A\ten V_B$ as a composition 
in two ways, 
by introducing the 
intermediate Hilbert spaces $H_A\ten K_B$ and $K_A\ten H_B$ and writing
$$V_A\ten V_B=(V_A\ten I_{K_B})\circ(I_{H_A}\ten V_B)=(I_{K_A}\ten V_B)\circ(V_A\ten {I_{H_B}});$$ 
in other words, we have a commutative diagram
\begin{equation*}
\begin{tikzcd}
H_A\ten H_B\arrow[d, "V_A\ten I_{H_B}"]\arrow[r, "I_{H_A}\ten V_B"] &H_A\ten K_B\arrow[d, "V_A\ten I_{K_B}"]\\
K_A\ten H_B\arrow[r, "I_{K_A}\ten V_B"] &K_A\ten K_B.
\end{tikzcd}
\end{equation*}

Towards formulating a suitable extension 
of the framework described in the previous paragraph,
we fix von Neumann algebras $\cl{A}$ and $\cl{B}$. 
A \textit{bipartite quantum system} over the pair $(\cl{A},\cl{B})$ 
is a pair $(H,\pi)$, where 
$H$ is a Hilbert space and 
$\pi : \cl{A}\ten_{\max} \cl{B}^{o}\to \cl B(H)$ is 
a unital $*$-representation, normal in each of the two variables (equivalently, $\pi$ is a unital representation of the binormal tensor product $\cl{A}\ten_{\rm bin} \cl{B}^o$ \cite{EL}).
We will often write $\pi=\pi_H$ if we want to emphasise the 
underlying Hilbert space.
We let $\pi_H(a) = \pi(a\ten 1)$, $a\in \cl{A}$, and 
$\pi_H(b^o) = \pi(1\otimes b^o)$, $b^o\in \cl{B}^o$.  
We set
$$a\cdot\xi := \pi_{H}(a)\xi \ \mbox{ and }
\ 
\xi\cdot b := \pi_{H}(b^o)\xi, \ \ \ a\in \cl{A}, \ b\in \cl{B}, \ \xi\in H;$$
thus, a bipartite system over $(\cl{A},\cl{B})$ is, equivalently, 
a normal Hilbertian $\cl{A}$-$\cl{B}$-bimodule \cite[\S IX.3]{takesaki2}. 
For clarity, we will use the notation $_\cl{A}H_\cl{B}$ in the place of $H$, 
and, by abuse of notation, 
will refer to $_\cl{A}H_\cl{B}$ as a biparite quantum system.
Heuristically, a bipartite quantum system $_\cl{A}H_\cl{B}$ is 
a joint quantum system of parties
Alice and Bob; the algebra $\cl{A}$ represents Alice's part of the system 
(namely, the observables accessible to Alice), while 
the algebra $\cl{B}$ -- Bob's part of the system 
(namely, the observables accessible to Bob).

\begin{remark}\label{r_intten}
If $_\cl{A}H_\cl{B}$ and $_\cl{C}K_\cl{D}$ are biparite quantum systems then 
the tensor product $H\otimes K$ is a bipartite quantum system 
over the pair $(\cl{A}\bar\otimes \cl{C}, \cl{B}\bar\otimes \cl{D})$ in a canonical fashion. 
We write $H\otimes K = \mbox{}_\cl{A}H_\cl{B} \otimes \mbox{}_\cl{C}K_\cl{D}$. 
\end{remark}

\begin{defn}\label{d:Alocal} 
Let $\cl{A}, \cl{A}_i, \cl{B}, \cl{B}_i$ be von Neumann algebras, $i = 1,2$, 
and $_{\cl{A}_1}H_{\cl{B}}$, $_{\cl{A}_2}K_{\cl{B}}$, $_{\cl{A}}\tilde{H}_{\cl{B}_1}$ and 
$_{\cl{A}}\tilde{K}_{\cl{B}_2}$ be biparite quantum systems.

\begin{itemize}
\item[(i)] 
An operator $T\in\cl B(H,K)$ is called \textit{$\cl{A}_1$-$\cl{A}_2$-local} if
\begin{itemize}
\item[(i')] $T(\xi\cdot b)=(T\xi)\cdot b$, $\xi\in H$, $b\in \cl{B}$, and 
\item[(i'')] $T\pi_{H}(\cl{A}_1)T^*\subseteq\pi_{K}(\cl{A}_2)$.
\end{itemize} 

\item[(ii)] 
An operator $\tilde{T}\in\cl B(\tilde{H},\tilde{K})$ is called \textit{$\cl{B}_1$-$\cl{B}_2$-local} if
\begin{enumerate}
\item[(ii')] $\tilde{T}(a\cdot \tilde{\xi}) = a\cdot (\tilde{T}\tilde{\xi})$, 
$\tilde{\xi}\in \tilde{H}$, $a\in \cl{A}$, and 
\item[(ii'')] $T\pi_{\tilde{H}}(\cl{B}_1^o)T^*\subseteq\pi_{\tilde{K}}(\cl{B}_2^o)$.
\end{enumerate} 
\end{itemize}
\end{defn}

If no confusion arises, $\cl{A}_1$-$\cl{A}_2$-local 
(resp. $\cl{B}_1$-$\cl{B}_2$-local) operators will be simply referred to 
as \emph{$\cl{A}$-local} (resp. \emph{$\cl{B}$-local}). 
We note that the $\cl{B}$-module condition (i') in Definition \ref{d:Alocal} 
is equivalent to the operator
$T : \mbox{}_{\cl{A}_1}H_{\cl{B}} \to \textbf{}_{\cl{A}_2}K_{\cl{B}}$ 
being an intertwiner between the representations $\pi_H$ and $\pi_K$ of $\cl{B}$.
Intuitively, such an intertwiner 
is $\cl{A}$-local if, in addition, the conjugation by $T$ leaves 
Alice's algebras of observables globally invariant. Such module and invariance conditions are common features among notions of local operations in algebraic quantum theory (see e.g. \cite{cckl, cklt, kita,lsww,vw}).

\begin{defn}\label{d:localiso} Let 
$_{\cl{A}_1}\hspace{-0.04cm}H_{\cl{B}_1}$ and 
$_{\cl{A}_2}\hspace{-0.04cm}K_{\cl{B}_2}$ be bipartite quantum systems over $(\cl{A}_i,\cl{B}_i)$, $i=1,2$. 
An operator $T : \mbox{}_{\cl{A}_1}\hspace{-0.04cm}H_{\cl{B}_1}\to \mbox{}_{\cl{A}_2}\hspace{-0.04cm}K_{\cl{B}_2}$ is 
called \textit{local} if there exist
bipartite quantum systems
$_{\cl{A}_1}\hspace{-0.04cm}L_{\cl{B}_2}$ and $_{\cl{A}_2}\hspace{-0.04cm}\widetilde{L}_{\cl{B}_1}$, $\cl{A}_1$-$\cl{A}_2$ local maps $$T_{1,1}: \mbox{}_{\cl{A}_1}\hspace{-0.04cm}H_{\cl{B}_1}\to \mbox{}_{\cl{A}_2}\hspace{-0.04cm}\widetilde{L}_{\cl{B}_1} \ 
\mbox{ and } \ 
T_{2,2}: \mbox{}_{\cl{A}_1}\hspace{-0.04cm}L_{\cl{B}_2}\to \mbox{}_{\cl{A}_2}\hspace{-0.04cm}K_{\cl{B}_2},$$ and $\cl{B}_1$-$\cl{B}_2$ local maps $$T_{1,2}: \mbox{}_{\cl{A}_1}\hspace{-0.04cm}H_{\cl{B}_1}\to \mbox{}_{\cl{A}_1}\hspace{-0.04cm}L_{\cl{B}_2} \ \mbox{ and } \ T_{2,1}= \mbox{}_{\cl{A}_2}\hspace{-0.04cm}\widetilde{L}_{\cl{B}_1}\to \mbox{}_{\cl{A}_2}\hspace{-0.04cm}K_{\cl{B}_2},$$
for which the diagram 

\begin{equation}\label{eq_diag!}
\begin{tikzcd}
_{\cl{A}_1}\hspace{-0.04cm}H_{\cl{B}_1}\arrow[d, "T_{1,1}"] \arrow[r, "T_{1,2}"] &_{\cl{A}_1}\hspace{-0.04cm}L_{\cl{B}_2}\arrow[d, "T_{2,2}"]\\
_{\cl{A}_2}\hspace{-0.04cm}\widetilde{L}_{\cl{B}_1} \arrow[r, "T_{2,1}"] &_{\cl{A}_2}\hspace{-0.04cm}K_{\cl{B}_2}
\end{tikzcd}
\end{equation}
commutes, and 
$T=T_{2,2}\circ T_{1,2} \ (=T_{2,1}\circ T_{1,1})$. When all the local maps $T_{i,j}$ can be chosen to be 
isometric, we say that $T$ is a \textit{local isometry}. In this case, we write $_{\cl{A}_1}\hspace{-0.04cm}H_{\cl{B}_1}\leq\, _{\cl{A}_2}\hspace{-0.04cm}K_{\cl{B}_2}$.
\end{defn}

Henceforth, when discussing local isometries we will stick with the convention that the ``veritcal maps'' $T_{1,1}$ and $T_{2,2}$ are $\cl{A}$-local and the ``horizontal maps'' are $\cl{B}$-local.

\begin{remark}\label{r_preciselo}
\rm 
{\bf (i) } 
Let $\cl{A}$ and $\cl{B}$ be commuting von Neumann algebras acting on a Hilbert space $H$. Then we can view $H=_\cl{A}\hspace{-0.04cm}H_{\cl{B}^o}$ as an $\cl{A}$-$\cl{B}^{o}$-bimodule in the canonical fashion. Let $u\in \cl{A}$ and $v\in \cl{B}$ be isometries. Then the product 
$uv: \, _\cl{A}H_{\cl{B}^o}\to \,  _\cl{A}H_{\cl{B}^o}$ is a local isometry in the sense of Definition \ref{d:localiso}, with canonical factorisations.

\smallskip

{\bf (ii) } 
In the context of item (i), assume that 
$_{\cl{A}}\hspace{-0.04cm}H_{\cl{B}}\leq\, 
_{\tilde{\cl{A}}}\hspace{-0.04cm}\tilde{H}_{\tilde{\cl{B}}}$, 
implemented by isometries $T_{i,j}$, $i,j = 1,2$, as in (\ref{eq_diag!}). 
Writing $_{\cl{A}}\hspace{-0.04cm}L_{\tilde{\cl{B}}}$ for the 
upper right corner of (\ref{eq_diag!}), we have that there 
exists a normal *-representation $\sigma : \cl A \to \cl B(L)$ 
and a normal *-representation $\rho : \sigma(\cl A)\to \cl B(\tilde{H})$, 
such that ${\rm ran}(\rho) \subseteq \pi_{\tilde{H}}(\tilde{\cl A})$, and 
\begin{equation}\label{eq_concretefol}
T_{1,2} a = \sigma(a) T_{1,2} \ \mbox{ and } \ 
T_{2,2} \sigma(a) = \rho(\sigma(a)) T_{2,2}, \ \ \ a\in \cl A.
\end{equation}
Indeed, the existence of $\sigma$ and the first relation in (\ref{eq_concretefol}) 
are a restatement of condition (ii') in Definition \ref{d:Alocal},
while $\rho$ can be defined by letting 
$\rho(\sigma(a)) = T_{2,2} \sigma(a) T_{2,2}^*$, $a\in \cl A$. 
\end{remark}

\begin{example}\label{ex_findimset}
For a Hilbert space $H$, let $\bar{H}$ be its dual Banach space, 
and let $\partial : H\to \bar{H}$ be the conjugate-linear isometry, 
given by $\partial(\xi)(\eta) = \langle\eta,\xi\rangle$, $\xi,\eta \in H$. 
We write $\bar{\xi} = \partial(\xi)$, $\xi\in H$. 
If $K$ is a(nother) Hilbert space, let 
$\cl S_2(\bar{K},H)$ be the Hilbert 
space of all Hilbert-Schmidt operators from $\bar{K}$ into $H$, and 
$\theta: H\otimes K\to \cl S_2(\bar{K},H)$ be the unitary operator, 
given by 
$\theta(\xi\otimes\eta)(\bar{\zeta}) = \langle \eta,\zeta\rangle \xi$. 

Let $H_A$, $K_A$, $H_B$ and $K_B$ be Hilbert spaces.
We note the canonical identification $\cl B(H_B)^o = \cl B(\bar{H}_B)$,
where an element $b^o\in \cl B(H_B)^o$ is identified 
with the dual operator $\bar{b} : \bar{H}_B\to \bar{H}_B$
of $b\in \cl B(H_B)$.
Consider the normal *-representation 
$\pi_{H_A\otimes H_B} : 
\cl B(H_A)\otimes_{\max}\cl B(H_B)^o \to \cl B(H_A\otimes H_B)$, given by 
$$\pi_H(a)(\xi) = \theta^{-1}(a\theta(\xi)), \ \ \ 
\pi_H(b^o)(\xi) = \theta^{-1}(\theta(\xi)\bar{b}),$$ 
where $a\in \cl B(H_A), b^o\in \cl B(H_B)^o$ and $\xi\in H_A\otimes H_B$.
We write $b^{\rm t}$ for the operator on $H_B$ for which 
$\pi_H(b^o) = I_{H_A}\otimes b^{\rm t}$, $b\in \cl B(H_B)$. 
The *-representation $\pi_{H_A\otimes H_B}$ 
gives rise to a bipartite quantum system
$_{\cl B(H_A)}(H_A\otimes H_B)_{\cl B(H_B)^{o}}$; we similarly obtain  
a bipartite quantum system 
$_{\cl B(K_A)}(K_A\otimes K_B)_{\cl B(K_B)^{o}}$. 
Quantum systems of the latter form will be called 
\emph{quantum spacial systems}.

Suppose that $V_A:H_A\to K_A$ and $V_B:H_B\to K_B$ are isometries. 
Let $L = H_A\ten K_B$ and $\wt{L} = K_A\otimes H_B$, 
and consider the bipartite quantum systems 
$_{\cl B(H_A)}{L}_{B(K_B)^{o}}$ and 
$_{\cl B(K_A)}\wt{L}_{\cl B(H_B)^{o}}$, 
arising as described in the previous paragraph. 
Setting 
$T_{1,1} = V_A\otimes I_{H_B}$, 
$T_{1,2} = I_{H_A}\ten V_B$, $T_{2,1} = I_{K_A}\otimes V_B$ and 
$T_{2,2} = V_A\otimes I_{K_B}$, we 
obtain a commutative diagram 
\begin{equation}\label{eq_diag1}
\begin{tikzcd}
_{\cl B(H_A)}(H_A\otimes H_B)_{\cl B(H_B)^{o}} \arrow[d, "T_{1,1}"] \arrow[r, "T_{1,2}"] 
& _{\cl B(H_A)}L_{B(K_B)^{o}}\arrow[d, "T_{2,2}"]\\
_{\cl B(K_A)}\wt{L}_{\cl B(H_B)^{o}} \arrow[r, "T_{2,1}"] & 
_{\cl B(K_A)}(K_A\otimes K_B)_{\cl B(K_B)^{o}}, 
\end{tikzcd}
\end{equation}
for which $V_A\otimes V_B = T_{2,2}T_{1,2} = T_{2,1}T_{1,1}$. 
We will call the local isometries of the latter form \emph{split}. 
\end{example}

\begin{proposition}\label{p_split}
Let $H_A$, $K_A$, $H_B$ and $K_B$ be Hilbert spaces.
Every local isometry from $_{\cl B(H_A)}(H_A\otimes H_B)_{\cl B(H_B)^{o}}$ to $_{\cl B(K_A)}(K_A\otimes K_B)_{\cl B(K_B)^{o}}$
is split. 
\end{proposition}

\begin{proof}
Suppose that 
$V : H_A\ten H_B\to K_A\ten K_B$ is a local isometry in the sense of Definition \ref{d:localiso}, and let 
$_{\cl B(H_A)}L_{B(K_B)^{o}}$ and $_{\cl B(K_A)}\wt{L}_{\cl B(H_B)^{o}}$
be bipartite quantum systems, while 
$T_{1,1}, T_{1,2}, T_{2,1}$ and $T_{2,2}$ are local isometries
in the diagram (\ref{eq_diag1}), 
such that $V = T_{2,2}T_{1,2} = T_{2,1}T_{1,1}$.
Up to unitary conjugations, we may assume that 
$L = H_A\ten K_B\ten L_0$ and $\wt{L} = K_A\ten H_B\ten\wt{L}_0$ 
for some Hilbert spaces $L_0$ and $\wt{L}_0$, 
so that 
$$\pi_L(a) = a\ten I_{K_B}\ten I_{L_0} \ \mbox{ and } \ 
\pi_L(b^o) = I_{H_A} \ten b^{\rm t} \ten I_{L_0}^{\rm t},$$
for $a\in \cl B(H_A)$ and $b\in \cl B(K_B)$, 
and 
$$\pi_{\wt{L}}(a) = a\ten I_{H_B}\ten I_{\wt{L}_0} \ \mbox{ and } \ 
\pi_{\wt{L}}(b^o) = I_{K_A} \ten b^{\rm t} \ten I_{\wt{L}_0}^{\rm t},$$
for $a\in \cl B(K_A)$ and $b\in \cl B(H_B)$.
Condition (ii') in Definition \ref{d:Alocal} implies that, for every
$h_1,h_2\in H_A$, $\xi\in H_B$, $\eta\in K_B\ten L_0$ and 
$a\in \cl B(H_A)$, we have 
$$\langle T_{1,2}(ah_1\otimes\xi), h_2\otimes\eta\rangle = 
\langle T_{1,2}(h_1\otimes\xi)), a^*h_2\otimes\eta\rangle.
$$
It follows that 
$$(\id_{H_A}\ten\omega_{\xi,\eta})(T_{1,2}) \in \cl B(H_A)' = \bb{C} I_{H_A}.$$ 
Letting $\lambda_{\xi,\eta}\in \bb{C}$ with 
$(\id_{H_A}\ten\omega_{\xi,\eta})(T_{1,2}) = \lambda_{\xi,\eta} I_{H_A}$, 
it is straightforward to verify that the map
$(\xi,\eta)\to \lambda_{\xi,\eta}$ is conjugate-linear and 
bounded, therefore yielding a bounded operator 
$W_B : H_B\to K_B\ten L_0$, such that 
$$\lambda_{\xi,\eta} = \langle W_B\xi,\eta\rangle, 
\ \ \ \xi\in H_B, \eta\in K_B\ten L_0.$$
Thus $T_{1,2} = I_{H_A}\ten W_B$ and, since $T_{1,2}$ is an 
isometry, so is $W_B$. 

Condition (ii'') in Definition \ref{d:Alocal} now reads
$$(I_{H_A}\ten W_B)(I_{H_A}\otimes \cl B(H_B)^o) (I_{H_A}\ten W_B^*) 
\subseteq I_{H_A}\otimes \cl B(K_B)^o\otimes I_{L_0},$$
that is, 
\begin{equation}\label{eq_ranko}
W_B \cl B(H_B)^o W_B^*\subseteq \cl B(K_B)^o\otimes I_{L_0}.
\end{equation}
If $b^o\in \cl B(H_B)^o$ is a rank one operator, then 
the left hand side in (\ref{eq_ranko}) is a rank one operator, 
implying that $\dim L_0 = 1$. 
Similarly, $\dim \wt{L}_0 = 1$, and there exist 
isometries $\wt{W_B} : H_B \to K_B$, 
$W_A : H_A\to K_A$ and 
$\wt{W_A} : H_A\to K_A$ such that 
$$T_{2,1} = I_{K_A}\ten \wt{W_B}, \ \ T_{1,1} = W_A\ten I_{H_B}, \ \ \textnormal{and} \ \ T_{2,2} = \wt{W_A}\ten I_{K_B}.$$
It follows that 
$V = \wt{W_A}\ten W_B = W_A\ten \wt{W_B}$, implying 
that $W_A = \wt{W_A}$ and $W_B = \wt{W_B}$.
The proof is complete. 
\end{proof}

\begin{example} 
Let $_\cl{A}H_\cl{C}$, $_\cl{A}\wt{H}_\cl{C}$, $_\cl{C}K_\cl{B}$ and $_\cl{C}\wt{K}_\cl{B}$ be bipartite quantum system, for some
von Neumann algebras $\cl{A}$, $\cl{B}$ and $\cl{C}$. Suppose that  
$V : \mbox{}_\cl{A}H_\cl{C}\to \mbox{}_\cl{A}\wt{H}_\cl{C}$ is an 
$\cl{A}$-local isometry and $W : \mbox{}_\cl{C}K_\cl{B}\to \mbox{}_\cl{C}\wt{K}_\cl{B}$ is a 
$\cl{B}$-local isometry. 
Fix a faithful normal semi-finite weight $\vphi$ on $\cl{C}$ (or a faithful normal state if $\cl{C}$ is $\sigma$-finite), 
and let $\boxtimes$ denote Connes' fusion product of bimodules relative to $\vphi$ \cite[Definition IX.3.16]{takesaki2}.
Then the commuting diagram

\begin{equation*}
\begin{tikzcd}
\mbox{}_\cl{A} H_\cl{C}\boxtimes \mbox{}_\cl{C}K_\cl{B}\arrow[d, "V\ten I_K"] \arrow[r, " I_H\ten W"] 
& \mbox{}_\cl{A} H_\cl{C}\boxtimes \mbox{}_\cl{C}\wt{K}_\cl{B}\arrow[d, "V\ten I_{\wt{K}}"]\\
 \mbox{}_\cl{A}\wt{H}_\cl{C}\boxtimes \mbox{}_\cl{C}K_\cl{B} \arrow[r, " I_{\wt{H}}\ten W"] 
& \mbox{}_\cl{A}\wt{H}_\cl{C}\boxtimes \mbox{}_\cl{C}\wt{K}_\cl{B}
\end{tikzcd}
\end{equation*}
gives rise to a local isometry 
$V\boxtimes W: \ _\cl{A}H_\cl{C}\boxtimes \mbox{}_\cl{C}K_\cl{B}\to _\cl{A}\wt{H}_\cl{C}\boxtimes \mbox{}_\cl{C}\wt{K}_\cl{B}$. 
\end{example}

\begin{remark}\label{r_compo}
(i) It is straightforward to see that, if 
$\cl{A}_i, \cl{B}$ (resp. $\cl{A},\cl{B}_i$) are von Neumann algebras,  
$_{\cl{A}_i}H^{(i)}_{\cl{B}}$ (resp. $_{\cl{A}}K^{(i)}_{\cl{B}_i}$) 
are biparite quantum systems, $i = 1,2,3$, and
$S_1\in\cl B(H^{(1)},H^{(2)})$ and $S_2\in\cl B(H^{(2)},H^{(3)})$
(resp. 
$T_1\in\cl B(K^{(1)},K^{(2)})$ and $T_2\in\cl B(K^{(2)},K^{(3)})$)
are $\cl{A}$-local (resp. $\cl{B}$-local) operators, then the operator $S_2S_1$
(resp. $T_2T_1$) is $\cl{A}$-local (resp. $\cl{B}$-local). 

(ii) 
Let $_{\cl{A}_1}\hspace{-0.05cm}H_{\cl{B}_1}$ and 
$_{\cl{A}_2}\hspace{-0.04cm}K_{\cl{B}_2}$ 
(resp. $_{\tilde{\cl{A}}_1}\hspace{-0.05cm}\tilde{H}_{\tilde{\cl{B}}_1}$ and 
$_{\tilde{\cl{A}}_2}\hspace{-0.04cm}\tilde{K}_{\tilde{\cl{B}}_2}$) 
be bipartite quantum systems over $(\cl{A}_i,\cl{B}_i)$ 
(resp. $(\tilde{\cl{A}}_i,\tilde{\cl{B}}_i)$), $i=1,2$. 
If $T : \mbox{}_{\cl{A}_1}\hspace{-0.04cm}H_{\cl{B}_1}\to \mbox{}_{\cl{A}_2}\hspace{-0.04cm}K_{\cl{B}_2}$ 
and $\tilde{T} : \mbox{}_{\tilde{\cl{A}}_1}\hspace{-0.04cm}\tilde{H}_{\tilde{\cl{B}}_1}\to \mbox{}_{\tilde{\cl{A}}_2}\hspace{-0.04cm}\tilde{K}_{\tilde{\cl{B}}_2}$
are local maps (resp. isometries) then 
$T\otimes \tilde{T} : 
\mbox{}_{\cl{A}_1}\hspace{-0.04cm}H_{\cl{B}_1} \otimes \mbox{}_{\tilde{\cl{A}}_1}\hspace{-0.04cm}\tilde{H}_{\tilde{\cl{B}}_1}
\to 
\mbox{}_{\cl{A}_2}\hspace{-0.04cm}K_{\cl{B}_2} \otimes \mbox{}_{\tilde{\cl{A}}_2}\hspace{-0.04cm}\tilde{K}_{\tilde{\cl{B}}_2}$
is a local map (resp. isometry).
Indeed, suppose that 
$T=T_{2,2}\circ T_{1,2}=T_{2,1}\circ T_{1,1}$ and 
$\tilde{T} = \tilde{T}_{2,2}\circ \tilde{T}_{1,2} 
= \tilde{T}_{2,1}\circ \tilde{T}_{1,1}$ as in diagram (\ref{eq_diag!}). 
Setting $\hat{T}_{i,j} = T_{i,j}\otimes \tilde{T}_{i,j}$, we witness the 
locality of $T\otimes \tilde{T}$ through the identity 
$$T\otimes \tilde{T} = \hat{T}_{2,2}\circ \hat{T}_{1,2}
 = \hat{T}_{2,1}\circ \hat{T}_{1,1}.$$
\end{remark}

We now extend Remark \ref{r_compo} (i) by showing that the relation $\leq$ is transitive.

\begin{theorem}\label{preorder} 
Let $\cl{A}_i$ and $\cl{B}_i$ be von Neumann algebras, $i=1,2,3$. 
If $_{\cl{A}_1}\hspace{-0.04cm}H_{\cl{B}_1}\leq\, _{\cl{A}_2}\hspace{-0.04cm}K_{\cl{B}_2}$ and $_{\cl{A}_2}\hspace{-0.04cm}K_{\cl{B}_2}\leq\, _{\cl{A}_3}\hspace{-0.04cm}L_{\cl{B}_3}$, then $_{\cl{A}_1}\hspace{-0.04cm}H_{\cl{B}_1}\leq\, _{\cl{A}_3}\hspace{-0.04cm}L_{\cl{B}_3}$. 
\end{theorem}

\begin{proof} We need to complete the diagram
\begin{equation*}
\begin{tikzcd}
_{\cl{A}_1}\hspace{-0.04cm}H_{\cl{B}_1}\arrow[d, "V_{1,1}"] \arrow[r, "V_{1,2}"] &_{\cl{A}_1}\hspace{-0.04cm}\wt{H}_{\cl{B}_2}\arrow[d, "V_{2,2}"] & \\
_{\cl{A}_2}\hspace{-0.04cm}\wt{\widetilde{H}}_{\cl{B}_1} \arrow[r, "V_{2,1}"] &_{\cl{A}_2}\hspace{-0.04cm}K_{\cl{B}_2}\arrow[r, "V_{2,3}"]\arrow[d, "\wt{V_{2,2}}"] & _{\cl{A}_2}\hspace{-0.04cm}\wt{K}_{\cl{B}_3}\arrow[d, "V_{3,3}"]\\
& _{\cl{A}_3}\hspace{-0.04cm}\wt{\wt{K}}_{\cl{B}_2}\arrow[r, "V_{3,2}"] & _{\cl{A}_3}\hspace{-0.04cm}L_{\cl{B}_3}
\end{tikzcd}
\end{equation*}
to a commuting diagram of the form
\begin{equation}\label{eq_diagdoub}
\begin{tikzcd}
_{\cl{A}_1}\hspace{-0.04cm}H_{\cl{B}_1}\arrow[d, "V_{1,1}"] \arrow[r, "V_{1,2}"] &_{\cl{A}_1}\hspace{-0.04cm}\wt{H}_{\cl{B}_2}\arrow[d, "V_{2,2}"]\arrow[r, dashed, "W_{1,3}"] & _{\cl{A}_1}\hspace{-0.04cm}\wt{L}_{\cl{B}_3}\arrow[d, dashed, "W_{2,3}"] \\
_{\cl{A}_2}\hspace{-0.04cm}\wt{\widetilde{H}}_{\cl{B}_1}\arrow[d, dashed, "W_{3,1}"] \arrow[r, "V_{2,1}"] &_{\cl{A}_2}\hspace{-0.04cm}K_{\cl{B}_2}\arrow[r, "V_{2,3}"]\arrow[d, "\wt{V_{2,2}}"] & _{\cl{A}_2}\hspace{-0.04cm}\wt{K}_{\cl{B}_3}\arrow[d, "V_{3,3}"]\\
_{\cl{A}_3}\hspace{-0.04cm}\wt{\wt{L}}_{\cl{B}_1}\arrow[r, dashed, "W_{3,2}"] & _{\cl{A}_3}\hspace{-0.04cm}\wt{\wt{K}}_{\cl{B}_2}\arrow[r, "V_{3,2}"] & _{\cl{A}_3}\hspace{-0.04cm}L_{\cl{B}_3},
\end{tikzcd}
\end{equation}
for some quantum systems $_{\cl{A}_1}\wt{L}_{\cl{B}_3}$ and $_{\cl{A}_3}\wt{\wt{L}}_{\cl{B}_1}$, 
some $\cl{B}$-local isometries $W_{1,3}$ and $W_{3,2}$, and 
some $\cl{A}$-local isometries $W_{2,3}$ and $W_{3,1}$.

Let $f_2\in \cl{A}_2$ be a central projection, such that 
$\ker \pi_K|_{\cl{A}_2} = f_2^{\perp} \cl{A}_2$. We have that 
$\pi_K|_{f_2 \cl{A}_2} : f_2 \cl{A}_2\to \pi_K(\cl{A}_2)$ is a normal unital *-isomorphism; 
in the rest of the proof, we denote by $\pi_K^{-1}$ its inverse, from 
$\pi_K(\cl{A}_2)$ onto $f_2\cl{A}_2$. 
By the $\cl{A}$-locality of $V_{2,2}$, there exists 
a (unique) projection $e_2\in f_2\cl{A}_2$ such that $p_{2} := V_{2,2}V_{2,2}^* = \pi_{K}(e_2)$. 
Let $p_\cl{A} = \pi_{\wt{K}}(e_2)\in\cl B(\wt{K})$, 
and set $\wt{L}:=p_\cl{A}\wt{K}$. 

Recalling that 
$\pi_K^{-1}(V_{2,2}\pi_{\wt{H}}(a_1)V_{2,2}^*)\in f_2\cl{A}_2$, 
define $\pi^\cl{A}_{\wt{L}}:\cl{A}_1\to\cl B(\wt{L})$ by
$$\pi^\cl{A}_{\wt{L}}(a_1)=\pi_{\wt{K}}(\pi_K^{-1}(V_{2,2}\pi_{\wt{H}}(a_1)V_{2,2}^*)), \ \ \ a_1\in \cl{A}_1.$$
Note that the $\cl{A}$-locality of $V_{2,2}$ implies $V_{2,2}\pi_{\wt{H}}(a_1)V_{2,2}^*\in\pi_{K}(\cl{A}_2)$, so by 
the fact that $V_{2,2}$ is an isometry, 
$\pi^\cl{A}_{\wt{L}}$ is a well-defined normal $*$-homomorphism. It is also unital as $\pi^\cl{A}_{\wt{L}}(1_{\cl{A}_1})=p_\cl{A}=1_{\cl B(p_\cl{A}\wt{K})}$. 

Next, define $\pi^\cl{B}_{\wt{L}} : \cl{B}_3^o\to\cl B(\wt{L})$ by 
letting 
$$\pi^\cl{B}_{\wt{L}}(b_3^o) = p_\cl{A}\pi_{\wt{K}}(b_3^o)p_\cl{A}, \ \ \ b^o\in \cl{B}_3^o.$$
Since $\pi_{\wt{K}}(\cl{A}_2)\subseteq\pi_{\wt{K}}(\cl{B}_3^o)'$, we have that 
$p_\cl{A}\in \pi_{\wt{K}}(\cl{B}_3^o)'$, and hence the map 
$\pi^\cl{B}_{\wt{L}}$ is a normal unital $*$-homomorphism. The inclusion 
$\pi_{\wt{K}}(\cl{A}_2)\subseteq\pi_{\wt{K}}(\cl{B}_3^o)'$ implies that $\pi^\cl{A}_{\wt{L}}$ and $\pi^\cl{B}_{\wt{L}}$ have commuting ranges, and therefore define a unital, 
separately weak* continuous, $*$-homomorphism 
$$\pi_{\wt{L}} := \pi^\cl{A}_{\wt{L}}\times \pi^\cl{B}_{\wt{L}}
: \cl{A}_1\ten_{\max}\cl{B}_3^o \to\cl B(\wt{L}),$$
turning $\wt{L}$ into an $\cl{A}_1$-$\cl{B}_3$-bimodule $_{\cl{A}_1}\wt{L}_{\cl{B}_3}$. 

Let $W_{1,3} = V_{2,3}V_{2,2}$.
Since
$$V_{2,3}V_{2,2}=V_{2,3}\pi_{K}(e_2)V_{2,2}=\pi_{\wt{K}}(e_2)V_{2,3}V_{2,2}=p_\cl{A}V_{2,3}V_{2,2},$$
we have that 
$W_{1,3}: \, _{\cl{A}_1}H_{\cl{B}_2}\to \, _{\cl{A}_1}\wt{L}_{\cl{B}_3}$ is a well-defined isometry. 
Given $a_1\in \cl{A}_1$, 
let $a_2$ be the unique element of $f_2\cl{A}_2$ such that 
$V_{2,2}\pi_{\wt{H}}(a_1)V_{2,2}^*=\pi_{K}(a_2)$. We have 
\begin{align*}
W_{1,3}\pi_{\wt{H}}(a_1)
&=
V_{2,3}V_{2,2}\pi_{\wt{H}}(a_1)
=V_{2,3}V_{2,2}\pi_{\wt{H}}(a_1)V_{2,2}^*V_{2,2}\\
&=V_{2,3}\pi_{K}(a_2)V_{2,2}
=\pi_{\wt{K}}(a_2)V_{2,3}V_{2,2}\\
&=\pi_{\wt{K}}(\pi_K^{-1}(V_{2,2}\pi_{\wt{H}}(a_1)V_{2,2}^*))V_{2,3}V_{2,2}
=\pi_{\wt{L}}(a_1)W_{1,3}.
\end{align*}
Given $b_2\in \cl{B}_2$, we have $V_{2,3}\pi_{K}(b_2^o)V_{2,3}^*=\pi_{\wt{K}}(b_3^o)$ for some $b_3\in \cl{B}_3$. Let $e_3\in \cl{B}_3$ be a projection for which $V_{2,3}V_{2,3}^*=\pi_{\wt{K}}(e_3^o)$. Then
\begin{align*}
W_{1,3}\pi_{\wt{H}}(b_2^o)W_{1,3}^*
&=V_{2,3}V_{2,2}\pi_{\wt{H}}(b_2^o)V_{2,2}^*V_{2,3}^*
=V_{2,3}\pi_{K}(b_2^o)V_{2,2}V_{2,2}^*V_{2,3}^*\\
&=V_{2,3}\pi_{K}(b_2^o)\pi_{K}(e_2)V_{2,3}^*
=V_{2,3}\pi_{K}(b_2^o)V_{2,3}^*V_{2,3}\pi_{K}(e_2)V_{2,3}^*\\
&=\pi_{\wt{K}}(b_3^o)V_{2,3}\pi_{K}(e_2)V_{2,3}^*
=\pi_{\wt{K}}(b_3^o)p_\cl{A}V_{2,3}V_{2,3}^*\\
&=\pi_{\wt{K}}(b_3^o)p_\cl{A}\pi_{\wt{K}}(e_3^o)
=\pi_{\wt{K}}(b_3^oe_3^o)p_\cl{A}
=\pi_{\wt{L}}(b_3^oe_3^o).
\end{align*}
Hence, $W_{1,3}$ is a $\cl{B}$-local map.
By Remark \ref{r_compo}, the composition $W_{1,3} V_{1,2}$ is a 
$\cl{B}$-local map. 

Let $W_{2,3} : \wt{L} \to \wt{K}$ be the inclusion map. We have that 
$W_{2,3}$ is $\cl{A}$-local.
In fact, condition (i') from Definition \ref{d:Alocal} is tautological
because of the definition of $W_{2,3}$, while, 
if $a_1\in \cl{A}_1$, there is unique $a_2\in f_2\cl{A}_2$ such that
\begin{align*}
W_{2,3}\pi_{\wt{L}}(a_1)W_{2,3}^*
&=\pi_{\wt{K}}(\pi_K^{-1}(V_{2,2}\pi_{\wt{H}}(a_1)V_{2,2}^*))p_\cl{A}\\
&=\pi_{\wt{K}}(a_2)\pi_{\wt{K}}(e_2)
=\pi_{\wt{K}}(a_2e_2),
\end{align*}
implying $W_{2,3}\pi_{\wt{L}}(a_1)W_{2,3}^* \in \pi_{\wt{K}}(\cl{A}_2)$,
and thus condition (i'') from Definition \ref{d:Alocal}. 

We similarly define a bipartite quantum system 
$\mbox{}_{\cl{A}_3}\wt{\wt{{L}}}_{\cl{B}_1}$, and isometries 
$W_{3,1} : \wt{\wt{{H}}}\to \wt{\wt{{L}}}$ and 
$W_{3,2} : \wt{\wt{{L}}}\to \wt{\wt{{K}}}$. Using 
symmetric arguments, we have that the operator
$W_{3,1}$ is $\cl{A}$-local, while the operator 
$W_{3,2}$ is $\cl{B}$-local. 
By Remark \ref{r_compo}, 
the composition $V_{3,2}W_{3,2}$ is $\cl{B}$-local, while 
the compositions $W_{3,1}V_{1,1}$ and $V_{3,3}W_{2,3}$ are $\cl{A}$-local. 

The relation $W_{2,3}W_{1,3} = V_{2,3}V_{2,2}$ follows from the definition 
of the operators $W_{1,3}$ and $W_{2,3}$; similarly, 
$W_{3,2}W_{3,1} = \wt{V}_{2,2}V_{2,1}$. 
The commutativity of the diagram (\ref{eq_diagdoub})
is now immediate, and the proof is complete.
\end{proof}


\section{Models and self-tests}\label{s_models}

Our generalisation of self-testing is based on the 
notion of an operator system; 
we recall some basic facts and concepts, and refer the 
reader to \cite{Pa} for details. 
If $H$ is a Hilbert space and 
$\cl S\subseteq\cl B(H)$ is a linear subspace, 
the space $M_n(\cl S)$ of all $n$ by $n$ matrices with entries in $\cl S$ 
can be viewed as a subspace of $\cl B(H^n)$ 
after identifying $M_n(\cl B(H))$ with $\cl B(H^n)$. 
If $\cl S \subseteq \cl B(H)$ and $\cl T \subseteq \cl B(K)$ 
(where $K$ is a(nother) Hilbert space) are subspaces and 
$\phi : \cl S\to \cl T$ is a linear map, 
we let $\phi^{(n)} : M_n(\cl S)\to M_n(\cl T)$ be the (linear) map, given by 
$\phi^{(n)}((a_{i,j})_{i,j}) = (\phi(a_{i,j}))_{i,j}$.
An \textit{operator system} is a subspace 
$\cl S \subseteq \cl B(H)$,
such that $I_{\cl H}\in \cl S$ and $s\in \cl S \Rightarrow s^*\in \cl S$. 
Every operator system $\cl S$ is an {\it abstract operator system} in the sense that 
(a) $\cl S$ is a linear *-space; (b) the real vector 
space $M_n(\cl S)_h$ of all hermitian elements in the 
*-space $M_n(\cl S)$
is equipped with a proper cone $M_n(\cl S)^+$;
(c) 
$T^* M_n(\cl S)^+ T\subseteq M_m(\cl S)^+$ for all $n,m\in\bb{N}$ and all $T\in M_{n,m}$, and 
(d) the cone family $(M_n(\cl S)^+)_{n\in \bb{N}}$ 
admits an Archimedean matrix order unit. 
If $\cl S$ and $\cl T$ are abstract operator systems, 
a linear map $\phi : \cl S\to \cl T$
is called \textit{positive} if $\phi (\cl S^+)\subseteq \cl T^+$,
\textit{completely positive} if $\phi^{(n)}$ is positive for every $n\in \bb{N}$, 
\textit{unital} if $\phi(I) = I$,
and a \textit{complete order isomorphism} if $\phi$ is completely positive, bijective, and 
its inverse $\phi^{-1}$ is completely positive. 
We note that every unital completely positive map is 
automatically contractive. 
By virtue of the Choi-Effros Theorem \cite[Theorem 13.1]{Pa}, 
every abstract operator system is completely order isomorphic to an
operator system. 
A \textit{state} of an operator system 
$\cl S$ is a positive unital linear functional; 
we denote by $S(\cl S)$ the (convex) set of all states of $\cl S$. 

Let $\cl S$ be an operator system. Recall that 
a pair $(C_u^*(\cl S), \iota)$ is called a universal cover of $\cl S$, if $C_u^*(\cl S)$ is a unital $C^*$-algebra, $\iota:\cl S\to C_u^*(\cl S)$ is a unital complete order embedding
such that $\iota(\cl S)$ generates $C_u^*(\cl S)$
and, whenever $H$ is a Hilbert space and $\phi:\cl S\to\cl B(H)$ is a unital completely positive map, there exists a $*$-representation $\pi_\phi:C_u^*(\cl S)\to \cl B(H)$ such that $\pi_\phi\circ\iota=\phi$
(see e.g. \cite{KPTT11}). It is clear that
the universal cover is unique up to a canonical $*$-isomorphism.

We fix throughout this section operator systems $\cl S_A$ and $\cl S_B$.
Their commuting tensor product $\cl S_A\otimes_{\rm c} \cl S_B$ 
is the operator system with underlying vector space the algebraic tensor product
$\cl S_A\otimes \cl S_B$, and matrix order structure inherited from the 
inclusion $\cl S_A\otimes_{\rm c} \cl S_B\subseteq C_u^*(\cl S_A)\otimes_{\rm max} C_u^*(\cl S_B)$; thus, $\cl S_A\otimes_{\rm c} \cl S_B$ 
sits completely order isomorphically in
$C_u^*(\cl S_A)\otimes_{\rm max} C_u^*(\cl S_B)$
(we note that we are using an equivalent definition of the 
commuting tensor product to the original one, 
see \cite[Theorem 6.4]{KPTT11}). 
Given a Hilbert space $H$ and 
unital completely positive maps
$\varphi_A : \cl S_A\to \cl B(H)$ and $\varphi_B : \cl S_B\to \cl B(H)$
with commuting ranges, 
there exists a (unique) unital completely positive map 
$\varphi_A\cdot \varphi_B : \cl S_A\otimes_{\rm c} \cl S_B \to \cl B(H)$, 
such that 
$$(\varphi_A\cdot \varphi_B)(s\otimes t) = 
\varphi_A(s) \varphi_B(t), \ \ \ s\in \cl S_A, t\in \cl S_B$$
(see \cite[Corollary 6.5]{KPTT11}).

A \emph{quantum commuting model} over the pair $(\cl S_A,\cl S_B)$
is a tuple
$$S=(_{\cl A}H_{\cl B},\varphi_A, \varphi_B, \xi),$$
where $H$ is a Hilbert space, $\xi\in H$ is a unit vector, 
$\varphi_A : \cl S_A\to \cl B(H)$ and $\varphi_B : \cl S_B\to \cl B(H)$
are unital completely positive maps with commuting ranges, 
$\cl A$ is von Neumann subalgebra of $\cl B(H)$, 
$\cl B$ is a von Neumann algebra with $\cl B^o\subseteq\cl B(H)$, 
$H = \mbox{}_{\cl A}H_{{\cl B}}$ is a bipartite quantum system over 
$(\cl A, \cl B)$, and the inclusions
$\varphi_{A}(\cl S_A)\subseteq \cl A$ and 
$\varphi_B(\cl S_B)\subseteq \cl B^o$ hold.
We say that $S$ is a \emph{Haag model} if $\cl B^o = \cl A'$.
We note that every pair $(\varphi_A,\varphi_B)$ of 
unital completely positive maps with commuting ranges, say, 
$\varphi_A : \cl S_A\to \cl B(H)$ and $\varphi_B : \cl S_B\to \cl B(H)$, 
and a choice of a unit vector
$\xi\in H$, give rise to a canonical quantum commuting model 
over the pair $(\cl S_A,\cl S_B)$ by letting $\cl A$ (resp. $\cl B^o$) 
be the von Neumann algebra, generated by $\varphi_A(\cl S_A)$
(resp. $\varphi_B(\cl S_B)$).

\begin{defn}\label{d_appdil}
Let $\cl S_A$ and $\cl S_B$ be operator systems, and 
suppose that 
$S = (_{\cl A}H_{\cl B}, \varphi_A, \varphi_B, \xi)$  and 
$\wt{S} = (_{\wt{\cl A}}\wt{H}_{\wt{\cl B}}, \wt{\varphi}_A, \wt{\varphi}_B, \wt{\xi})$ are quantum commuting models over the pair $(\cl S_A,\cl S_B)$.

\begin{itemize}
\item[(i)] We say that $\wt{S}$
is a \emph{local dilation} of $S$, and write $S\preceq \wt{S}$, 
if there exist a bipartite quantum system 
$_{{\cl A}_{\rm aux}}(H_{\rm aux})_{{\cl B}_{\rm aux}}$ 
over a pair of von Neumann algebras (${\cl A}_{\rm aux}$, ${\cl B}_{\rm aux}$), a unit vector $\xi_{\rm aux}\in H_{\rm aux}$, and
a local isometry $V : \mbox{}_{\cl A}H_{\cl B}\to
\mbox{}_{\wt{\cl A}}\wt{H}_{\wt{\cl B}}\otimes
\mbox{}_{{\cl A}_{\rm aux}}(H_{\rm aux})_{{\cl B}_{\rm aux}}$ such that 
\begin{equation}\label{eq_TAB}
V\varphi_A(s) \varphi_B(t)\xi = (\tilde\varphi_A(s)\tilde\varphi_B(t)\tilde\xi) \otimes\xi_{\rm aux}, \ \ \ s\in \cl S_A, t\in\cl S_B.
\end{equation}

\item[(ii)] 
We say that $\wt{S}$
is an \emph{approximate local dilation} of $S$, and write $S\preceq_{\rm a} \wt{S}$, 
if there exist bipartite quantum systems
$_{{\cl A}_{i,\rm aux}}(H_{i,\rm aux})_{{\cl B}_{i,\rm aux}}$ 
over a pair of von Neumann algebras (${\cl A}_{i,\rm aux}$, ${\cl B}_{i,\rm aux}$), 
unit vectors $\xi_{i,\rm aux}\in H_{i,\rm aux}$, and
local isometries $V_i : \mbox{}_{\cl A}H_{\cl B}\to
\mbox{}_{\wt{\cl A}}\wt{H}_{\wt{\cl B}}\otimes
\mbox{}_{{\cl A}_{i,\rm aux}}(H_{i,\rm aux})_{{\cl B}_{i,\rm aux}}$, 
$i\in \bb{N}$, such that 
\begin{equation}\label{eq_TABa}
\left\|V_i\varphi_A(s) \varphi_B(t)\xi - (\tilde\varphi_A(s)\tilde\varphi_B(t)\tilde\xi) \otimes\xi_{i,\rm aux}\right\|\to_{i\to \infty} 0, \ \ \ s\in \cl S_A, t\in\cl S_B.
\end{equation}
\end{itemize}
\end{defn}

We note that the reverse notation $\wt{S}\preceq S$ was used in 
\cite{pszz} to designate that $\wt{S}$
is a local dilation of $S$, but we have decided to employ the 
one specified in Definition \ref{d_appdil} as it agrees 
with the usual conventions in operator algebra theory; see e.g. 
\cite{davk}. 

\begin{remark} In self-testing examples, an ideal model $\wt{S}$ is typically smaller (in the sense of Hilbert space dimension) than a given model $S$ satisfying $S\preceq\wt{S}$ (in line with the reverse notation from \cite{pszz}). However, assistance from an auxiliary system together with a (typically entangled) auxiliary state is needed to witness the local dilation property $S\preceq \wt{S}$. We therefore view this property as a form of entanglement assisted local dilation.
\end{remark}

\begin{remark}\label{r_someap}
\rm 
By linearity, condition (\ref{eq_TAB}) is equivalent to 
$$V(\varphi_A\cdot\varphi_B)(u)\xi = 
(\tilde\varphi_A\cdot\tilde\varphi_B)(u)\tilde\xi \otimes\xi_{\rm aux}$$
being fulfilled for every $u$ in the 
completion of the tensor product $\cl S_A\otimes_{\rm c} \cl S_B$.
In addition, choosing $u = 1$, we have that 
$V\xi = \tilde{\xi}\otimes \xi_{\rm aux}$.

An $\epsilon/3$-argument, together with the uniform boundedness of the 
sequence $(V_i)_{i\in \bb{N}}$ appearing in (\ref{eq_TABa}), shows that 
(\ref{eq_TABa}) is equivalent to the condition
\begin{equation}\label{eq_arbu}
\left\|V_i (\varphi_A\cdot\varphi_B)(u) \xi - 
(\tilde\varphi_A\cdot\tilde\varphi_B)(u)\tilde\xi \otimes\xi_{i,\rm aux}\right\|\to_{i\to \infty} 0,
\end{equation}
for all $u$ in the completion of the tensor product $\cl S_A\otimes_{\rm c} \cl S_B$. 
Choosing $u = 1$, we have, in particular, 
$$\left\|V_i \xi - \tilde\xi \otimes\xi_{i,\rm aux}\right\|\to_{i\to \infty} 0.$$
\end{remark}

\noindent It is clear that 
$$S\preceq \wt{S} \ \ \Longrightarrow \ \ S\preceq_{\rm a} \wt{S}.$$

\begin{proposition}\label{p_preorder}
The relations $\preceq$ and $\preceq_{\rm a}$ are preorders.    
\end{proposition}

\begin{proof}
Assume that $S\preceq_{\rm a}\wt{S}$ and $\wt{S}\preceq_{\rm a} \wt{\wt{S}}$, 
realised via sequences $(V_i)_{i\in \bb{N}}$ and $(\wt{V}_i)_{i\in \bb{N}}$
of local isometries, where 
$V_i : \mbox{}_{{\cl A}}{H}_{{\cl B}}
\to \mbox{}_{{\wt{\cl A}}}{\wt{H}}_{{\wt{\cl B}}}\otimes 
\mbox{}_{{{\cl A}}_{i,\rm aux}}({H}_{i,\rm aux})_{{{\cl B}}_{i,\rm aux}}$
and 
$\wt{V}_i: \mbox{}_{\wt{\cl A}}\wt{H}_{\wt{\cl B}}\to
\mbox{}_{\wt{\wt{\cl A}}}\wt{\wt{H}}_{\wt{\wt{\cl B}}}\otimes
\mbox{}_{{\wt{\cl A}}_{i,\rm aux}}(\wt{H}_{i,\rm aux})_{{\wt{\cl B}}_{i,\rm aux}}$,
$i\in \bb{N}$. 
By Remark \ref{r_compo}, 
$\wt{V}_i\otimes I_{H_{i, \rm aux}}$ is a local isometry from $_{{\wt{\cl A}}}{\wt{H}}_{{\wt{\cl B}}}\otimes\mbox{}_{{{\cl A}}_{i, \rm aux}}({H}_{\rm aux})_{{{\cl B}}_{i, \rm aux}}$ to $_{\wt{\wt{\cl A}}}\wt{\wt{H}}_{\wt{\wt{\cl B}}}\otimes\mbox{}_{{\wt{\cl A}}_{i, \rm aux}}(\wt{H}_{i, \rm aux})_{{\wt{\cl B}}_{i, \rm aux}}\otimes\mbox{}_{{{\cl A}}_{i, \rm aux}}({H}_{i,\rm aux})_{{{\cl B}}_{i,\rm aux}}$. 
It follows now from 
Theorem \ref{preorder} that $(\wt{V}_i\otimes I_{H_{i,\rm aux}})V_i$ is a local isometry  from $_{{\cl A}}{H}_{{\cl B}}$ to $_{\wt{\wt{\cl A}}}\wt{\wt{H}}_{\wt{\wt{\cl B}}}\otimes\mbox{}_{{\wt{\cl A}}_{i,\rm aux}}(\wt{H}_{i,\rm aux})_{{\wt{\cl B}}_{i,\rm aux}}\otimes\mbox{}_{{{\cl A}}_{i,\rm aux}}({H}_{i,\rm aux})_{{{\cl B}}_{i,\rm aux}}$, $i\in \bb{N}$. 
Further, if $s\in \cl S_A$ and $t\in \cl S_B$ then 
\begin{eqnarray*}
& & 
\left\|(\wt{V}_i\otimes I_{H_{i,\rm aux}})V_i\varphi_A(s) \varphi_B(t) \xi 
- 
(\wt{\wt{\varphi}}_A(s)\wt{\wt{\varphi}}_B(t)\wt{\wt{\xi}})
\otimes\wt{\xi}_{i,\rm aux}\otimes\xi_{i,\rm aux}\right\|\\
&& \leq  
\left\|(\wt{V}_i\otimes I_{H_{i,\rm aux}})V_i\varphi_A(s) \varphi_B(t)\xi - (\wt{V}_i\otimes I_{H_{i,\rm aux}})(\tilde\varphi_A(s)\tilde\varphi_B(t)\tilde\xi) \otimes\xi_{i,\rm aux}\right\|\\
&& +  
\left\|(\wt{V}_i\otimes I_{H_{i,\rm aux}})\tilde\varphi_A(s)\tilde\varphi_B(t)\tilde\xi \otimes\xi_{i,\rm aux} - 
(\wt{\wt{\varphi}}_A(s)\wt{\wt{\varphi}}_B(t)\wt{\wt{\xi}})
\otimes\wt{\xi}_{i,\rm aux}\otimes\xi_{i,\rm aux}\right\|\\
& &\leq  
\left\|V_i\varphi_A(s) \varphi_B(t)\xi 
- 
\tilde\varphi_A(s)\tilde\varphi_B(t)\tilde\xi \otimes\xi_{i,\rm aux}\right\|\\
&& +  
\left\|(\wt{V}_i\otimes I_{H_{i,\rm aux}})\tilde\varphi_A(s)\tilde\varphi_B(t)\tilde\xi  
- 
(\wt{\wt{\varphi}}_A(s)\wt{\wt{\varphi}}_B(t)\wt{\wt{\xi}})
\otimes\wt{\xi}_{i,\rm aux}\right\|.
\end{eqnarray*}
It follows that 
$$ \left\|(\wt{V}_i\otimes I_{H_{i,\rm aux}})V_i\varphi_A(s) \varphi_B(t) \xi 
- 
(\wt{\wt{\varphi}}_A(s)\wt{\wt{\varphi}}_B(t)\wt{\wt{\xi}})
\otimes\wt{\xi}_{i,\rm aux}\otimes\xi_{i,\rm aux}\right\|\to_{i\to \infty} 0.$$
Thus the family $((\wt{V}_i\otimes I_{H_{i,\rm aux}})V_i)_{i\in \bb{N}}$
implements an approximate dilation yielding the relation 
$S\preceq_{\rm a} \wt{\wt{S}}$, and showing that 
the relation $\preceq_{\rm a}$ is a preorder.
The fact that $\preceq$ is a preorder now follows
after noticing that it is a special case of $\preceq_{\rm a}$, 
implemented by constant sequences of local isometries. 
\end{proof}


To each model $S = (_{\cl A}H_{\cl B}, \varphi_A, \varphi_B, \xi)$, we associate a linear functional 
$f_S : \cl S_A\otimes_{\rm c} \cl S_B\to\mathbb C$ by letting
$$f_S(u)=\langle (\nph_A\cdot\nph_B)(u)\xi,\xi\rangle, \ \ \ 
u\in \cl S_A\otimes_{\rm c} \cl S_B,$$ 
and a linear functional 
$\tilde{f}_S : C^*_u(\cl S_A)\otimes_{\max} C^*_u(\cl S_B)\to\mathbb C$ by letting
$$\tilde{f}_S(u) = \langle (\pi_A\cdot\pi_B)(u)\xi,\xi\rangle, 
\ \ \ u\in C^*_u(\cl S_A)\otimes_{\max} C^*_u(\cl S_B).$$
We note that $f_S$ and $\tilde{f}_S$
depend only on the tuple $(H, \varphi_A, \varphi_B, \xi)$
and not on the von Neumann algebras $\cl A$ and $\cl B$.
It is straightforward that if $S\preceq \wt{S}$ then $f_S=f_{\wt{S}}$. 
In fact, the same implication holds for approximate dilations.

\begin{proposition}\label{p_fS=f}
If $S\preceq_{\rm a} \wt{S}$ then $f_S=f_{\wt{S}}$.
\end{proposition}

\begin{proof}
Let 
$_{{\cl A}_{i,\rm aux}}(H_{i,\rm aux})_{{\cl B}_{i,\rm aux}}$ 
be bipartite quantum systems,
$\xi_{i,\rm aux}\in H_{i,\rm aux}$ be unit vectors, and
$V_i : \mbox{}_{\cl A}H_{\cl B}\to
\mbox{}_{\wt{\cl A}}\wt{H}_{\wt{\cl B}}\otimes
\mbox{}_{{\cl A}_{i,\rm aux}}(H_{i,\rm aux})_{{\cl B}_{\rm aux}}$
be local isometries, $i\in \bb{N}$, for which condition (\ref{eq_TABa}) holds. 
Given $\epsilon > 0$ and $u\in \cl S_A\otimes_{\rm c}\cl S_B$, 
in view of (\ref{eq_arbu}), let $i\in \bb{N}$ be such that 
$$\left\|V_i (\varphi_A\cdot\varphi_B)(u) \xi - 
(\tilde\varphi_A\cdot\tilde\varphi_B)(u)\tilde\xi \otimes\xi_{i,\rm aux}\right\| 
< \frac{\epsilon}{2}$$
and 
$$\left\|V_i \xi - \wt{\xi}\otimes \xi_{i,\rm aux}\right\| < \frac{\epsilon}{2\|u\|}.$$
We have 
\begin{eqnarray*}
& & \left|f_{\wt{S}}(u) - f_{S}(u)\right|
= 
\left|\left\langle (\tilde\varphi_A\cdot \tilde\varphi_B)(u)\tilde\xi,\tilde\xi \right\rangle
- \left\langle (\varphi_A\cdot \varphi_B)(u)\xi,\xi \right\rangle\right|\\
& = & 
\left|\left\langle (\tilde\varphi_A\cdot\tilde\varphi_B)(u)\tilde\xi \otimes\xi_{i,\rm aux}, 
\tilde\xi \otimes\xi_{i,\rm aux}\right\rangle
- 
\left\langle V_i(\varphi_A\cdot \varphi_B)(u)\xi,V_i\xi \right\rangle\right|
\\
& \leq &
\left|\left\langle (\tilde\varphi_A\cdot\tilde\varphi_B)(u)\tilde\xi \otimes\xi_{i,\rm aux}, 
\tilde\xi \otimes\xi_{i,\rm aux}\right\rangle
- 
\left\langle V_i(\varphi_A\cdot \varphi_B)(u)\xi,\tilde\xi \otimes\xi_{i,\rm aux} \right\rangle\right|\\
& + & 
\left|\left\langle V_i(\varphi_A\cdot \varphi_B)(u)\xi,\tilde\xi \otimes\xi_{i,\rm aux} \right\rangle
-
\left\langle V_i(\varphi_A\cdot \varphi_B)(u)\xi,V_i\xi \right\rangle
\right|\\
& \leq & 
\frac{\epsilon}{2} \left\|\tilde\xi \otimes\xi_{i,\rm aux}\right\| 
+ \frac{\epsilon}{2\|u\|} \left\|V_i(\varphi_A\cdot \varphi_B)(u)\xi\right\| \leq \epsilon.
\end{eqnarray*}
As $\epsilon$ is an arbitrary positive real, we have that 
$f_{\wt{S}}(u) = f_{S}(u)$ for every $u\in \cl S_A\otimes_{\rm c}\cl S_B$.
\end{proof}

\begin{defn}\label{d_st}
Let $\cl C$ be a class of quantum commuting models
over the pair $(\cl S_A,\cl S_B)$, and 
$\cl S$ be a subset of states of the C*-algebra
$C_u^*(\cl S_A)\otimes_{\max}C_u^*(\cl S_B)$.

\begin{itemize}
\item[(i)]
We say that a state 
$f : \cl S_A\otimes_{\rm c} \cl S_B\to\mathbb C$ is a 
\emph{self-test} for the class $\cl C$ if there exists
$\wt{S}\in\cl C$ such that $f=f_{\wt{S}}$ and, 
whenever $S\in\cl C$ is such that $f_S = f$, we have that 
$S\preceq \wt{S}$. 
In this case, we say that $\wt{S}$ is an \emph{ideal model} for $f$.

\item[(ii)]
We say that a state 
$f : \cl S_A\otimes_{\rm c} \cl S_B\to\mathbb C$ is a 
\emph{weak self-test} for the class $\cl C$ if there exists
$\wt{S}\in\cl C$ such that $f=f_{\wt{S}}$ and, 
whenever $S\in\cl C$ is such that $f_S = f$, we have that 
$S\preceq_{\rm a} \wt{S}$. 
In this case, we say that $\wt{S}$ is an \emph{weak ideal model} for $f$.

\item[(iii)]
We say that a state $f : \cl S_A\otimes_{\rm c}\cl S_B\to\mathbb C$ is 
an \emph{abstract self-test} for $\cl S$ 
if there exists a unique state $g\in \cl S$ such that 
$g|_{\cl S_A\otimes_{\rm c}\cl S_B} = f$.
\end{itemize}
\end{defn}

It is clear that, in the notation of Definition \ref{d_st},  
the state $f : \cl S_A\otimes_{\rm c} \cl S_B\to\bb{C}$ is a 
self-test for the class $\cl C$ if and only if 
it is a self-test for the class 
$\cl C_f := \{S\in \cl C : f_S = f\}$.

\begin{remark} 
The setup of Definition \ref{d_st} will be applied in 
several different contexts in Section \ref{s_app}. 
The primary motivation behind it is the 
self-testing paradigm for no-signalling correlations of quantum type \cite{mps, pszz}. 
Letting $X,Y,A,B$ be finite sets, 
$\cl{S}_{X,A}$ (resp. $\cl S_{Y,B}$) 
be the universal operator system of $|X|$ (resp. $|Y|$) 
POVM's, each of cardinality $|A|$ (resp. $|B|$, see (\ref{eq_defSXA0})), Definition \ref{d_st} (i) provides a commuting operator generalisation of the 
usually employed notion of self-testing for quantum spacial systems (see e.g. \cite{supic-bowles}). We refer the reader to Subsection \ref{ss_POVMNS} for an extended discussion.

Note that \emph{robust self-testing} for 
no-signalling correlations arising from finite dimensional 
models (see e.g. \cite[Definition 3.3]{zhao}) is a priori stronger than 
the corresponding notion of 
weak self-testing from Definition \ref{d_st} (ii).
\end{remark}

We recall that, if $H$ is a Hilbert space, $\cl M \subseteq \cl B(H)$ 
is a von Neumann algebra and $\omega : \cl M\to \bb{C}$
is a normal positive functional, the \emph{support} of $\omega$ 
\cite[Definition 7.2.4]{kadison-ringrose}
can be defined as
the smallest projection $p\in \cl M$ such that $\omega(x) = \omega(pxp)$ for all $x\in \cl M$. 
For $\xi\in H$, write $\omega_\xi$ for the vector state, given by 
$\omega_\xi(x) = \langle x\xi,\xi\rangle$, $x\in \cl B(H)$.
Note that, if $p$ is the support of the functional $\omega_{\xi}|_{\cl M}$ then 
\begin{equation}\label{eq_supre}
p\xi = \xi.
\end{equation} 
Indeed, 
$$\|\xi-p\xi\|^2=\langle \xi,\xi\rangle-\langle \xi,p\xi\rangle-\langle p\xi,\xi\rangle+\langle p\xi,p\xi\rangle=\omega_\xi(1)-\omega_\xi(p)=0,$$
and (\ref{eq_supre}) follows. 

Given a quantum commuting model $S=(_{\cl A}\hspace{-0.04cm}H_{\cl B^o}, \varphi_A, \varphi_B, \xi)$, 
where $\cl A, \cl B^o\subseteq \cl B(H)$,
we let $\epsilon_A \in \cl A$ (resp. $\epsilon_B\in \cl B^o$) 
be the support projection of the state $\omega_{\xi}|_{\cl A}$
(resp. $\omega_{\xi}|_{\cl B^o}$) of $\cl A$ (resp. $\cl B^o$). Following \cite{pszz}, we say that $S$ is \emph{centrally supported} if 
$$\varphi_A(s)\epsilon_A = \epsilon_A\varphi_A(s) \ 
\mbox{ and } \ \varphi_B(t)\epsilon_B = \epsilon_B\varphi_B(t),
\ \ s\in \cl S_A, t\in \cl S_B.$$

\begin{lemma}\label{uptoeps}
    Let $\cl A\subseteq\cl B(H)$ be a von Neumann algebra and 
    $\xi\in H$ be a unit vector. 
    Assume that $\xi$ is separating for $\cl A$.
    Then, given $x\in\cl A$ and $\epsilon > 0$, there exists $y\in\cl A'$ such that $\|x\xi-y\xi\|<\epsilon$. 
\end{lemma}

\begin{proof}
Since $\xi$ is separating for $\cl A$, it is cyclic for its commutant 
$\cl A'$. The statement is now immediate.
\end{proof}

In what follows, for $\xi$, $\eta\in H$ and $\epsilon > 0$, 
we write $\xi\sim^\epsilon \eta$ if $\|\xi-\eta\|<\epsilon$.

\begin{lemma}\label{sep_cycl}
 Let $\cl A \subseteq \cl B(H)$ be a von Neumann algebra and  $\xi\in H$
 be a unit vector. Let  $\epsilon_A\in \cl A$ be the support 
 of the state $\omega_\xi|_{\cl A}$ of $\cl A$, 
 and $p_A\in \cl B(H)$ be the projection onto 
 $\overline{\cl A\xi}$.  Then $\xi \in \epsilon_Ap_AH$, and 
 $\xi$ is cyclic and separating for the von Neumann subalgebra 
 $\epsilon_Ap_A\cl A\epsilon_Ap_A$ of $\cl B(\epsilon_Ap_AH)$. 
\end{lemma}

\begin{proof}
Note that, since $p_A\in\cl A'$ and $\epsilon_A\in\cl A$, 
we have that $\epsilon_Ap_A\cl A\epsilon_Ap_A$ is indeed 
a von Neumann algebra acting on $\epsilon_Ap_AH$. 
By (\ref{eq_supre}), $\epsilon_Ap_A\xi=\xi$ and hence 
$\omega_{\epsilon_Ap_A\xi}(\epsilon_Ap_Aa\epsilon_Ap_A)=\omega_\xi(a)$, 
$a\in\cl A$. 

Assume that $a\in\cl A$, $\epsilon_Ap_Aa\epsilon_Ap_A\geq 0$ and $\omega_\xi(p_A\epsilon_Aa\epsilon_Ap_A)=0$. Then 
$$0 = p_A\epsilon_Aa\epsilon_Ap_A\xi=\epsilon_Aa\epsilon_Ap_A\xi=\epsilon_Aa\epsilon_A\xi;$$ 
thus $\epsilon_Aa^*\epsilon_Aa\epsilon_A\xi=0$ and hence
$\omega_\xi(\epsilon_Aa^*\epsilon_Aa\epsilon_A)=0$. As $\omega_\xi$ is faithful 
on $\epsilon_A\cl A\epsilon_A$, we get $\epsilon_Aa^*\epsilon_Aa\epsilon_A=0$ and hence $\epsilon_Aa\epsilon_A=0$. 
This implies that $\omega_{\xi}$ is faithful on $\epsilon_Ap_A\cl A\epsilon_Ap_A$ and therefore, $\xi$ is separating for $\epsilon_Ap_A\cl A\epsilon_Ap_A$.

Let $\eta\in H$ be such that
$\langle \epsilon_Ap_Aa\epsilon_Ap_A\xi, \epsilon_Ap_A\eta\rangle=0$
for all $a\in\cl A$. Then 
$\langle a\xi, p_A\epsilon_A\eta\rangle=0$ and, 
since $p_A\epsilon_A\eta \in p_AH$, while 
$p_AH = \overline{\cl A\xi}$, 
we have $p_A\epsilon_A\eta=0$, showing that $\xi$ is cyclic. 
\end{proof}

\begin{lemma}\label{support_com}
Let $H$ be a Hilbert space, $\cl A\subseteq \cl B(H)$ be a von Neumann algebra, 
$\xi\in H$, and $\epsilon_A$ be the support projection of $\omega_\xi|_{\cl A}$. 
Suppose that for every selfadjoint $a\in\cl A$ and 
every $\epsilon > 0$, there exists
$\hat a\in\cl A'$ such that $a\xi\sim^\epsilon\hat a\xi$. Then $[\epsilon_A,a] = 0$
for every $a\in \cl A$.
    \end{lemma}
    
    \begin{proof}
        By (\ref{eq_supre}), $\epsilon_A\xi=\xi$.
        Thus 
        $$\epsilon_Aa\xi\sim^{\epsilon}\epsilon_A\hat a\xi=\hat a\epsilon_A\xi=\hat a\xi\sim^{\epsilon}a\xi.$$
        Therefore $\omega_\xi(a(1-\epsilon_A)a)=0$ and, as 
        $\epsilon_A$ is the support projection of $\omega_{\xi}|_{\cl A}$ 
        and $a(1-\epsilon_A)a$ is positive, we obtain 
        $\epsilon_Aa(1-\epsilon_A)a\epsilon_A=0$, and hence $(1-\epsilon_A)a\epsilon_A=0$,
        that is, 
        $a\epsilon_A=\epsilon_A a\epsilon_A$. As $a$ is selfadjoint, 
        this implies $\epsilon_A a = a\epsilon_A$.
    \end{proof}

\begin{lemma}\label{p_apextends}
 Let $S = (_{\cl A}H_{\cl B}, \varphi_A, \varphi_B, \xi)$ and
 $\wt{S} = (_{\wt{\cl A}}\wt{H}_{\wt{\cl B}}, \wt{\varphi}_A, 
 \wt{\varphi}_B, \wt{\xi})$ be quantum commuting models
 over $(\cl S_A,\cl S_B)$. 
 Suppose that $S\preceq_{\rm a} \wt{S}$ and that $\wt{S}$ is centrally supported. 
 Then the states 
 $\omega_\xi\circ\pi_{\varphi_A\cdot\varphi_B}$
 and $\omega_{\wt{\xi}}\circ\pi_{\wt{\varphi}_A\cdot\wt{\varphi}_B}$ of $C^*_u(\cl S_A)\otimes_{\max} C^*_u(\cl S_B)$ coincide. 
\end{lemma}

\begin{proof}
We write $\nph = \nph_A\cdot\nph_B$ and 
$\wt{\nph} = \wt{\nph}_A\cdot \wt{\nph}_B$ for brevity. 
It suffices to show that  $(\omega_\xi\circ\pi_{\varphi})(a\otimes b) = (\omega_{\wt{\xi}}\circ\pi_{\wt{\varphi}})(a\otimes b)$ holds for all 
elements $a\in C_u^*(S_A)$ (resp. $b\in C_u^*(S_B)$) of the form 
$a = s_1\ldots s_n$ (resp. $b = t_1\ldots t_n$), where $s_i\in \cl S_A$
(resp. $t_i\in\cl S_B$), $i = 1,\dots,n$. 
Write $\pi_A(a) = \pi_\varphi(a\otimes 1)$, $\wt\pi_A(a) = \pi_{\wt\varphi}(a\otimes 1)$, 
$\pi_B(b) = \pi_\varphi(1\otimes b)$ and $\wt\pi_B(b) = \pi_{\wt\varphi}(1\otimes b)$, $a\in C_u^*(\cl S_A)$, $b\in C_u^*(\cl S_B)$. 

Let $(V_i)_{i\in \bb{N}}$ be a sequence of local isometries, 
where $V_i : H\to \wt{H}\otimes H_{i,\rm aux}$, $i\in \bb{N}$,
implementing the preorder assumption $S\preceq_{\rm a} \wt{S}$
(see Definition \ref{d_appdil} (ii)), and write
$V_i = V_{i,2,2}\circ V_{i,1,2} = V_{i,2,1}\circ V_{i,1,1}$, where 
$V_{i,1,2}$, $V_{i,2,1}$ are $\cl{B}$-local (and $V_{i,1,1}$, $V_{i,2,2}$  are $\cl{A}$-local).

We have that 
\begin{equation}\label{eq_Vi22}
V_{i,1,2}\pi_A(a) =\sigma_{A,i}(\pi_A(a))V_{i,1,2}\ \mbox{ and } \ 
V_{i,2,2}\sigma_{A,i}(\pi_A(a))= \rho_{A,i}(\pi_A(a))V_{i,2,2},
\end{equation}
for $a\in C_u^*(\cl S_A)$ and 
for some normal *-representations $\sigma_{A,i}$ and $\rho_{A,i}$ of $\cl A$, the latter mapping into $\pi_{\wt{H}}(\wt{\cl A})\bar\otimes\pi_{H_{i,\rm aux}}(\cl A_{i,\rm aux})$.
Let $\wt{\epsilon}_{A}$ be the support projection of 
the restriction of $\omega_{\wt{\xi}}$ to $\pi_{\wt{H}}({\wt{\cl A}})$
and $\wt{p}_A$ be the projection onto $\overline{\pi_{\wt{H}}(\wt{\cl A})\wt{\xi}}$. 
Set $\wt{q}_A = \wt{\epsilon}_A\wt{p}_A$. As $\wt{S}$ is centrally supported, 
\begin{equation}\label{eq_cen!}
\wt{\epsilon}_A\wt{\pi}_A(a) = \wt{\pi}_A(a)\wt{\epsilon}_A, 
\ \ \ a\in C_u^*(\cl S_A), 
\end{equation}
and hence 
\begin{equation}\label{eq_wtq}
\wt{q}_A\wt{\pi}_A(a) = \wt{\pi}_A(a)\wt{q}_A, \ \ \ a\in C_u^*(\cl S_A).
\end{equation}
We claim that for every $n\in\N$, contractions $s_1,\dots,s_n\in\cl S_A$, $\epsilon>0$, there exists $i_{0}$ such that, 
setting $a=s_1\cdots s_n$, if $i\geq i_{0}$ then
\begin{eqnarray}\label{eq_aimind}
\wt{\pi}_A(a)\wt{\xi}\otimes\xi_{i,{\rm aux}}
& \sim^{\epsilon} & (\wt{\epsilon}_A
\otimes I_{H_{i,{\rm aux}}})\rho_{A,i}(\pi_A(a))(\wt{\xi}\otimes
\xi_{i,{\rm aux}}).
\end{eqnarray}

Let $n=1$, $s\in\cl S_A$ be contractive and $\epsilon > 0$. By the relation $S\preceq_{\rm a} \wt{S}$ and (\ref{eq_Vi22}), there exists $i_0$ such that, if $i\geq i_0$ then 
\begin{eqnarray*}
\wt{\pi}_A(s)\wt{\xi}\otimes\xi_{i,\rm aux}
& \sim^{\epsilon/2} &
V_{i,2,2}V_{i,1,2}\pi_A(s)\xi = V_{i,2,2}\sigma_{A,i}(\pi_A(s))V_{i,1,2}\xi\\
& = & \rho_{A,i}(\pi_A(s))V_i\xi\\
& \sim^{\epsilon/2}&
\rho_{A,i}(\pi_A(s))(\wt{\xi}\otimes\xi_{i,\rm aux}).
\end{eqnarray*} 
Claim (\ref{eq_aimind}) now follows from (\ref{eq_supre}) and (\ref{eq_cen!}).

Assume the claim holds for $n>1$. Consider $a = s_1\ldots s_ns_{n+1}$, where each $s_{i}$ is a contraction in $\cl S_A$, and let $\epsilon>0$. By hypothesis (and the case $n=1$), there is an $i_0$ such that for every $i\geq i_0$,
\begin{eqnarray}\label{eq_aimind2}
\wt{\pi}_A(a')\wt{\xi}\otimes\xi_{i,{\rm aux}}
& \sim^{\epsilon/5} & (\wt{\epsilon}_A\otimes 1_{H_{i,{\rm aux}}})\rho_{A,i}(\pi_A(a'))(\wt{\xi}\otimes
\xi_{i,{\rm aux}}),
\end{eqnarray}
where $a'=s_1\ldots s_n$ or $a'=s_{n+1}$. By Lemmas \ref{uptoeps} and \ref{sep_cycl}, 
 there exists 
 $\tilde a\in(\wt{q}_A\pi_{\wt{H}}(\wt{\cl A})\wt{q}_A)'\subseteq 
 \cl B(\wt{q}_A\wt{H})$ such that  
$$\wt{\pi}_A(s_{n+1})\wt{\xi}=\wt{q}_A\wt{\pi}_A(s_{n+1})\wt{q}_A\wt{\xi}\sim^{\epsilon/5}\tilde a\wt{\xi}.$$
Using (\ref{eq_Vi22}), (\ref{eq_cen!}), (\ref{eq_wtq}), 
(\ref{eq_aimind2}) and 
the fact that 
$$\rho_{A,i}(\pi_A(a'))\in\pi_{\wt{H}}(\wt{\cl A})\bar\otimes\pi_{H_{i,\rm aux}}(\cl A_{i,\rm aux}),$$ 
for $i\geq i_0$, we have
\begin{eqnarray*}
 &&
 \hspace{-0.7cm} \wt{\pi}_A(a')\wt{\pi}_A(s_{n+1})\wt{\xi}\otimes\xi_{i,\rm aux}
 \sim^{\epsilon/5}   
 \wt{\pi}_A(a')\tilde a\wt{\xi}\otimes\xi_{i,\rm aux}\\
 && =
 \wt{q}_A\wt{\pi}_A(a')\wt{q}_A\tilde a\wt{\xi}\otimes\xi_{i,\rm aux}
 =
 (\tilde a\otimes 1_{i,\rm aux})(\wt{\pi}_A(a')\wt{\xi}\otimes\xi_{i,\rm aux})\\
 && \sim^{\epsilon/5}
(\tilde a\otimes 1_{i,\rm aux})(\wt{\epsilon}_A\otimes 1_{H_{i,\rm aux}})\rho_{A,i}(\pi_A(a'))(\wt{\xi}\otimes\xi_{i,\rm aux})\\
 && =
 (\wt{q}_A\otimes 1_{H_{i,{\rm aux}}})\rho_{A,i}(\pi_A(a'))
 (\tilde a\wt{\xi}\otimes\xi_{i,{\rm aux}})\\
 && \sim^{\epsilon/5} 
 (\wt{q}_A\otimes 1_{H_{i,\rm aux}})\rho_{A,i}(\pi_A(a'))(\wt{\pi}_A(s_{n+1})\wt{\xi}\otimes\xi_{i,{\rm aux}})\\
 && \sim^{\epsilon/5}
 (\wt{q}_A\otimes 1_{H_{i,\rm aux}})\rho_{A,i}(\pi_A(a'))V_{i,2,2}V_{i,1,2}\pi_A(s_{n+1})\xi\\
 && =
 (\wt{q}_A\otimes 1_{H_{i,\rm aux}})\rho_{A,i}(\pi_A(a'))(V_{i,2,2}\sigma_{A,i}(\pi_A(s_{n+1}))V_{i,1,2}\xi\\
 && =
 (\wt{q}_A\otimes 1_{H_{i,{\rm aux}}})\rho_{A,i}(\pi_A(a's_{n+1}))V_{i,2,2}V_{i,1,2}\xi\\
 && =
 (\wt{q}_A\otimes 1_{H_{i,\rm aux}})\rho_{A,i}(\pi_A(a))V_i\xi\\
 && \sim^{\epsilon/5}
 (\wt{\epsilon}_A\otimes 1_{H_{i,{\rm aux}}})\rho_{A,i}(\pi_A(a))(\wt{\xi}\otimes\xi_{i, {\rm aux}}).
\end{eqnarray*}
Thus, 
$$\wt{\pi}_A(a)\wt{\xi}\otimes\xi_{i,{\rm aux}} 
\sim^{\epsilon} (\wt{\epsilon}_A\otimes 1_{H_{i,{\rm aux}}})
\rho_{A,i}(\pi_A(a))(\wt{\xi}\otimes\xi_{i,{\rm aux}}).$$
A similar induction argument shows that, 
if $\wt{\epsilon}_B$ is the support projection of the restriction of  
the functional $\omega_{\wt{\xi}}$ on $\pi_{\wt{H}}(\cl B^o)$, then for all monomials $b\in C_u^*(\cl S_B)$ whose individual terms are contractions in $\cl S_B$, and $\epsilon>0$, there is an $i_0$ such that for $i\geq i_0$,
\begin{equation}\label{eq_forB}
\wt{\pi}_B(b)\wt{\xi}\otimes\xi_{i,{\rm aux}}
\sim^{\epsilon} 
(\wt{\epsilon}_B\otimes 1_{H_{i,{\rm aux}}})\rho_{B,i}(\pi_B(b))(\wt{\xi}\otimes\xi_{i,{\rm aux}}),
\end{equation}
where $\rho_{B,i}$ is the normal $*$-representation of $\cl B^o$ constructed from the locality of $V_i$, analogous to $\rho_{A,i}$.
We have that 
$\wt{\epsilon}_A\in \pi_{\wt{H}}(\wt{\cl A})$, $\wt{\epsilon}_B\in \pi_{\wt{H}}(\wt{\cl B}^o)$ and $[\pi_{\wt{H}}(x), \pi_{\wt{H}}(y^o)]=0$, $x\in \wt{\cl A}$, $y\in\wt{\cl B}$.
Using (\ref{eq_supre}), (\ref{eq_cen!}),  (\ref{eq_aimind}), (\ref{eq_forB}), the relations
$$[\wt{\pi}_A(a)\otimes 1_{H_{i,\rm aux}},\rho_{B,i}(\pi_{B}(b))] = 0$$ 
and the fact that $V_{i,2,2}V_{i,1,2} = V_{i,2,1}V_{i,1,1}$, 
for monomials $a\in C_u^*(\cl S_A)$ and $b\in C_u^*(\cl S_B)$, whose individual terms are contractions in $\cl S_A$ and $\cl S_B$, respectively and $\epsilon>0$, there is an $i_0$ such that for $i\geq i_0$, we obtain 
\begin{eqnarray*}
    &&
    (\omega_{\wt{\xi}}\circ\pi_{\wt{\varphi}})(a\otimes b)
    =
    \langle \wt{\pi}_A(a)\wt{\pi}_B(b)\wt{\xi},\wt{\xi}\rangle
    =
    \langle \wt{\pi}_A(a)\wt{\pi}_B(b)\wt{\xi}\otimes\xi_{i,{\rm aux}},\wt{\xi}\otimes\xi_{i,{\rm aux}}\rangle\\
    && \sim^{\epsilon/4}
    \langle (\wt{\pi}_A(a)\wt{\epsilon}_B\otimes 1_{H_{i,{\rm aux}}})\rho_{B,i}(\pi_B(b))(\wt{\xi}\otimes\xi_{i,{\rm aux}}), \wt{\xi}\otimes\xi_{i,{\rm aux}}\rangle\\
    && =
    \langle \rho_{B,i}(\pi_B(b))(\wt{\pi}_A(a)\wt{\xi}\otimes\xi_{i,{\rm aux}}), \wt{\xi}\otimes\xi_{i,{\rm aux}}\rangle\\
    && \sim^{\epsilon/4}
    \langle \rho_{B,i}(\pi_B(b))
    (\wt{\epsilon}_A\otimes 1_{H_{i,{\rm aux}}})\rho_{A,i}(\pi_A(a))
    (\wt{\xi}\otimes\xi_{i,{\rm aux}}), 
    \wt{\xi}\otimes\xi_{i,{\rm aux}}\rangle\\
    && \sim^{\epsilon/4}
    \langle \rho_{B,i}(\pi_B(b))\rho_{A,i}(\pi_A(a))V_{i,2,2}V_{i,1,2}\xi, \wt{\xi}\otimes\xi_{i,{\rm aux}}\rangle\\
    && =
    \langle \rho_{B,i}(\pi_B(b))V_{i,2,2}V_{i,1,2}\pi_A(a)\xi, \wt{\xi}\otimes\xi_{i,{\rm aux}}\rangle\\
    && =
    \langle \rho_{B,i}(\pi_B(b))V_{i,2,1}V_{i,1,1}\pi_A(a)\xi, \wt{\xi}\otimes\xi_{i,{\rm aux}}\rangle\\
    && =
    \langle V_{i,2,1}V_{i,1,1} \pi_B(b)\pi_A(a)\xi, \wt{\xi}\otimes\xi_{i,{\rm aux}}\rangle\\
    && = 
    \langle \pi_B(b)\pi_A(a)\xi, V_i^*(\wt{\xi}\otimes\xi_{i,{\rm aux}})\rangle\\
    && \sim^{\epsilon/4}
    \langle \pi_B(b)\pi_A(a)\xi, \xi\rangle
    = 
    (\omega_\xi\circ\pi_\varphi)(a\otimes b).
    \end{eqnarray*}
Since $\epsilon$ is an arbitrary positive real, we conclude that 
$$(\omega_{\wt{\xi}}\circ\pi_{\wt{\varphi}})(a\otimes b) = (\omega_\xi\circ\pi_\varphi)(a\otimes b).$$
By linearity, density and continuity it follows that $\omega_\xi\circ\pi_{\varphi_A\cdot\varphi_B}
 = \omega_{\wt{\xi}}\circ\pi_{\wt{\varphi}_A\cdot\wt{\varphi}_B}$.
\end{proof}


\begin{proposition}\label{p_Haagn}
Let $S$ and $\wt{S}$ be models of a functional $f : \cl S_A\otimes_{\rm c}\cl S_B\to\mathbb{C}$
such that $S\preceq_{\rm a}\wt{S}$. If $S$ is a centrally supported Haag model then 
the model $\wt{S}$ is centrally supported. 
\end{proposition}

\begin{proof}
Let $S = (_{\cl A}H_{\cl B}, \varphi_A, \varphi_B, \xi)$ and 
$\wt{S} = (_{\wt{\cl A}}\wt{H}_{\wt{\cl B}}, \wt{\varphi}_A, 
\wt{\varphi}_B, \wt{\xi})$, 
    and suppose that 
    $$V_i = V_{i,2,2}V_{i,1,2} = V_{i,2,1}V_{i,1,1}$$ 
    is a local isometry from $H$ to 
    $\wt{ H}\otimes H_{i,\rm aux}$, $i\in \bb{N}$, such that the 
    sequence $(V_i)_{i\in \bb{N}}$ implements the relation $S\preceq_{\rm a}\wt{S}$; here 
    $V_{i,1,2} : H\to L_i$ and $V_{i,2,2} : L_i\to\wt{ H}\otimes H_{i,\rm aux}$, for some 
    Hilbert space $L_i$, $i\in \bb{N}$.
    Let $\epsilon>0$ and $s\in\cl S_A$ be self-adjoint.
Since $S$ is centrally supported, $\epsilon_A p_A\in \vphi_A(\cl S_A)'$, 
and thus 
$$\vphi_A(s)\xi=\epsilon_A p_A\vphi_A(s)\epsilon_A p_A\xi.$$
Noting the equality $(\epsilon_A p_A\pi_H(\cl A)\epsilon_A p_A)' = 
\epsilon_A p_A\pi_H(\cl A)'\epsilon_A p_A$, 
by Lemmas \ref{uptoeps} and \ref{sep_cycl},  
    there exists 
    $x\in  \pi_H(\cl A)'$, $x\neq 0$, such that $$\varphi_A(s)\xi\sim^\epsilon \epsilon_A p_Ax\epsilon_A p_A\xi=p_A x\xi.$$
    Choose $i\in \bb{N}$ such that 
$$\left\|V_i \varphi_A(s) \xi - 
\tilde\varphi_A(s)\tilde\xi \otimes\xi_{i,\rm aux}\right\| 
< \min\left\{\epsilon, \frac{\epsilon}{\|x\|}\right\}$$
and 
$$\left\|V_i \xi - 
\tilde\xi \otimes\xi_{i,\rm aux}\right\| 
< \frac{\epsilon}{\|x\|}.$$
We have 
\begin{eqnarray*}
\wt{\varphi}_A(s)\wt{\xi}\otimes\xi_{i,\rm aux} 
& \sim^{\epsilon} & 
V_{i,2,2}V_{i,1,2}\varphi_A(s)\xi
\sim^{\epsilon}
V_{i,2,2}V_{i,1,2}p_Ax\xi\\
& \sim^{\epsilon} & 
V_{i,2,2}V_{i,1,2}p_AxV_{i,1,2}^*V_{i,2,2}^*\wt{\xi}\otimes\xi_{i,\rm aux}.  
    \end{eqnarray*}
    By Haag duality, $p_Ax\in \pi_H(\cl A)'=\pi_H(\cl B^o)$; let $b^o\in \cl B^o$ be such that 
    $p_Ax = \pi_H(b^o)$, and note that  
$$V_{i,2,2}V_{i,1,2}\pi_H(b^o)V_{i,1,2}^*V_{i,2,2}^*
\in 
    (\pi_H(\wt{\cl B}^o)\bar\otimes \pi_{H_{i,{\rm aux}}}(\cl B_{i,{\rm aux}}^o))
    V_{i,2,2}V_{i,2,2}^*.$$
Let $y\in \pi_{\wt H}(\wt{\cl B}^o)\bar\otimes \pi_{H_{i,{\rm aux}}}(\cl B_{i,{\rm aux}}^o)$ 
be such that 
$$V_{i,2,2}V_{i,1,2}\pi_H(b^o)V_{i,1,2}^*V_{i,2,2}^* = yV_{i,2,2}V_{i,2,2}^*.$$ 
We have
$$\wt{\varphi}_A(s)\wt{\xi}\otimes\xi_{i,{\rm aux}}\sim^{3\epsilon}yV_{i,2,2}V_{i,2,2}^*(\wt{\xi}\otimes
\xi_{i,{\rm aux}}) \sim^{2 \epsilon} y(\wt{\xi}\otimes\xi_{i,{\rm aux}}).$$
If $\epsilon_{i,{\rm aux}}^A$ is the support projection of 
$\omega_{\xi_{i,{\rm aux}}}|_{\cl A_{i,{\rm aux}}}$, by Lemma \ref{support_com}, we obtain
$[\wt{\varphi}_A(s)\otimes 1_{H_{i,{\rm aux}}}, \wt{\epsilon}_A\otimes\epsilon_{i,{\rm aux}}^A]=0$ and  hence
$[\wt{\varphi}_A(s), \wt{\epsilon}_A]=0$. Similar arguments applied to $\wt{\varphi}_B$ show  that $\wt{S}$ is centrally supported.
\end{proof}

\begin{theorem}\label{th_stisabs}
Let $\cl S_A$ and $\cl S_B$ be operator systems and 
$\frak{M}$ be a family of quantum commuting models over $(\cl S_A,\cl S_B)$. 
Let $\cl S = \{\tilde{f}_S : S\in \frak{M}\}$ and 
$f : \cl S_A\otimes_{\rm c} \cl S_B\to \bb{C}$ 
is the restriction of an element of $\cl S$.
Assume that $\frak{M}$ contains a centrally supported Haag model of $f$.
If $f$ a weak self-test for $\frak{M}$ then $f$ is an abstract self-test for 
$\cl S$. 
\end{theorem}

\begin{proof}
Let $\wt{S}$ be a weak ideal model for $\frak M$. 
If $S \in \frak{M}$ is a centrally supported Haag model of $f$ then, 
by Proposition \ref{p_Haagn}, $\wt{S}$ is centrally supported. The statement now follows from Lemma \ref{p_apextends}.
\end{proof}

In the finite dimensional tensor product setting of Example \ref{ex_findimset}, 
the authors of \cite{pszz} study the notion of a full rank model and 
its relevance to self-testing. In the sequel, we 
comment on a natural commuting operator version of this notion. 
Let 
$S = (_{\cl A}H_{\cl B}, \varphi_A, \varphi_B, \xi)$ be a quantum commuting model of 
$(\cl S_A,\cl S_B)$, and let $\epsilon_A$ (resp. $\epsilon_B$) be the support of 
the functional $\omega_{\xi}|_{\cl A}$ (resp. $\omega_{\xi}|_{\cl B}$). 
Write $r = \epsilon_A\epsilon_B$, and let 
$\nph_{A,r} : \cl S_A\to \cl B(rH)$ (resp. $\nph_{B,r} : \cl S_B\to \cl B(rH)$)
be the unital completely positive map, given by 
$\nph_{A,r}(a) = r\nph_A(a)r$ (resp. $\nph_{B,r}(b) = r\nph_B(b)r$).
We note that, if
$a\in \cl S_A$ and $b\in \cl S_B$ then 
\begin{eqnarray*}
\nph_{A,r}(a)\nph_{B,r}(b) 
& = & 
\epsilon_A\epsilon_B\nph_A(a) \epsilon_A\epsilon_B \nph_B(b)\epsilon_A\epsilon_B\\
& = & 
\epsilon_A\epsilon_B\nph_A(a) \epsilon_B\epsilon_A \nph_B(b)\epsilon_A\epsilon_B
= 
\epsilon_A\epsilon_B\nph_A(a) \nph_B(b)\epsilon_A\epsilon_B\\
& = & 
\epsilon_A\epsilon_B \nph_B(b)\nph_A(a) \epsilon_A\epsilon_B
= 
\epsilon_B \epsilon_A\nph_B(b)\nph_A(a) \epsilon_B\epsilon_A\\
& = & 
\epsilon_A \epsilon_B\nph_B(b)\epsilon_A \epsilon_B\nph_A(a) \epsilon_A\epsilon_B
= 
\nph_{B,r}(b)\nph_{A,r}(a).
\end{eqnarray*}
Taking into account (\ref{eq_supre}), we see that 
$S_r := (_{r\cl A r}(rH)_{r\cl B r}, \varphi_{A,r}, \varphi_{B,r}, \xi)$
is a quantum commuting model. It is straightforward to check that 
$f_{S_r} = f_S$. We call the model $S_r$ the \emph{reduced model} of $S$.

A quantum commuting model 
$S = (_{\cl A}H_{\cl B}, \varphi_A, \varphi_B, \xi)$ is said to be of \emph{full rank}
if $\epsilon_A = \epsilon_B = I$. 

\begin{corollary}\label{c_red}
Let $\frak{M}$ be a family of quantum commuting models over $(\cl S_A,\cl S_B)$,
closed under passing to reduced models, that contains a Haag model. 
Then every weak self-test for $\frak{M}$ is an abstract self-test.
\end{corollary}

\begin{proof}
Let $S$ be a Haag model in $\frak{M}$. The reduction $S_r$ is a full rank Haag model.
Since every full rank model is trivially centrally supported, the claim 
follows from Theorem \ref{th_stisabs}. 
\end{proof}

\section{Abstract self-tests}\label{s_abstractselft}

In this section, we examine the property of being an 
abstract self-test and 
show that in some cases it is equivalent to the property of 
being a self-test. 
We fix some notation that will be used subsequently. 
Let $\cl S_A$ and $\cl S_B$ be operator systems, 
$\cl S \subseteq S(C^*_u(\cl S_A)\otimes_{\max} C^*_u(\cl S_B))$ 
and, writing $\tilde{\cl S} = \{f|_{\cl S_A\otimes_{\rm c} \cl S_B} : f\in \cl S\}$, let  
$$\cl C = \{S : \mbox{ model over } (\cl S_A,\cl S_B)
\mbox{ and there exists } f\in \tilde{\cl S} \mbox{ s.t. } f = f_S\}.$$
For technical simplicity, we impose the restriction that all models we consider 
have associated Hilbert spaces that are separable.

The specific algebras $\cl A$ and $\cl B$ participating in a model
$(\mbox{}_{\cl A}H_{\cl B}, \nph_A,\nph_B,\xi)$ will not play a role in the 
next result, so we will temporarily omit them from the notation. 
The models $(H,\nph_A,\nph_B,\xi)$ and $(K,\psi_A,\psi_B,\eta)$
will be called \emph{unitarily equivalent} if there exists a unitary 
operator $U : H\to K$ such that 
$$U\xi = \eta \ \ \mbox{ and } \ \ U\nph_A(a)\nph_B(b)U^* = \psi_A(a)\psi_B(b), 
\ \  a\in \cl S_A, b\in \cl S_B.$$

Given an operator system $\cl T$, a Hilbert space $H$, and a 
unital completely positive map $\nph : \cl T\to \cl B(H)$,
for $\kappa\in \bb{N}\cup \{\infty\}$, we let 
$\ell^2(\kappa) = \oplus_{i=1}^{\kappa} \bb{C}$ 
(as usual, we set $\ell^2 = \ell^2(\infty)$), 
write
$\nph\otimes 1_{\kappa} : \cl T\to \cl B(H\otimes \ell^2(\kappa))$
for the map, given by $(\nph\otimes 1_{\kappa})(u) = \nph(u)\otimes I_{\ell^2(\kappa)}$ and let
$\nph^{(\infty)} = \nph\otimes 1_{\infty}$.

The following theorem is a version of \cite[Theorem 7.5]{pszz} in 
the case of arbitrary bipartite quantum systems. 
For a convex set $C$, we denote by ${\rm Ext}(C)$ the 
set of extreme points of $C$.

\begin{theorem}\label{th_charallqc}
Let
$f\in {\rm Ext}(S(C^*_u(\cl S_A)\otimes_{\max} C^*_u(\cl S_B)))$. The following are equivalent: 
\begin{itemize}
\item[(i)] $f$ is an abstract self-test for $\cl S$;

\item[(ii)] there exists a model 
$\tilde{S} = (\tilde{H},\tilde{\nph}_A, \tilde{\nph}_B,\tilde{\xi})$
of $f$ such that, for every model $S$ of $f$, there exists a unit vector 
$\xi_{\rm aux}\in \ell^2$ such that $S$ is unitarily equivalent to the model
$(\tilde{H}\otimes\ell^2,\tilde{\nph}_A^{(\infty)}, \tilde{\nph}_B^{(\infty)},
\tilde{\xi}\otimes \xi_{\rm aux})$.
\end{itemize}
\end{theorem}

\begin{proof}
(i)$\Rightarrow$(ii) 
We set $\cl A = C^*_u(\cl S_A)$ and $\cl B = C^*_u(\cl S_B)$ for brevity. 
Let $\tilde{f}\in\mc{S}$ be the unique extension of $f$ to a state on 
$\cl A\otimes_{\max}\cl B$, and note that 
$\tilde{f}\in {\rm Ext}(S(\cl A\otimes_{\max}\cl B))$.
Indeed, 
if $g_1, g_2 \in S(\cl A\otimes_{\max}\cl B)$ are such that 
$\tilde{f} = \lambda g_1 + (1-\lambda)g_2$ for some 
$\lambda\in (0,1)$, then 
$$f = \lambda 
g_1|_{\cl S_A\otimes_{\rm c}\cl S_B} + 
(1-\lambda)g_2|_{\cl S_A\otimes_{\rm c}\cl S_B}$$
and, by the extremallity of $f$, 
we have that 
$$f = g_1|_{\cl S_A\otimes_{\rm c}\cl S_B} = 
g_2|_{\cl S_A\otimes_{\rm c}\cl S_B}.$$
Since $f$ is an abstract self-test for $\cl S$, 
we have that $g_1 = g_2 = \tilde{f}$.

Let $(\tilde{H},\tilde{\pi},\tilde{\xi})$ be the GNS triple associated with 
$\tilde{f}$; thus, $\tilde{H}$ is a Hilbert space, 
$\tilde{\pi} : \cl A\otimes_{\max}\cl B\to \cl B(\tilde{H})$ is a 
unital *-representation and $\tilde{\xi}\in \tilde{H}$ is a unit vector, 
cyclic for $\tilde{\pi}$, such that 
\begin{equation}\label{eq_GNS0}
\tilde{f}(u) = \langle \tilde{\pi}(u)\tilde{\xi},\tilde{\xi}\rangle, \ \ \ 
u\in \cl A\otimes_{\max}\cl B.
\end{equation}
Using the inclusion $\cl S_A\otimes_{\rm c}\cl S_B\subseteq \cl A\otimes_{\max}\cl B$, 
let $\tilde{\nph}_A = \tilde{\pi}|_{\cl S_A}$, 
$\tilde{\nph}_B = \tilde{\pi}|_{\cl S_B}$, and 
$\tilde{S} = (\tilde{H},\tilde{\nph}_A, \tilde{\nph}_B,\tilde{\xi})$;
by (\ref{eq_GNS0}), $\tilde{S}$ is a model of $f$. 

Let $S = (H,\nph_A, \nph_B,\xi)$ be a model of $f$, and 
write $\pi_A$ (resp. $\pi_B$) for the unique extension of 
$\nph_A$ (resp. $\nph_B$) to a unital *-homomorphism from 
$\cl A$ (resp. $\cl B$) into $\cl B(H)$. Since the ranges of 
$\nph_A$ and $\nph_B$ commute, so do the ranges of $\pi_A$ and $\pi_B$; 
let $\pi = \pi_A\cdot \pi_B$, viewed as a unital *-representation of 
$\cl A\otimes_{\max}\cl B$ on $H$. 
Since $f$ is an abstract self-test for $\cl S$, we have that 
$\tilde{f}(u) = \langle \pi(u)\xi,\xi\rangle$,
$u\in \cl A\otimes_{\max}\cl B$. 
Write $H = \oplus_{i\in \bb{N}} H_i$ and $\xi = (\xi_i)_{i\in \bb{N}}$, 
where $\xi_i\in H_i$, $i\in \bb{N}$, 
$H_i$ is invariant for $\pi$, and $\xi_i$ is a cyclic vector 
for the representation $\pi_i := \pi|_{H_i}$. 
Let $\lambda_i = \|\xi_i\|^2$, $i\in \bb{N}$; thus, 
$\sum_{i=1}^{\infty} \lambda_i = 1$.  
Setting $\eta_i = \frac{1}{\sqrt{\lambda_i}}\xi_i$, we have that 
$$\tilde{f}(u) = \sum_{i=1}^{\infty} \lambda_i \langle \pi_i(u)\eta_i,\eta_i\rangle, 
\ \ u\in \cl A\otimes_{\max}\cl B.$$
Since $\tilde{f}$ is an extreme point of $S(\cl A\otimes_{\max}\cl B)$, 
we have that 
$$\tilde{f}(u) = \langle \pi_i(u)\eta_i,\eta_i\rangle, 
\ \ u\in \cl A\otimes_{\max}\cl B, i\in \bb{N}.$$
Since the vector $\eta_i$ is cyclic for $\pi_i$, there exists a 
unitary operator $U_i : H_i \to \tilde{H}$, such that 
\begin{equation}\label{eq_Ui}
U_i \pi_i(u) U_i^* = \tilde{\pi}(u) 
\mbox{ and } U_i \eta_i = \tilde{\xi}, \ \ \ i\in \bb{N}. 
\end{equation}

Set $\xi_{\rm aux} = (\sqrt{\lambda_i})_{i\in \bb{N}}$, viewed as a (unit) vector 
in $\ell^2$. 
Write $U = \oplus_{i\in \bb{N}} U_i$; thus, $U : H\to \oplus_{i\in \bb{N}}\tilde{H}$
is a unitary operator. Identifying $\oplus_{i\in \bb{N}}\tilde{H}$ with 
$\tilde{H}\otimes\ell^2$ canonically, 
we have that 
$$U\xi = U\left((\sqrt{\lambda_i}\eta_i)_{i\in \bb{N}}\right) = \left(\sqrt{\lambda_i}U_i\eta_i\right)_{i\in \bb{N}}
= \left(\sqrt{\lambda_i}\tilde{\xi}\right)_{i\in \bb{N}} = \tilde{\xi}\otimes\xi_{\rm aux}.$$
Equation (\ref{eq_Ui}) now implies that 
$S$ is unitarily equivalent, via $U$, to the model
$(\tilde{H}\otimes\ell^2,\tilde{\nph}_A^{(\infty)}, \tilde{\nph}_B^{(\infty)},
\tilde{\xi}\otimes \xi_{\rm aux})$.

(ii)$\Rightarrow$(i) 
In the notation of (ii), let 
$\tilde{f} : \cl A\otimes_{\max}\cl B\to \bb{C}$ be the state, 
given by (\ref{eq_GNS0}). 
Let $g \in \cl S$ be an extension of $f$, 
$(H,\pi,\xi)$ be the GNS triple associated with $g$, 
and $\nph_A$ (resp. $\nph_B$) be the restrictions of $\pi$ to 
$\cl S_A$ (resp. $\cl S_B$). 
By assumption, there exists a unit vector $\xi_{\rm aux}\in \ell^2$
and a unitary operator $U : H\to \tilde{H}\otimes\ell^2$, such that 
$U\xi = \tilde{\xi}\otimes\xi_{\rm aux}$, and 
\begin{equation}\label{eq_abI}
U\nph_A(a)\nph_B(b)U^* = \tilde{\pi}(a\otimes b)\otimes I, \ \ \ 
a\in \cl S_A, b\in \cl S_B.
\end{equation}
Let $\pi_A$ (resp. $\pi_B$) be the unique extension of $\nph_A$ (resp. $\nph_B$)
to a unital *-representation of $\cl A$ (resp. $\cl B$). 
We have that 
$$\pi(a\otimes b) = (\pi_A\cdot \pi_B)(a\otimes b), \ \ \ a\in \cl S_A, b\in \cl S_B;$$
since the elementary tensors of the form $a\otimes b$, $a\in \cl S_A$, $b\in \cl S_B$, 
generate $\cl A\otimes_{\max}\cl B$ as a C*-algebra, we have that 
$\pi = \pi_A\cdot \pi_B$. 
Now (\ref{eq_abI}) implies that 
$$U\pi(u)U^* = \tilde{\pi}(u)\otimes I, \ \ \ u\in \cl A\otimes_{\max}\cl B.$$
It follows that, if $u\in \cl A\otimes_{\max}\cl B$, then 
\begin{eqnarray*}
g(u) 
& = & 
\langle \pi(u)\xi,\xi\rangle 
= \langle U\pi(u)U^*U\xi,U\xi\rangle\\
& = & \langle (\tilde{\pi}(u)\otimes I)(\tilde{\xi}\otimes\xi_{\rm aux}),
\tilde{\xi}\otimes\xi_{\rm aux}\rangle
=
\langle \tilde{\pi}(u)\tilde{\xi},\tilde{\xi}\rangle;
\end{eqnarray*}
Thus, $g = \tilde{f}$, and the proof is complete. 
\end{proof}

\begin{remark}\label{r_locality0}
\rm
In the notation of Theorem \ref{th_charallqc}, 
suppose that a state $f\in \tilde{\cl S}$ is an abstract self-test for $\cl S$. 
By Theorem \ref{th_charallqc}, $f$ has a model $\tilde{S}$, such that every other model of $f$
is unitarily equivalent to an ampliation of $\tilde{S}$. 
Let 
\begin{equation}\label{eq_forthin}
\frak M = 
\left\{S=(_{\cl A}H_{\cl B}, \varphi_A, \varphi_B,\xi) : f_S = f, 
\cl A=(\varphi_A(S_A))'', \cl B = \left((\varphi_B(S_B))''\right)^o\right\}.
\end{equation}
The proof of Theorem \ref{th_charallqc}
shows that the unitary operator $U$, constructed therein, 
implements the dilation relation 
$_{\cl A}H_{\cl B}\leq _{\tilde{\cl A}\otimes I}(\tilde H\otimes\ell^2)_{\tilde{\cl B}\otimes I}$. 
Thus, the problem in reversing the implication established in Theorem \ref{th_stisabs} resides in allowing the use of
models, whose algebras of observables are more general than the 
ones indicated in (\ref{eq_forthin}). Indeed, often the interest lies in 
\emph{Haag models}, where $\cl B^o = \cl A'$.
We next exhibit cases where such a reversal can be achieved more generally. 
Our next result, 
Theorem \ref{th_rev},
is an extension of \cite[Theorem 4.12]{pszz}.
\end{remark}

Let $\cl A$ be a $C^*$-algebra. 
Recall that a representation $\pi : \cl A\to \cl B(K)$
is called a type I representation if 
the bicommutant $\pi(\cl A)''$ in 
$\cl B(K)$ is a von Neumann algebra of type I.  
Let $\{ H(\gamma): \gamma\in \Gamma\}$ be a measurable field of Hilbert spaces on a Borel space $\Gamma$ equipped with a $\sigma$-finite Borel measure $\mu$. Suppose that, to each $\gamma\in \Gamma$, there corresponds a representation $\pi_\gamma$ of $\cl A$ on $H(\gamma)$ such that for every $a\in \cl A$, the operator field $\gamma\mapsto\pi_\gamma(a)\in\cl B(H(\gamma))$ is measurable, in which case we say that the field $\{\pi_\gamma\}_{\gamma\in\Gamma}$ is a measurable field of representations. Set $H =\int_\Gamma^{\oplus}H(\gamma)d\mu(\gamma)$ and 
$\pi(a) = \int_\Gamma^{\oplus}\pi_\gamma(a)d\mu(\gamma)\in \cl B(H)$. Then $\pi$ is a representation of $\cl A$ on $H$, called the direct integral of $\{\pi_\gamma\}_{\gamma\in\Gamma}$ and written 
$\pi = \int_\Gamma^{\oplus}\pi_\gamma d\mu(\gamma)$.
We refer the reader to \cite[\S IV.8]{takesaki} for further details concerning measurable fields of Hilbert spaces and operators. 
 By \cite{dixmier}, if $\pi$ is a representation of type I on a separable Hilbert space, there is a standard $\sigma$-finite measure space $(\Gamma,\mu)$ a measurable field of irreducible representations $\{\pi_\gamma\}_{\gamma\in\Gamma}$, measurable function $\gamma\mapsto n(\gamma)\in \mathbb N_{\infty}=\{1,2,\ldots,\infty\}$ such that 
$$\pi\simeq\int^\oplus_\Gamma\pi_\gamma\otimes 1_{n(\gamma)}d\mu(\gamma).$$

For the remainder of this section, 
$\cl C$ will denote the class of quantum commuting models defined by the canonical bimodules $_{\cl B(H_A)}H_{\cl B(H_B)^o}= H_A\otimes H_B$ over
$(\cl B(H_A), \cl B(H_B)^o)$ and 
unital completely positive maps
$\varphi = \varphi_A\otimes \varphi_B$, where 
$\varphi_A :\cl S_A \to \cl B(H_A)$ and 
$\varphi_B :\cl S_B \to \cl B(H_B)$ 
extend to type I representations $\pi_A$ and $\pi_B$ of $\cl A:=C_u^*(\cl S_A)$ and $\cl B:=C_u^*(\cl S_B)$ on $H_A$ and $H_B$, respectively. 
We note that, if $S\in \cl C$ then $f_S$ is in fact a state 
(not only on the commuting tensor product $\cl S_A\otimes_{\rm c}\cl S_B$
but also) on the minimal tensor product 
$\cl S_A\otimes_{\min}\cl S_B$.
We say that a model is irreducible if the corresponding 
representations $\pi_A$ and $\pi_B$ are irreducible.
We note that, if a model $S\in \cl C$ is irreducible then 
$\pi_A(\cl S_A)'' = \cl B(H_A)$ and $\pi_B(\cl S_B)'' = \cl B(H_B)$.

\begin{theorem}\label{th_rev}
Let $\cl S = \{\tilde{f}_S: S\in\cl C\}$, 
let $\frak{M}\in \cl C$, and let $f = f_{\frak{M}}$. 
Assume that
$f\in {\rm Ext}(S(\cl S_A\otimes_{\rm c}\cl S_B))$. 
Suppose that $f$ has unique extension to a state on 
$C^*_u(\cl S_A)\otimes_{\min} C^*_u(\cl S_B)$.
Then $f$ is a self-test for $\cl C$ and there exists an irreducible ideal model $S\in\cl C$.
In particular, if $f$ is an abstract self-test for $\cl S$ then 
$f$ is a self-test for $\cl C$ that admits an irreducible ideal model $S\in\cl C$.
\end{theorem}

\begin{proof}
Write 
$\frak{M} = (_{\cl B(H_A)}(H_A\otimes H_B)_{\cl B(H_B)^o}, 
\nph_A, \nph_B,\xi),$
so that 
$$f(u) = \langle(\varphi_A\otimes\varphi_B)(u)\xi,\xi\rangle, 
\ \ \ u\in \cl S_A\otimes_{\min} \cl S_B.$$ 
By assumption, the *-representations $\pi_A$ and $\pi_B$ extending 
$\nph_A$ and $\nph_B$, respectively, are type I. 
Consider the direct integral decompositions of $\pi_A$ and $\pi_B$ into irreducible representations: 
  $$\pi_A=\int_X (\pi_x^A\otimes I_{m(x)})d\mu(x)
  \ \mbox{ and } \ \pi_B=\int_Y (\pi_y^A\otimes I_{n(y)})d\nu(y),$$ 
  acting on
  $$H_A=\int_X^\oplus H_A(x)\otimes\ell^2(m(x))d\mu(x) 
  \mbox{ and } H_B=\int_Y^\oplus H_B(y)\otimes\ell^2(n(y))d\nu(y),$$ respectively. 
  Let $\{e_x^i\}_{i=1}^{m(x)}$ and $\{e_y^j\}_{j=1}^{n(y)}$ be the standard bases in $\ell^2(m(x))$ and $\ell^2(n(y))$, respectively, 
  and write $\delta$ for the counting measure on $\mathbb N_\infty\times \mathbb N_{\infty}$
  (we set $\mathbb N_k=\{1,2,\ldots, k\}$, so that 
  $\bb{N}_{\infty} = \bb{N}\cup\{\infty\}$). 
  As $\xi\in H_A\otimes H_B$, there exist measurable 
  families $(\xi_{x,y}^{i,j})_{x,y}^{i,j}$, where 
  $\xi_{x,y}^{i,j} \in H_A(x)\otimes H_B(y)$, 
  $\|\xi_{x,y}^{i,j}\|=1$, and 
  $(\lambda_{x,y}^{i,j})_{x,y}^{i,j}\subseteq \mathbb C$, such that $\lambda_{x,y}^{i,j}=0$ if  $(i,j)\in (\mathbb N_{m(x)}\times\mathbb N_{n(y)})^c$,  
  $$\int_{X\times Y}\int_{ \mathbb N\times\mathbb N}|\lambda_{x,y}^{i,j}|^2d\delta(i,j)d(\mu\times\nu)(x,y)=1,$$ 
  and 
  $$\xi = \int_{X\times Y} \int_{\mathbb N\times\mathbb N}\xi_{x,y}^{i,j}\otimes\lambda_{x,y}^{i,j}(e_x^i\otimes e_{y}^j)d\delta(i,j)d(\mu\times\nu)(x,y).$$ 
  
  Write 
  $$S_{x,y}^{i,j} = 
  (H_A(x)\otimes H_B(y), \pi_x^A|_{\cl S_{A}}, \pi_y^B|_{\cl S_{B}}, \xi_{x,y}^{i,j}), \ \ \ x,y\in X, i,j\in \bb{N}_{\infty}.$$
  As $\pi_A$ and $\pi_B$ are type I representations, so are $\pi_x^A$ and $\pi_y^B$ for almost all $x$ and $y$ \cite[Proposition 8.4.8]{dixmier}.
  Thus, for almost all $x$, $y$, the model 
  $S_{x,y}^{i,j}$ belongs to the class $\cl C$. 
  Moreover,
  $$f(u) = \int_{X\times Y}\int_{\mathbb N\times\mathbb N}\langle(\pi_x^A\otimes\pi_y^B)(u)\xi_{x,y}^{i,j},\xi_{x,y}^{i,j}\rangle|\lambda_{x,y}^{i,j}|^2 d\delta(i,j)d(\mu\times\nu)(x,y). $$
  Consider the probability measure $\alpha$ on 
  $X\times Y\times\mathbb N_\infty\times\mathbb N_\infty$, given by $$\alpha(E)=\int_E|\lambda_{x,y}^{i,j}|^2d\delta(i,j)d(\mu\times\nu)(x,y),$$ 
  where $E \subseteq X\times Y\times\mathbb N_\infty\times\mathbb N_\infty$ is measurable.
  Suppose that $\alpha(E) \ne 0$ and $\alpha(E^c)\ne 0$, set
$$f_1(u)= \frac{1}{\alpha(E)}\int_E\langle(\pi_x^A\otimes\pi_y^B)(u)\xi_{x,y}^{i,j},\xi_{x,y}^{i,j}\rangle|\lambda_{x,y}^{i,j}|^2
 d\delta(i,j)d(\mu\times\nu)(x,y),$$ 
considered as a state on $C^*_u(\cl S_A)\otimes_{\min} C^*_u(\cl S_B)$, 
and let $f_2$ be defined similarly, using the set $E^c$ in the place of $E$.
Then  
$$f(u) = \alpha(E)f_1(u)+\alpha({E^c}) f_2(u), \ \ \ 
u\in \cl S_A\otimes_{\min} \cl S_B.$$ 
Since $f$ is an extreme point,  
$$f(u) = f_1(u) = f_2(u), \ \ \ u\in \cl S_A\otimes_{\min} \cl S_B.$$
It follows that 
\begin{eqnarray*}
  && 
  \int_{E}\langle(\pi_x^A\otimes\pi_y^B)(u)\xi_{x,y}^{i,j},\xi_{x,y}^{i,j}\rangle|\lambda_{x,y}^{i,j}|^2d\delta(i,j)
  d(\mu\times\nu)(x,y)\\
  &&
  \hspace{4cm} =\int_Ef(u)|\lambda_{x,y}^{i,j}|^2d\delta(i,j)
  d(\mu\times\nu)(x,y).
  \end{eqnarray*}
  As the equality trivially holds for $E$ of full measure, we obtain that
\begin{equation}\label{eq_f(u)=}
f(u) = \left\langle(\pi_x^A\otimes\pi_y^B)(u)\xi_{x,y}^{i,j},\xi_{x,y}^{i,j}\right\rangle 
  \ \ \ \alpha\mbox{-almost everywhere},
\end{equation}  
  showing in particular that there is an irreducible model in $\cl C$ 
  that gives rise to $f$. 

  Fix irreducible representations 
  $\tilde \pi^A$ and $\tilde \pi^B$ of 
  $C^*_u(\cl S_A)$ and $C^*_u(\cl S_B)$, acting on Hilbert spaces
$\tilde H_A$ and $\tilde H_B$, respectively, 
and a unit vector $\tilde \xi\in\tilde H_A\otimes\tilde H_B$, 
such that 
$f(u) = \langle (\tilde\pi^A\otimes\tilde \pi^B)(u)\tilde\xi,\tilde \xi\rangle$, $u\in \cl S_A\otimes_{\min}\cl S_B$,   
and  let 
$$N = \{(x,y,i,j)\in X\times Y\times \mathbb N_\infty\times \mathbb N_\infty 
 : (\ref{eq_f(u)=}) \mbox{ holds}\}.$$ 
Then $\alpha(N^c)=0$.
  Since $f$ has a unique extension to 
  $C^*_u(\cl S_A)\otimes_{\min} C^*_u(\cl S_B)$, 
  we have that 
  $$\left\langle (\tilde \pi^A\otimes\tilde \pi^B)(u)\tilde \xi,\tilde\xi\right\rangle 
  = \left\langle (\pi_{x}^A\otimes\pi_{y}^B)(u)\xi_{x,y}^{i,j},\xi_{x,y}^{i,j}\right\rangle,  
  \ \ \ (x,y,i,j)\in N,$$ 
  for every $u \in C^*_u(\cl S_A)\otimes_{\min} C^*_u(\cl S_B)$. 
  As $\pi_x^A$ and $\pi_y^B$ are irreducible, so is $\pi_x^A\otimes\pi_y^B$ and hence $(\pi_x^A\otimes\pi_y^B, \xi_{x,y}^{i,j})$ is a GNS representation for the state $f$, 
  $(x,y,i,j)\in N$. We have, in particular,  
  $\xi_{x,y}^{i,j}=\alpha_{x,y}^{i,j,i',j'}\xi_{x,y}^{i',j'}$ for some $\alpha_{x,y}^{i,j,i',j'}\in\mathbb T$. 
  
  Let $\leq$ be the lexicographic order on $\mathbb N_\infty\times\mathbb N_\infty$, that is, 
  $(n_1,n_2)<(m_1,m_2)$ if $n_1<m_1$ or $n_1=m_1$ and $n_2<m_2$. 
  Define
  $\tau: X\times Y\to \mathbb N_\infty\times\mathbb N_\infty\cup\{\infty\}$ by 
  letting $\tau(\omega)=\min\{(i,j): (\omega,i,j)\in N\}$, for $\omega=(x,y)\in X\times Y$, where we have set 
  $\min\emptyset = \infty$.  Let $\pi_{X\times Y}: (X\times Y)\times (\mathbb N_\infty\times\mathbb N_\infty)\to X\times Y$ be the projection map. Clearly, $(\omega,\tau(\omega))\in N$ for every $\omega\in\pi_{X\times Y}(N)$. 
  We claim that $\tau$ is measurable. 
  Indeed, 
  writing $N_{i,j}$ for the slice of $N$ along $(i,j)\in \mathbb N_\infty\times\mathbb N_\infty$, 
  note that $N\cap ((X\times Y)\times\{(i,j)\})=N_{i,j}\times\{(i,j)\}$ and hence $N_{i,j}$ is measurable. This shows that 
  $\tau^{-1}(\{(i,j)\})=N_{i,j}\setminus\cup_{(i',j')< (i,j)} N_{i',j'}$ is measurable.  
  In addition, the set 
  $\pi_{X\times Y}(N)=\cup_{i,j}N_{i,j}$ is measurable. 
Set $\zeta_{x,y} = \xi_{x,y}^{\tau(x,y)}\in H_A(x)\otimes H_B(y)$. 
For $(x,y)\in X\times X$, let 
$$\tilde{f}_{x,y} = \sum_{(i,j) : (x,y,i,j)\in N} \alpha_{x,y}^{(i,j),\tau(x,y)}\lambda_{x,y}^{i,j}(e_x^i\otimes e_y^j).$$
Then 
\begin{eqnarray*}
  \xi
  & = & 
  \int_N\xi_{x,y}^{i,j}\otimes\lambda_{x,y}^{i,j}(e_x^i\otimes e_y^j)d\delta(i,j)d(\mu\times\nu)(x,y)\\
  & = & 
  \int_{\pi_{X\times Y}(N)} 
\zeta_{x,y}\otimes \tilde{f}_{x,y}
 d(\mu\times\nu)(x,y).
 \end{eqnarray*}  
Let $r_{x,y} = \|\tilde{f}_{x,y}\|$. 
If $\tilde{f}_{x,y}\neq 0$, set 
$f_{x,y} = \frac{\tilde{f}_{x,y}}{\|\tilde{f}_{x,y}\|}$; otherwise, 
let $f_{x,y} = 0$. 
We thus have that 
$\int_{X\times Y}|r_{x,y}|^2 d(\mu\times\nu)(x,y)=1$ and 
  $\{f_{x,y}\}_{(x,y)\in X\times Y}$
  is a measurable field of unit vectors in $\ell^2(m(x))\otimes\ell^2(n(y))$
  with 
  $$r_{x,y}f_{x,y}=\int \alpha_{x,y}^{(i,j),\tau(x,y)}\lambda_{x,y}^{i,j}(e_x^i\otimes e_y^j)d\delta(i,j).$$
  
  Consider now  the set $\Lambda=\{(x,y)\in\pi_{X\times Y}(N): r_{x,y}\ne 0\}$. Then for $(x,y)\in\Lambda$, $(\pi_x^A\otimes\pi_y^B, \zeta_{x,y})$ is a GNS representation of $f$. 
Since $\pi_A$ and $\pi_B$ are irreducible, we obtain that $\pi_x^A\otimes\pi_y^B\sim\tilde\pi^A\otimes\tilde\pi^B$  
  and hence $\pi_x^A\sim \tilde\pi^A$ and $\pi_y^B\sim\tilde\pi^B$ 
  whenever $(x,y)\in \Lambda$ (we use the symbol $\sim$ to denote unitary 
  equivalence). Let $\Lambda_A=\pi_X(\Lambda)$ and $\Lambda_B=\pi_Y(\Lambda)$, where $\pi_X$ and $\pi_Y$ are the corresponding projections in the Cartesian product $X\times Y$. We have that $\Lambda_A$ and $\Lambda_B$ are analytic sets and there exist subsets $M_A\subseteq \Lambda_A$ and $M_B\subseteq \Lambda_B$, such that $\mu(M_A)=\nu(M_B)=0$ and $\wt\Lambda_A:=\Lambda_A\setminus M_A$ and $\wt\Lambda_B:=\Lambda_B\setminus M_B$ are measurable (see \cite[Appendix]{takesaki}). 
  By \cite[p. 166, Lemme 2]{dixmier_von_neumann}, 
  there exist measurable $U_x: H_A(x)\to \tilde H_A$ and $U_y: H_B(y)\to \tilde H_B$ such that
$$U_x\pi_x^A(a)U_x^*=\tilde\pi^A(a) \ \mbox{ and } \ 
U_y\pi_y^B(b)U_y^*=\tilde\pi^B(b), \ \ 
a\in C^*_u(\cl S_A), b\in C^*_u(\cl S_B).$$ 
Then $(U_x\otimes U_y)\zeta_{x,y}=\beta_{x,y}\tilde\xi$ for $\beta_{x,y}\in\mathbb T$.
Hence $\xi=\tilde\xi\otimes\psi_{\rm aux}$, where 
$$\psi_{\rm aux} = \int_{\wt\Lambda_A \times \wt\Lambda_B} r_{x,y}\beta_{x,y}f_{x,y}d(\mu\times \nu)(x,y).$$

For $x\in\wt\Lambda_A$ let $V_x: H_A(x)\otimes\ell^2(m(x))\to\tilde H_A\otimes\ell^2(m(x))$ be given by $V_x = U_x\otimes 1_{m(x)}$; if $x\not\in \wt\Lambda_A$ let $V_x: H_A(x)\otimes\ell^2(m(x))\to \tilde H_A\otimes (H_A(x)\otimes\ell^2(m(x)))$ be given by 
$V_x(v) = w\otimes v$ for a fixed $w\in\tilde H_A$, and set $V_A=\int_X^\oplus V_x d\mu(x)$. 
Define an isometry $V_B$ in a similar way. 
Let 
$$H_A^{\rm aux}=\int_{\wt\Lambda_A}^{\oplus}\ell^2(m(x))d\mu(x)\oplus \int_{\wt\Lambda_A^c}^\oplus H_A(x)\otimes\ell^2(m(x))d\mu(x)$$ 
and 
$$H_B^{\rm aux}=\int_{\wt\Lambda_B}^{\oplus}\ell^2(n(y))d\nu(y)\oplus \int_{\wt\Lambda_B^c}^{\oplus}H_B(y)\otimes \ell^2(n(y))d\nu(y).$$ 
Then
$\psi_{\rm aux}\in H_A^{\rm aux}\otimes H_B^{\rm aux}$ and
$$(V_A\otimes V_B)(\pi_A\otimes\pi_B)(a\otimes b)\xi=(\tilde\pi_A\otimes\tilde\pi_B)(a\otimes b)(\tilde\xi\otimes \psi_{\rm aux})$$ for $a\in C^*_u(\cl S_A)$, $b\in C^*_u(\cl S_B).$    
\end{proof}


\section{Applications and examples}\label{s_app}

In this section, we apply the general operator system framework developed in the previous sections to several special cases, including those of QNS 
correlations, quantum graph homomorphisms, synchronous correlations
and positive definite functions defined on groups. 
The special cases we consider are based at pairs $(\cl S_A,\cl S_B)$ 
of finitely generated operator systems, say 
$$\cl S_A = {\rm span}\{e_1,\dots,e_k\} \ \mbox{ and } \ 
\cl S_B = {\rm span}\{f_1,\dots,f_l\},$$
so that the pair $(\nph_A,\nph_B)$ of unital completely positive maps, 
where $\nph_A : \cl S_A\to \cl B(H)$ and $\nph_B : \cl S_B\to \cl B(H)$, 
is determined by the mutually commuting 
families $(E_i)_{i=1}^k$ and $(F_j)_{j=1}^l$ of 
operators on $H$ vis the assignments $\nph_A(e_i) = E_i$, $i\in [k]$ and 
$\nph_B(f_j) = F_j$, $j\in [l]$. Thus, 
we will consider a commuting operator model over $(\cl S_A,\cl S_B)$ 
as a tuple
$S = (H,(E_i)_{i=1}^k,(F_j)_{j=1}^l,\xi)$, where $\xi\in H$ is a unit vector. 
The tuple $S$ gives rise to the 
\emph{correlation} $p_S : [k]\times [l]\to \bb{C}$, given by 
\begin{equation}\label{eq_pij}
p_S(i,j) = \langle E_i F_j\xi,\xi\rangle, \ \ \ i\in [k],j\in [l];
\end{equation}
we say that $S$ is a \emph{model} of $p_S$. 
The correlations of the form $p_S$ correspond precisely to 
states $s : \cl S_A\otimes_{\rm c}\cl S_B\to \bb{C}$ via 
the assignment $s(e_i\otimes f_j) = p_S(i,j)$, $i\in [k]$, $j\in [l]$. 

A tuple $\tilde{S} = (\tilde{H},(\tilde{E}_i)_{i=1}^k,(\tilde{F}_j)_{j=1}^l,\tilde{\xi})$ 
is an \emph{ideal model} of a correlation $p : [k]\times [l]\to \bb{C}$ if 
$p = p_{\tilde{S}}$ and, whenever 
$S = (H,(E_i)_{i=1}^k,(F_j)_{j=1}^l,\xi)$ is a model of $p$ then there exists 
a Hilbert space $H_{\rm aux}$, a unit vector $\xi_{\rm aux}\in H_{\rm aux}$ 
and a local isometry $V : H\to \tilde{H}\otimes H_{\rm aux}$ such that 
$$V E_i F_j \xi = \tilde{E}_i \tilde{F}_j \tilde{\xi} \otimes \xi_{\rm aux}, 
\ \ \ i\in [k], j\in [l].$$

The framework of self-testing described above will be referred to as 
\emph{finitary}; we will refer to the pair $(\cl S_A,\cl S_B)$ as a 
\emph{finitary context}.
Quantum models of correlations $p : [k]\times [l]\to \bb{C}$
are similarly described in the finitary framework by replacing 
the operator products $E_iF_j$ in (\ref{eq_pij}) by tensor 
products of finite dimensionally acting families $(E_i)_{i=1}^k$ and 
$(F_j)_{j=1}^l$.


\subsection{Self-testing for QNS correlations}\label{ss_QNS}

Let $X$ and $A$ be finite sets and $H$ be a Hilbert space. 
A \emph{quantum channel} from $M_X$ into $M_A$ is a completely positive trace preserving map
$\Phi : M_X\to M_A$.
A \emph{stochastic operator matrix (SOM)} over $(X,A)$ acting on $H$
is a positive block operator operator matrix $E = (E_{x,x',a,a'})_{x,x',a,a'}$, 
where $E_{x,x',a,a'}\in \cl B(H)$ for all $x,x'\in X$ and all $a,a'\in A$, 
such that ${\rm Tr}_A E = I_X\otimes I_H$
(as usual, ${\rm Tr}_A$ denotes the partial trace along 
$M_A$). 
It was shown in \cite{tt-QNS} that there exists a unital C*-algebra 
$\cl C_{X,A}$, generated by elements $e_{x,x',a,a'}$, $x,x'\in X$, $a,a'\in A$,
such that the matrix $(e_{x,x',a,a'})_{x,x',a,a'}$ is positive
as an element of $M_X\otimes M_A\otimes \cl C_{X,A}$, 
$$ \sum_{a\in A} e_{x,x',a,a} = \delta_{x,x'} 1, \ \ \ x,x'\in X,$$
and possessing the universal property that for every stochastic operator matrix 
$E = (E_{x,x',a,a'})_{x,x',a,a'}$, acting on a Hilbert space $H$, there exists a unique 
*-homomorphism $\pi_E : \cl C_{X,A}\to \cl B(H)$, such that $\pi_E(e_{x,x',a,a'}) = E_{x,x',a,a'}$, 
$x,x'\in X$, $a,a'\in A$.
Let 
$$ \cl T_{X,A} = {\rm span}\{e_{x,x',a,a'} : x,x'\in X, a,a'\in A\},$$
viewed as an operator subsystem of $\cl C_{X,A}$. 
By
\cite[Corollaries 5.3 and 5.4]{tt-QNS}, $C^*_u(\cl T_{X,A}) = \cl C_{X,A}$ and
the stochastic operator matrices $(E_{x,x',a,a'})_{x,x',a,a'}$
are in one-to-one correspondence with
the unital completely positive maps 
$\phi_E : \cl T_{X,A}\to \cl B(H)$
via the assignment $\phi_E(e_{x,x',a,a'}) =  E_{x,x',a,a'}$. 

Letting $Y$ and $B$ be further finite sets, 
the pair $(\cl T_{X,A},\cl T_{Y,B})$ of operator systems 
determines a finitary framework for self-testing. 
A commuting operator model 
$S = (H,(E_{x,x',a,a'}), (F_{y,y',b,b'}), \xi)$
(referred to later as a \emph{SOM qc-model})
gives rise to a correlation $p_S$ via (\ref{eq_pij}) which, 
in its own turn, determines a
linear map $\Gamma_S : M_{XY}\to M_{AB}$, given by 
\begin{equation}\label{eq_EFp}
\Gamma(\epsilon_{x,x'} \otimes \epsilon_{y,y'}) = \sum_{a,a'\in A} \sum_{b,b'\in B}
\left\langle E_{x,x',a,a'}F_{y,y',b,b'}\xi,\xi \right\rangle \epsilon_{a,a'} \otimes \epsilon_{b,b'}, 
\end{equation}
for all $x,x' \in X$ and all $y,y' \in Y$.
The map $\Gamma = \Gamma_S$ is a 
\emph{quantum no-signalling (QNS) correlation} over $(X,Y,A,B)$ 
\cite{dw} in that 
\begin{equation}\label{eq_NSQne1}
\Tr\hspace{-0.07cm}\mbox{}_A\Gamma(\rho_X\otimes \rho_Y) = 0 \ \mbox{ whenever } \Tr(\rho_X) = 0
\end{equation}
and
\begin{equation}\label{eq_NSQne2}
\Tr\hspace{-0.07cm}\mbox{}_B\Gamma(\rho_X\otimes \rho_Y) = 0 \ \mbox{ whenever } \Tr(\rho_Y) = 0.
\end{equation}
QNS correlations $\Gamma : M_{XY}\to M_{AB}$ admitting a 
representation of the form (\ref{eq_EFp}) are said to be of 
\emph{quantum commuting type} \cite{tt-QNS}.
One defines QNS correlations of \emph{quantum type}
by replacing the operator product in (\ref{eq_EFp}) by 
tensor products of finite dimensionally acting SOM's. 
We write $\cl Q_{\rm qc}$ (resp. $\cl Q_{\rm q}$) for the (convex) set of all quantum commuting (resp. quantum) QNS correlations, and note the inclusion 
$\cl Q_{\rm q}\subseteq \cl Q_{\rm qc}$.
We refer to $\Gamma_S$ being an self-test 
(resp. abstract self-test) if the corresponding 
correlation $p_S$ is a self-test (resp. an abstract self-test). 
We distinguish between 
\emph{qc-self-tests} (self-tests among QNS correlations of 
quantum commuting type) and \emph{q-self-tests}
(self-tests among QNS correlations of quantum type).


Given a set $\cl M$ of quantum commuting models (resp. 
quantum models) for the finitary context 
$(\cl T_{X,A},\cl T_{Y,B})$, let $\tilde{\cl M} = \{\Gamma_S : S\in \cl M\}$. 
It is clear from the preceding discussion that 
if  
$\cl S_{\cl M} = \{s_{\Gamma_S} : S\in \cl M\}$ then a quantum commuting QNS correlation $\Gamma\in \tilde{\cl M}$
is a (abstract) qc-self-test for $\cl M$ if and only if 
its corresponding state 
$s_{\Gamma} : \cl T_{X,A}\otimes_{\rm c}\cl T_{Y,B}\to \bb{C}$ is an (abstract) self-test for 
$\tilde{\cl M}$. 
The abstract q-self-tests have a convenient characterisation, as follows.

\begin{proposition}\label{p_cqcste}
Let $\cl F$ be the set of quantum commuting models whose 
underlying Hilbert space is finite dimensional. 
Then $\Gamma\in \tilde{\cl F}$
is an abstract self-test for $\cl Q_{\rm q}$ if and only if 
$s_{\Gamma}$ is an abstract self-test for $\cl S_{\cl F}$.
\end{proposition}

\begin{proof}
Let $S = (H,(E_{x,x',a,a'}), (F_{y,y',b,b'}), \xi)$ 
be a quantum commuting model with $H$ finite dimensional. 
We will show that, up to unitary equivalence, there exists a 
model of the form
\begin{equation}\label{eq_S'mo}
S' = (H_A\otimes H_B, (E'_{x,x',a,a'}\otimes I_{H_B}), 
I_{H_A}\otimes F'_{y,y',b,b'},\xi'),
\end{equation}
such that $S\preceq S'$. 

Let $\pi_A : \cl C_{X,A}\to \cl B(H)$ and 
$\pi_B : \cl C_{Y,B}\to \cl B(H)$ be the *-representations, 
determined by the SOM's $(E_{x,x',a,a'})$ and $(F_{y,y',b,b'})$, 
respectively. 
Let $\cl A = \pi_A(\cl C_{X,A})$ and $\cl B = \pi_B(\cl C_{Y,B})$. 
Since $\cl A$ is a unital finite dimensional C*-subalgebra of $\cl B(H)$, 
up to unitary equivalence, 
$H = \oplus_{i=1}^N H_i\otimes K_i$, for some (finite dimensional) Hilbert 
spaces $H_i$ and $K_i$, $i\in [N]$, and 
$\pi_A(u) = \oplus_{i=1}^N \pi_i(u)\otimes I_{k_i}$, where $k_i = {\rm dim}(K_i)$, 
and $\pi_i$ are inequivalent irreducible representations, $i\in [N]$. 
Since $\cl B$ commutes with $\cl A$, 
we have that $\pi_B(v) = \oplus_{i=1}^N I_{n_i}\otimes \rho_i(v)$, where 
$n_i = {\rm dim}(H_i)$ and $\rho_i : \cl C_{Y,B}\to \cl B(K_i)$ is a 
unital *-representation, $i\in [N]$. 
Set 
$$\hat{E}^{(i)}_{x,x',a,a'} = \pi_i(e_{x,x',a,a'}) \ 
\mbox{ and } \ \hat{F}^{(i)}_{y,y',b,b'} = \rho_i(f_{y,y',b,b'});
$$
thus, $\hat{E}^{(i)} := \left(\hat{E}^{(i)}_{x,x',a,a'}\right)$ 
and $\hat{F}^{(i)} := \left(\hat{F}^{(i)}_{y,y',b,b'}\right)$ are stochastic 
operator matrices, giving rise to QNS correlations, say, $\Gamma_i$, 
$i\in [N]$. 
Assuming that $\xi = (\xi_i)_{i=1}^N$ and writing 
$$S^{(i)} = (H_i\otimes K_i, \hat{E}^{(i)}, \hat{F}^{(i)}, \xi_i/\|\xi_i\|),
\ \ \ i\in [N],$$
it follows that $s_{\Gamma}$ is a convex combination of 
the family $\{s_{\Gamma_i}\}_{i=1}^N$ and hence is in $\cl Q_{\rm q}$. 
It is clear that $S\preceq S'$. 
\end{proof}

\subsection{POVM self-testing}\label{ss_POVMNS}

In this subsection we show how POVM self-testing 
considered in \cite{pszz} fits into the general framework 
of Section \ref{s_models}. 
We start by introducing the relevant finitary context. 

Let $X$ and $A$ be finite sets. 
The C*-algebra $\cl A_{\rm POVM}$ was introduced in \cite{pszz} as the 
universal C*-algebra of a family of POVM's over the set $A$ with $|X|$ elements, that is,
the unital C*-algebra generated by positive elements $\tilde{e}_{x,a}$, $x\in X$, $a\in A$, 
satisfying the relations $\sum_{a\in A} \tilde{e}_{x,a} = 1$, $x\in X$, such that 
whenever $(P_{x,a})_{a\in A}$ is a POVM acting on the Hilbert space $H$, $x\in X$,
there exists a unique *-representation $\pi : \cl A_{\rm POVM}\to \cl B(H)$ such that 
$\pi(\tilde{e}_{x,a}) = P_{x,a}$, $x\in X$, $a\in A$. 
In the next proposition, we identify a concrete description of $\cl A_{\rm POVM}$. 
Recall the C*-algebra $\cl C_{X,A}$ from Subsection \ref{ss_QNS} and 
let $\tilde{\cl A}_{X,A}$ be its C*-subalgebra, generated by the elements 
$e_{x,x,a,a}$, $x\in X$, $a\in A$. 
Let 
$$\tilde{\cl S}_{X,A} = {\rm span}\{e_{x,x,a,a} : x\in X, a\in A\},$$
viewed as an operator subsystem of $\tilde{\cl A}_{X,A}$. 

Further, let $\cl A_{X,A}$ be the universal C*-algebra, generated by projections $e_{x,a}$, $x\in X$, $a\in A$, 
satisfying the relations $\sum_{a\in A} e_{x,a} = 1$, $x\in X$;
the C*-algebra $\cl A_{X,A}$ satisfies the analogous universal 
property to the one described in the previous paragrph for 
$\tilde{\cl A}_{X,A}$ but with $(P_{x,a})_{a\in A}$ being PVM's as 
opposed to POVM's. 
Let 
\begin{equation}\label{eq_defSXA0}
\cl S_{X,A} = {\rm span}\{e_{x,a} : x\in X, a\in A\}, 
\end{equation}
viewed as an operator subsystem of $\cl A_{X,A}$ (see e.g. \cite{lmprsstw}).

\begin{proposition}\label{p_APOVM}
\begin{itemize}
\item[(i)]
There exists a *-isomorphism $\rho : \cl A_{\rm POVM} \to \tilde{\cl A}_{X,A}$, such that 
$\rho(\tilde{e}_{x,a}) = e_{x,x,a,a}$, $x\in X$, $a\in A$. 

\item[(ii)]
The map $e_{x,a}\mapsto \tilde{e}_{x,a}$ defines a unital 
complete order isomorphism $\cl S_{X,A}\cong \tilde{\cl S}_{X,A}$.

\item[(iii)]
Up to a canonical *-isomorphism, $C^*_u(\cl S_{X,A}) = \tilde{\cl A}_{X,A}$.
\end{itemize}
\end{proposition}

\begin{proof}
(i) 
We show that the C*-algebra $\tilde{\cl A}_{X,A}$ satisfies the universal property 
of $\cl A_{\rm POVM}$. 
Clearly, $\{e_{x,x,a,a}\}_{a\in A}$ is a POVM in $\tilde{\cl A}_{X,A}$, $x\in X$. 
Suppose that $(P_{x,a})_{a\in A}$, $x\in X$, are POVM's acting on the Hilbert space $H$. 
Let $E_{x,x',a,a'} := \delta_{x,x'} P_{x,a}$, $x,x'\in X$, $a,a'\in A$; then 
$E := (E_{x,x',a,a'})_{x,x',a,a'}$ is a stochastic operator matrix and, by the universal property 
of $\cl C_{X,A}$, there exists a unital *-homomorhism $\pi : \cl C_{X,A}\to \cl B(H)$, such that 
$\pi(e_{x,x',a,a'}) = E_{x,x',a,a'}$, $x,x'\in X$, $a,a'\in A$. 
The restriction $\rho = \pi|_{\tilde{\cl A}_{X,A}}$ of $\pi$ to $\tilde{\cl A}_{X,A}$
is a *-representation with the property that $\rho(e_{x,x,a,a}) = P_{x,a}$, $x\in X$, $a\in A$. 
Since the elements $e_{x,x,a,a}$ generate $\tilde{\cl A}_{X,A}$, such a representation is unique.

(ii) 
By (i), the families $E = \{(E_{x,a})_{a\in A} : x\in X\}$ of POVM's
acting on a Hilbert space $H$ are in bijective corresponence
with the unital completely positive maps 
$\phi_E : \tilde{\cl S}_{X,A} \to \cl B(H)$ 
via the assignment $\phi_E(\tilde{e}_{x,a}) = E_{x,a}$. 
A combination of Arveson's Extension Theorem and 
Stinespring's Dilation Theorem shows that the same 
universal property holds for $\cl S_{X,A}$ (see e.g.
\cite[p. 680]{pt}). The conclusion follows.

(iii) 
Let $H$ be a Hilbert space and $\phi : \cl S_{X,A}\to \cl B(H)$ be a 
completely positive map. Then $(\phi(e_{x,a}))_{a\in A}$ is a POVM, $x\in X$.
By (i), there exists a unital *-homomoprhism $\pi : \tilde{\cl A}_{X,A}\to \cl B(H)$, such that 
$\pi(e_{x,a}) = \phi(e_{x,a})$, $x\in X$, $a\in A$. 
The proof is complete. 
\end{proof}

It was shown in \cite[Lemma 2.8]{pt} that 
$$\cl S_{X,A}\otimes_{\rm c} \cl S_{Y,B} \subseteq \cl A_{X,A}\otimes_{\max} \cl A_{Y,B}$$
as an operator subsystem. 
The next corollary complements this fact.

\begin{corollary}\label{c_cincl}
Up to a canonical complete order embedding, 
$\cl S_{X,A}\otimes_{\rm c} \cl S_{Y,B} \subseteq 
\cl T_{X,A}\otimes_{\rm c} \cl T_{Y,B}$.
\end{corollary}

\begin{proof}
Using Proposition \ref{p_APOVM} (ii), 
let $\iota_{X,A} : \cl S_{X,A}\to \cl T_{X,A}$ be the inclusion map, and 
$\gamma_{X,A} : \cl T_{X,A}\to \cl S_{X,A}$ be the unital completely positive map, 
given by 
$$\gamma_{X,A}(e_{x,x',a,a'}) = \delta_{x,x'} \delta_{a,a'} e_{x,a}, \ \ \ 
x\in X, a\in A;$$
note that 
$\gamma_{X,A}\circ\iota_{X,A} = \id\hspace{-0.05cm}\mbox{}_{\cl S_{X,A}}$.
We have that 
$\iota := \iota_{X,A}\otimes \iota_{Y,B}$ is a unital 
completely positive map from 
$\cl S_{X,A}\otimes_{\rm c} \cl S_{Y,B}$ into 
$\cl T_{X,A}\otimes_{\rm c} \cl T_{Y,B}$. We show that $\iota$ is a complete order isomorphism 
onto its range. Suppose that $u\in M_n(\cl S_{X,A}\otimes \cl S_{Y,B})$
is such that $\iota^{(n)}(u)\in M_n(\cl T_{X,A}\otimes_{\rm c} \cl T_{Y,B})^+$.
Let $H$ be a Hilbert space, and 
$\phi : \cl S_{X,A}\to \cl B(H)$ and $\psi : \cl S_{Y,B}\to \cl B(H)$
be unital completely positive maps with commuting ranges. 
Then the maps $\phi\circ \gamma_{X,A} : \cl T_{X,A}\to \cl B(H)$ and 
$\psi\circ \gamma_{Y,B} : \cl T_{Y,B}\to \cl B(H)$ are unital and completely positive, 
and have commuting ranges. It follows that 
$$
(\phi\cdot\psi)^{(n)}(u) 
= \left((\psi\circ \gamma_{Y,B})\cdot (\psi\circ \gamma_{Y,B})\right)^{(n)}
\left(\iota^{(n)}(u)\right) \in M_n(\cl B(H))^+.
$$
The proof is complete. 
\end{proof}

Fix further finite sets $Y$ and $B$. 
A quantum commuting model for the 
finatary context $(\cl S_{X,A},\cl S_{Y,B})$ 
is thus a tuple 
$S = (H, (E_{x,a})_{x,a}, (F_{y,b})_{y,b}, \xi)$,
where $H$ is a Hilbert space, $\xi\in H$ is a unit vector, 
and $(E_{x,a})_{a\in A}$ (resp. $(F_{y,b})_{b\in B}$) is a
POVM on $H$ for every $x\in X$ (resp. $y\in Y$) such that 
$E_{x,a}F_{y,b} = F_{y,b}E_{x,a}$ for all $x,y,a,b$; 
such a model will be referred to as a \emph{POVM qc-model}. 
The model $S$ gives rise, via (\ref{eq_pij}), to 
the correlation $p_S$
\emph{of quantum commuting type}, given by 
$$p_S(a,b|x,y) = \langle E_{x,a}F_{y,b}\xi,\xi\rangle,
\ \ x\in X, y\in Y,a\in A,b\in B;$$
we note that $p = p_S$ is a \emph{no-signalling correlation}
over the quadruple $(X,Y,A,B)$ in that 
$(p(a,b|x,y))_{a,b}$ is a probability distribution over $A\times B$ for every 
$(x,y)\in X\times Y$, and
$$\sum_{b\in B} p(a,b|x,y) = \sum_{b\in B} p(a,b|x,y'), \ \ x\in X, y,y'\in Y, a\in A,$$
and 
$$\sum_{a\in A} p(a,b|x,y) = \sum_{a\in A} p(a,b|x',y), \ \ x,x'\in X,  y\in Y, b\in B$$
(see e.g. \cite{lmprsstw, psstw}). 
\emph{POVM q-models} of no-signalling correlations are defined analogously, using tensor products of POVM's acting on finite dimensional 
Hilbert spaces. 
We denote the (convex) set of all NS correlations by $\cl C_{\rm ns}$.
With a correlation $p \in \cl C_{\rm ns}$, we associate the classical information channel 
$\cl N_p : \cl D_{XY}\to \cl D_{AB}$, given by 
\begin{equation}\label{eq_Gammap}
\cl N_p(\epsilon_{x,x}\otimes\epsilon_{y,y}) = \sum_{a\in A}\sum_{b\in B} p(a,b|x,y)
\epsilon_{a,a}\otimes\epsilon_{b,b},
\end{equation}
and the quantum information channel 
$\Gamma_p : M_{XY}\to M_{AB}$, given by 
$\Gamma_p = \iota_{AB} \circ \cl N_p\circ \Delta_{XY}$,
where $\iota_{AB} : \cl D_{AB}\to M_{AB}$ is the 
inclusion map and $\Delta_{XY} : M_{XY}\to \cl D_{XY}$
is the canonical conditional expectation;
it can be easily verified that $\Gamma_p$ is a QNS correlation. 
We let 
$\cl C_{\rm qc} = \{p_S : S\mbox{ a POVM qc-model}\}$ 
(resp. 
$\cl C_{\rm q} =$ $\{p_S : S\mbox{ a POVM}$ $\mbox{q-model}\}$)
be the set of \emph{quantum commuting}
(resp. \emph{quantum}) NS correlations, and
note the inclusions
$\cl C_{\rm q}\subseteq \cl C_{\rm qc}\subseteq \cl C_{\rm ns}$.



In view of Corollary \ref{c_cincl}, a state 
$f : \cl T_{X,A}\otimes_{\rm c}\cl T_{Y,B}\to \bb{C}$
gives rise, via restriction, to a state 
$f_{\rm cl} : \cl S_{X,A}\otimes_{\rm c} \cl S_{Y,B}\to \bb{C}$. 
In the reverse direction, a state
$g : \cl S_{X,A}\otimes_{\rm c} \cl S_{Y,B}\to \bb{C}$
gives rise to the state 
$g_{\rm q} : \cl T_{X,A}\otimes_{\rm c}\cl T_{Y,B}\to \bb{C}$
by letting $g_{\rm q} = g\circ (\gamma_{X,A}\otimes_{\rm c} \gamma_{Y,B})$. 

A model $S = (H,\psi_A,\psi_B,\xi)$ of the state
$f : \cl T_{X,A}\otimes_{\rm c}\cl T_{Y,B}\to \bb{C}$
for the bipartite quantum system $(\cl T_{X,A},\cl T_{Y,B})$ gives rise to 
the model 
$S_{\rm cl} = (H,\nph_A,\nph_B,\xi)$ of the state
$f_{\rm cl} : \cl S_{X,A}\otimes_{\rm c} \cl S_{Y,B}\to \bb{C}$
for the bipartite quantum system $(\cl S_{X,A},\cl S_{Y,B})$, by 
letting $\nph_A = \psi_A\circ \iota_{X,A}$ and $\nph_B = \psi_B\circ \iota_{Y,B}$. 
In the reverse direction, a model 
$N = (H,\nph_A,\nph_B,\xi)$ for a state 
$g : \cl S_{X,A}\otimes_{\rm c} \cl S_{Y,B}\to \bb{C}$, 
gives rise to the model 
$N_{\rm q} = (H,\psi_A,\psi_B,\xi)$ for the state
$g_{\rm q}$ of $\cl T_{X,A}\otimes_{\rm c} \cl T_{Y,B}$.
We note that $(N_{\rm q})_{\rm cl} = N$.

\begin{proposition}\label{p_impliesst}
Let $\frak{M}$ be a set of SOM qc-models for the finitary context
$(\cl T_{X,A},\cl T_{Y,B})$, set
$\frak{M}_{\rm cl} = \{S_{\rm cl} : S\in \frak{M}\}$ and assume that 
$(\frak{M}_{\rm cl})_{\rm q}\subseteq \frak{M}$. 
If a state $f$ of $\cl T_{X,A}\otimes_{\rm c} \cl T_{Y,B}$ is a self-test for $\frak{M}$
with the property that 
$f = f\circ \iota_{X,A}\circ\gamma_{X,A}$
then the state $f_{\rm cl}$ of $\cl S_{X,A}\otimes_{\rm c} \cl S_{Y,B}$ 
is a self-test for $\frak{M}_{\rm cl}$.
\end{proposition}

\begin{proof}
Let $\tilde{S} = (\tilde{H},\tilde{\psi}_A, \tilde{\psi}_B,\tilde{\xi})$ be an ideal model 
for the state $f$ of $\cl T_{X,A}\otimes_{\rm c} \cl T_{Y,B}$ in the set 
$\frak{M}$. 
We show that $\tilde{S}_{\rm cl}$ is an ideal model for $f_{\rm cl}$. To this end,
suppose that $S = (H,\nph_A, \nph_B,\xi)$ is a POVM qc-model of $f_{\rm cl}$ that belongs to 
$\frak{M}_{\rm cl}$. 
Since $f = f\circ \iota_{X,A}\circ\gamma_{X,A}$, we have that 
$S_{\rm q}$ is a model for $f$, and since
$(\frak{M}_{\rm cl})_{\rm q}\subseteq \frak{M}$, it belongs to $\frak{M}$. 
Thus, $S_{\rm q}\preceq \tilde{S}$. 
If $V$ is a local isometry that implement the latter relation, 
we have that $V$ implements the relation
$(S_{\rm q})_{\rm cl}\preceq \tilde{S}_{\rm cl}$, that is, 
$S\preceq \tilde{S}_{\rm cl}$. 
\end{proof}

\begin{remark}\label{r_POVMST}
\rm 
A special case of POVM self-testing is PVM self-testing \cite{pszz}. 
In the latter setup, 
the pool of models of a given no-signalling correlation 
of quantum commuting type is restricted in that the 
participating measurements are PVM's. 
We note that PVM self-testing also fits in our general framework.
Indeed, here the families of states $s : C^*_u(\cl S_{X,A})\otimes_{\max} C^*_u(\cl S_{Y,B})$ to be selftested is restricted to 
ones that factor from the quotient map
$$C^*_u(\cl S_{X,A})\otimes_{\max} C^*_u(\cl S_{Y,B})\mapsto 
\cl A_{X,A}\otimes_{\max} \cl A_{Y,B},$$
whose existence is guaranteed by Proposition \ref{p_APOVM}.
\end{remark}


In view of Proposition \ref{p_impliesst}, it is natural to ask 
if, in general, every PVM self-test canonically gives rise to a 
SOM self-test. In the next example, we show that 
the standard self-test within the family 
$\cl C_{\rm q}$, arising from the CHSH game, does not canonically give 
rise to a self-test within the family $\cl Q_{\rm q}$. 

\begin{example}\label{CHSH_Cq_Qq}
\rm 
Recall the CHSH game \cite{chsh}; here $X=Y=A=B={\mathbb Z}_2=\{0,1\}$ 
and a quadruple $(x,y,a,b) \in X\times Y \times A\times B$ 
belongs to the support of the rule function precisely when
$xy=(a+b)\;(\text{mod } 2)$. 
Recall the Pauli matrices 
$$\sigma_x=\left(\begin{array}{cc}0&1\\1&0\end{array}\right) 
\ \mbox{ and } \ \sigma_z=\left(\begin{array}{cc}1&0\\0&-1\end{array}\right),$$
let 
$$A_0=\frac{\sigma_x+\sigma_z}{\sqrt{2}}, \ \ A_1=\frac{\sigma_x-\sigma_z}{\sqrt{2}}, \ \
B_0=\sigma_x \ \mbox{ and } \ B_1=\sigma_z,$$  
and write $(E_{x,a})_{a\in A}$ (resp. $(F_{y,b})_{b\in B}$) 
for the spectral resolution of the operator $A_x$ (resp. $B_y$). 
Write $\Omega_2=\frac{1}{\sqrt{2}}(e_0\otimes e_0+e_1\otimes e_1)$ for the maximally entangled vector in $\bb{C}^2\otimes \bb{C}^2$
(here 
$\{e_0,e_1\}$ denotes the standard basis of $\mathbb C^2$). 
It is well-known that the model 
\begin{equation}\label{eq_optimmoSt}
\wt{S} = (\mathbb C^2\otimes\mathbb C^2, (E_{x,a})_{x\in X,a\in A}, (F_{y,b})_{y\in Y,b\in B}, \Omega_2)
\end{equation}
yields a correlation $p_{\wt{S}}\in \cl C_{\rm q}$ that is 
an optimal strategy of quantum type for the CHSH game and a self-test for the class 
$\cl C_{\rm q}$ with ideal model $\tilde S$ (see e.g. \cite{supic-bowles}
and \cite{mps}).  

Let $\tilde{E} =(\delta_{x,x'}\delta_{a,a'}E_{x,a})_{x,x',a,a'}$ and 
$\tilde{F} = (\delta_{y,y'}\delta_{b,b'}F_{y,b})_{y,y',b,b'}$ be the canonical extensions of the families $\{(E_{x,a})_{a\in A} : x \in X\}$ 
(resp. $\{(F_{y,b})_{b\in B} : y\in Y\}$ to stochastic operator matrices and 
$\wt{S}_{\rm q} = (\mathbb C^2\otimes\mathbb C^2, \tilde{E}, \tilde{F}, \Omega_2)$; 
thus, $\wt{S}_{\rm q}$ is a model over $(\cl T_{X,A},\cl T_{Y,B})$, 
giving rise to a canonical element $f_{\wt{S}}$ of $\cl Q_{\rm q}$, 
namely, 
$$f_{\wt{S}}(e_{x,x',a,a'}\otimes e_{y,y',b,b'})=\delta_{x,x'}\delta_{a,a'}\delta_{y,y'}\delta_{b,b'}\langle (E_{x,a}\otimes E_{y,b})\Omega_2,\Omega_2\rangle.$$

We claim that $f_{\wt{S}}$ is not a self-test for $\cl Q_{\rm q}$. 
Indeed, let 
$\xi_{x,0}$, $\xi_{x,1}$,  $\eta_{y,0}$ and $\eta_{y,1}$ be unit vectors in the range of $E_{x,0}$, $E_{x,1}$, $F_{y,0}$ and $F_{y,1}$, respectively, 
and let
$V_x$, $W_y:\mathbb C^2\otimes\mathbb C^2\to\mathbb C^2\otimes\mathbb C^2$, 
be isometries, satisfying
$$V_x(\xi_{x,0}\otimes e_{0}) = \xi_{x,1}\otimes e_{1} \ \mbox{ and } \ 
W_y(\eta_{y,0}\otimes e_0) = \eta_{y,1}\otimes e_1,$$ 
and
$$V_x(\xi_{x,k}\otimes e_m) = 0  \ \mbox{ and } \   W_y(\eta_{y,k}\otimes e_m) = 0,
\ \mbox{ if } (k,m)\ne (0,0).$$
Further, let 
$$(E_{x,x',a,b})_{a,b}=\delta_{x,x'}\left(\begin{array}{cc}E_{x,0}\otimes I_2&V_x^*\\V_x&E_{x,1}\otimes I_2\end{array}\right), \ \ x,x'\in X,$$
$$(F_{y,y',a,b})_{a,b}=\delta_{y,y'}\left(\begin{array}{cc}F_{y,0}\otimes I_2&W_y^*\\W_y&F_{y,1}\otimes I_2\end{array}\right), \ \ y,y'\in Y,$$
 and observe that 
 $E := (E_{x,x',a,a'})_{x,x',a,a'}$ and $F := (F_{y,y',b,b'})_{y,y',b,b'}$ are stochastic operator matrices.
Set $\xi = \Omega_2\otimes e_0\otimes e_0$ and
$S = \left(\bb{C}^2\otimes \bb{C}^2, E,F, \right)$.
We have that 
 $$\langle (E_{x,x',a,a'}\otimes F_{y,y',b,b'})\xi,\xi\rangle=\delta_{x,x'}\delta_{y,'y}\delta_{a,a'}\delta_{b,b'}\langle (E_{x,a}\otimes F_{y,b})\Omega_2,\Omega_2\rangle,$$
that is, $S$ is a model of $f_{\wt{S}}$.

Suppose that $f_{\wt{S}}$ were a self-test for $\cl Q_{\rm q}$ 
with an ideal model 
$S^{\rm id} = \left(\bb{C}^2\otimes \bb{C}^2, (E_{x,x',a,a'}^{\rm id}, (F_{y,y',b,b'}^{\rm id}),\xi^{\rm id}\right)$.
Then, necessarily,
$$ (E_{x,x',a,a'}^{\rm id}\otimes F_{y,y',b,b'}^{\rm id})\xi^{\rm id} = \delta_{x,x'}\delta_{y,y'}\delta_{a,a'}\delta_{b,b'}(E_{x,x',a,a'}^{\rm ideal}\otimes F_{y,y',b,b'}^{\rm id})\xi^{\rm id}.$$
 As $(E_{x,x,a,a'}\otimes F_{y,y,b,b'})\xi\ne 0$ for $a\ne a'$, $b\ne b'$,  there are no isometries $V_A$, $V_B$ such that $(V_A\otimes V_B)(E_{x,x,a,a'}\otimes F_{y,y',b,b'})\xi=(E_{x,x',a,a'}^{\rm id}\otimes F_{y,y',b,b'}^{\rm id})\xi^{\rm id}\otimes\xi_{\rm aux}$ for a unit vector $\xi_{\rm aux}$, giving the statement.
\end{example}

\subsection{The CHSH game: quantum commuting self-tests}
\label{ss_CSHS-qcst}

It is known that the quantum value $\omega_{\rm q}({\rm CHSH})$ 
of the CHSH game coincides with its quantum commuting value $\omega_{\rm qc}({\rm CHSH})$ and that these are equal to $\frac{1}{2}+\frac{1}{2\sqrt{2}}$, with an optimal quantum 
strategy underlying the model $\tilde S$ given by (\ref{eq_optimmoSt}). 
As pointed out in Example \ref{CHSH_Cq_Qq}, 
the corresponding correlation $p_{\tilde S}$ is a self-test for the quantum PVM models with $\tilde S$ being an ideal model. In this subsection, we extend this by showing that $p_{\tilde S}$
determines an abstract self-test, as well as a self-test,  
for the class $\cl C_{\rm qc}$ of  quantum commuting POVM models.  
While Theorem \ref{th_stisabs} can be used to deduce the
former fact from the latter, we include a direct argument, showing
how the algebraic relations lying at the core of 
the fact that $p_{\tilde{S}}$ is a quantum abstract self-test can be extended to the commuting operator framework.  We follow a well-established route for the quantum case, 
generalising some of the constructions that appear in \cite{supic-bowles}.

\subsubsection{The correlation $p_{\tilde S}$ determines an abstract 
quantum commuting self-test.}

Let $X=Y=A=B=\mathbb Z_2$, 
$H$ be a Hilbert space, $\xi\in H$ be a unit vector, 
and $\{A_x\}_{x\in X}$ $\{B_y\}_{y\in Y}$ be families of selfadjoint unitary 
operators on $H$. 
Set $\cl A=\{A_x:x\in X\}''$
and $\cl B=\{B_y:y\in Y\}''$.
Let
$\{E_{x,a}\}_{a\in A}$ (resp. $\{F_{y,b}\}_{b\in B}$) be the 
spectral family of $A_x$ (resp. $B_y$), so that the relations
$A_x=E_{x,0}-E_{x,1}$ (resp. $B_y=F_{y,0}-F_{y,1}$) are 
satisfied and, writing 
$$S=(\mbox{}_{\cl A}H_{\cl B^{o}}, (E_{x,a})_{x\in X,a\in A},(F_{y,b})_{y\in Y,b\in B},\xi),$$ assume that 
$p:=p_S = p_{\tilde{S}}$, that is,
\begin{equation}\label{eq_newmo}
p_{\tilde S}(a,b|x,y)=\langle E_{x,a}F_{y,b}\xi,\xi\rangle,
\ \ \ x,y,a,b\in \bb{Z}_2.
\end{equation}
Let
$\beta_{S}$  be the \emph{bias operator} of the model $S$, 
defined by letting
$$\beta_{S} = A_0 B_0+A_0B_1+A_1B_0-A_1 B_1;$$
using a straightforward calculation, (\ref{eq_newmo}) implies
$$\frac{1}{8}\langle\beta_{S}\xi,\xi\rangle+\frac{1}{2}=
\frac{1}{4}\sum_{xy = a+b} p_{\tilde S}(a,b|x,y)
=\frac{1}{2}+\frac{1}{2\sqrt{2}}.$$ 
Set 
\begin{equation}\label{eq_ZA=XA=}
Z_A=\frac{A_0+A_1}{\sqrt{2}}\text { and }X_A=\frac{A_0-A_1}{\sqrt{2}}
\end{equation}
and observe that 
\begin{equation}\label{eq_ZAXA+}
Z_AX_A+X_AZ_A=0.
\end{equation}
    Moreover, 
    $$2\sqrt{2}-\beta_{S} = \frac{1}{2}\left(\left[\frac{A_0+A_1}{\sqrt{2}}-B_0\right]^2+\left[\frac{A_0-A_1}{\sqrt{2}}- B_1\right]^2\right).$$
As $\langle \beta_{S}\xi,\xi\rangle=2\sqrt{2}$,  
we conclude that 
\begin{equation}\label{eq_ZAB0}
Z_A\xi=B_0\xi \text{ and } X_A\xi=B_1\xi,
\end{equation}
        giving, by (\ref{eq_ZAXA+}),
        \begin{eqnarray*}
            &&(B_0B_1+B_1B_0)\xi=\frac{(A_0-A_1)B_0+(A_0+A_1)B_1}{\sqrt{2}}\xi\\&&=\frac{((A_0-A_1)(A_0+A_1)+(A_0+A_1)(A_0-A_1)))}{\sqrt{2}}\xi=0.
        \end{eqnarray*}
By symmetry, we also have the relation  $(A_0A_1+A_1A_0)\xi=0$.
        

Let $\frak A=C^*(\mathbb Z_2\ast\mathbb Z_2)\otimes C^*(\mathbb Z_2\ast\mathbb Z_2)$ as a C*-algebraic tensor product 
(since the C*-algebra $C^*(\mathbb Z_2\ast\mathbb Z_2)$ is nuclear, the C*-tensor product is unabiguiusly defined); we have that 
$\frak A$ is the universal $C^*$-algebra, generated by selfadjoint unitaries $a_x$, $b_y$, satisfying the 
property $a_xb_y = b_y a_x$, $x,y\in \bb{Z}_2$. 
Thus, the families $(A_x)_{x\in X}$ and $(B_y)_{y\in Y}$ determine a (unique) *-representation $\pi$ of $\frak A$ 
via the relations 
$$\pi(a_x\otimes b_y) = A_xB_y, \ \ \ x,y\in\mathbb Z_2.$$
Moreover, any representation $\pi$ of $\frak A$ is determined by commuting pairs of selfadjoint unitaries $(A_x)_{x\in \mathbb Z_2}$, $(B_y)_{y\in \mathbb Z_2}$,  by letting $\pi(a_x\otimes 1)=A_x$ and  $\pi(1\otimes b_y)=B_y$. 

The irreducible representation of $\frak A$ are given by $\pi_A\otimes\pi_B$, where $\pi_A$ and $\pi_B$ are irreducible representations of $C^*(\mathbb Z_2\ast\mathbb Z_2)$; the latter are unitarily equivalent to one of the following (see e.g. \cite{os}): 

\begin{enumerate}
\item A continuum of two dimensional *-representations:
$$\pi_{\varphi}(a_0)=\left(\begin{array}{cc} \cos\varphi & \sin\varphi \\
\sin\varphi &-\cos\varphi \end{array}\right) \mbox{ and } 
\pi_{\varphi}(a_1)=\left(\begin{array}{cc} \cos\varphi&-\sin\varphi\\-\sin\varphi&-\cos\varphi\end{array}\right),$$
where $\varphi\in(0,\pi/2)$; 

\item Four one-dimensional *-representations: 
$$\pi_{i,j}(a_0)=(-1)^i, \quad \pi_{i,j}(a_1)=(-1)^j, 
\ \mbox{ where }  i,j\in\{0,1\}.$$
\end{enumerate}
We note the identities 
\begin{equation}\label{eq_a0a1}
\pi_{\varphi}(a_0a_1+a_1a_0) = 2(\cos^2\varphi - \sin^2\varphi) \text{ and }\pi_{i,j}(a_0a_1+a_1a_0)=2(-1)^{i+j}.
\end{equation}

Write $\widehat {\frak A}$ for the dual space of classes of irreducible representations of $\frak A$; thus, 
$$\widehat{\frak A} = 
\{\pi_{\varphi}\otimes\pi_{\psi}, \pi_{\varphi}\otimes\pi_{i,j}, \pi_{i,j}\otimes\pi_{\psi}, \pi_{i,j}\otimes\pi_{s,t} : $$ 
$$ \hspace{4cm} \varphi,\psi\in(0,\pi/2), i,j,s,t\in\{0,1\}\}.$$
As the pairs $(A_0,A_1)$ and $(B_0,B_1)$ determine a 
$*$-representation of $\frak A$ on $H$, by the structure result 
\cite[Theorem 8.6.6]{dixmier},
we obtain the existence of a measure $\mu$ on the 
Borel space $\widehat{\frak A}$ such that, up to unitary equivalence, 
$H=\int_{\widehat{\frak A}}^\oplus H_\pi\otimes \ell^2(m(\pi))d\mu(\pi)$, 
where $H_{\pi} = \bb{C}^2$ is the Hilbert space on which 
the irreducible representation $\pi$ acts, 
$$A_x=\int_{\widehat{\frak A}}^\oplus\pi(a_x\otimes 1)\otimes I_{m(\pi)}d\mu(\pi)
\ \mbox{ and } \ B_y=\int_{\widehat{\frak A}}^\oplus\pi(1\otimes b_y)\otimes I_{m(\pi)}d\mu(\pi).$$

Write $\xi=\int_{\widehat{\frak A}}^\oplus\xi(\pi)d\mu(\pi)$, where, 
$\xi(\pi)\in H_\pi\otimes\ell^2(m(\pi))$ for each $\pi\in \widehat{\frak A}$, and
$\pi\mapsto \xi(\pi)$ is a measurable field of vectors over 
$\widehat{\frak A}$.
Since $(A_0A_1+A_1A_0)\xi=0$  and  $(B_0B_1+B_1B_0)\xi=0$, 
identities (\ref{eq_a0a1}) imply that, if 
$\tau:=\pi_{\pi/4}\otimes\pi_{\pi/4}$ then 
the measure $\mu$ has $\{\tau\}$ as an atom, 
the function 
$\pi\mapsto \xi(\pi)$ is supported in the singleton
$\{\tau\}$, and  $H(\tau):=H_\tau\otimes \ell^2(m(\tau))$ is a direct summand in $H$, invariant with respect to $(A_0,A_1)$ and $(B_0,B_1)$. Furthermore, 
up to unitary equivalence, 
$$\tau(a_i) = \pi_{\pi/4}(a_i)\otimes I_2, \ 
\tau(b_0) = I_2\otimes\sigma_x \mbox{ and } 
\tau(b_1) = I_2\otimes\sigma_z,$$ 
so that, up to unitary equivalence, 
\begin{eqnarray*}
&&A_0|_{H(\tau)}= \frac{\sigma_x+\sigma_z}{\sqrt{2}}\otimes I_2\otimes I_{m(\tau)}, \ \ A_1|_{H(\tau)}= \frac{\sigma_z-\sigma_x}{\sqrt{2}}\otimes I_2\otimes I_{m(\tau)}\\
&&B_0|_{H(\tau)}=I_2\otimes\sigma_z\otimes I_{m(\tau)}
\ \mbox{ and } \ B_1|_{H(\tau)}= I_2\otimes\sigma_x\otimes I_{m(\tau)}.
\end{eqnarray*}

Recall that $\cl S_{X,A}$ is the operator subsystem of 
the C*-algebra $\cl A_{X,A}$,
with canonical generators thee projections $e_{x,a}$, $x\in X$, $a\in A$
(see Subsection \ref{ss_POVMNS}); 
in the case under consideration, 
$\cl A_{X,A} = \cl A_{Y,B} \simeq C^*(\mathbb Z_2\ast\mathbb Z_2)$. 
Let $f:\cl S_{X,A}\otimes_{\rm c}\cl S_{Y,B}\to\mathbb C$ be a state such that $f(e_{x,a}\otimes e_{y,b})=p_{\tilde S}(a,b|x,y)$, $x,y,a,b\in\mathbb Z_2$,  and let $g$ be a state on $\cl A_{X,A}\otimes_{\rm max} \cl A_{Y,B}$,
given by $g(u) = \langle\pi(u)\xi,\xi\rangle$ such that $g|_{\cl S_{X,A}\otimes_{\rm c}\cl S_{Y,B}}=f$.
By the previous paragraph,
$$g(u) = \langle (\tau(u)\otimes I_{m(\tau)})\xi(\tau),\xi(\tau)\rangle,  \ \ \ 
u\in \cl A_{X,A}\otimes_{\rm max} \cl A_{Y,B}.$$
We will show that $\xi(\tau) = \Omega_2\otimes\xi_{\rm aux}$, where $\Omega_2$ is the maximal entangled vector and $\xi_{\rm aux}\in \ell^2(m(\tau))$. A straightforward calculation shows that 
the identity
$$f(a_0\otimes b_0+a_0\otimes b_1+a_1\otimes b_0-a_1\otimes b_1)=2\sqrt{2}$$
implies that
        $$\frac{1}{2}
\left\langle((\sigma_x\otimes\sigma_x+\sigma_z\otimes\sigma_z)\otimes I_{m(\tau)})\xi(\tau),\xi(\tau)\right\rangle = 1.$$
The largest eigenvalue of the operator $\frac{1}{2}(\sigma_x\otimes\sigma_x+\sigma_z\otimes\sigma_z)$ 
is equal to 1 and has corresponding eigenspace the one-dimensional subspace spanned by $\Omega_2$.  
This implies that $\xi(\tau)=\Omega_2\otimes \xi_{\rm aux}$ for some $\xi_{\rm aux}\in \ell^2(m(\tau))$. Thus,
$$g(u) = \langle\tau(u)\Omega_2,\Omega_2\rangle, \ \ \ 
u\in \cl A_{X,A}\otimes_{\rm max} \cl A_{Y,B}.$$
This also finishes the proof that $f$ (equivalently $p$) is an abstract self-test for the class of states that factor through $\cl A_{X,A}\otimes_{\max}\cl A_{Y,B}$.  

\begin{remark} \label{optimal} 
We note that in proving that 
$p_{\tilde{\cl S}}$ is an abstract self-test 
we only used that the fact that $p_{\tilde{\cl  S}}$ is an optimal quantum commuting strategy for CHSH game. Since any $p\in\cl C_{\rm qc}$ admits a PVM model, in this way we show in particular the uniqueness of the optimal correlation in the class $\cl C_{\rm qc}$.
\end{remark}

\subsubsection{The correlation $p_{\tilde S}$ determines a 
quantum commuting self-test.}

We show that $p$ is a self-test for 
the class of quantum commuting PVM models. We start with some preparations. 
Let $P_A$ be the kernel of the operator $Z_A$ 
defined in (\ref{eq_ZA=XA=}) and set $\hat Z_A=(Z_A+P_A)|Z_A+P_A|^{-1}$ (a regularization of $Z_A$ in terms of \cite[A.2]{supic-bowles}), 
where $|T| := (T^*T)^{1/2}$ is the absolute value of an operator $T$. Note that $Z_A+P_A$ is an injective selfadjoint operator and hence $\hat Z_A$ is a unitary operator. We similarly let 
$\hat X_A=(X_A+Q_A)|X_A+Q_A|^{-1}$, where $Q_A$ be the projection onto 
the kernel of the operator $X_A$ defined in (\ref{eq_ZA=XA=}). 
Write 
$$z_A(\pi)=\left\{\begin{array}{ll}\sqrt{2}\cos\varphi\left(\begin{array}{cc}1&0\\0&-1\end{array}\right)\otimes I,& 
\pi = \pi_{\varphi}\otimes\pi_{\psi}\text{ or }\pi_{\varphi}\otimes\pi_{i,j}\\
((-1)^{i}+(-1)^j)/\sqrt{2}\otimes I,& \pi=\pi_{i,j}\otimes\pi_{\psi}\text{ or }\pi_{i,j}\otimes\pi_{s,t},
\end{array}\right.$$
$$ x_A(\pi)=\left\{\begin{array}{ll}\sqrt{2}\sin\varphi\left(\begin{array}{cc}0&1\\1&0\end{array}\right)\otimes I,
& \pi = \pi_{\varphi}\otimes\pi_{\psi}\text{ or }
\pi_{\varphi}\otimes\pi_{i,j}\\
((-1)^{i}-(-1)^j)/\sqrt{2}\otimes I,& \pi=\pi_{i,j}\otimes\pi_{\psi}\text{ or }\pi_{i,j}\otimes\pi_{s,t},
\end{array}\right.$$
$$\hat z_A(\pi)=\left\{\begin{array}{ll}\left(\begin{array}{cc}1&0\\0&-1\end{array}\right)\otimes I,& \pi=\pi_{\varphi}\otimes\pi_{\psi}\text{ or }\pi_{\varphi}\otimes\pi_{i,j}\\
\pm 1\otimes I,& \pi=\pi_{i,j}\otimes\pi_{\psi}\text{ or }\pi_{i,j}\otimes\pi_{s,t},
\end{array}\right.$$
and
$$\hat x_A(\pi)=\left\{\begin{array}{ll}\left(\begin{array}{cc}0&1\\1&0\end{array}\right)\otimes I,& \pi=\pi_{\varphi}\otimes\pi_{\psi}\text{ or }\pi_{\varphi}\otimes\pi_{i,j}\\
\pm 1\otimes I,& \pi=\pi_{i,j}\otimes\pi_{\psi}\text{ or }\pi_{i,j}\otimes\pi_{s,t},
\end{array}\right.$$
so that 
$$Z_A=\int_{\widehat{\frak A}}^\oplus z_A(\pi)\otimes I_{m(\pi)}d\mu(\pi),\quad  X_A=\int_{\widehat{\frak A}}^\oplus x_A(\pi)\otimes I_{m(\pi)}d\mu(\pi)$$
and 
$$\hat Z_A=\int_{\widehat{\frak A}}^\oplus\hat z_A(\pi)\otimes I_{m(\pi)}d\mu(\pi),\quad  \hat X_A=\int_{\widehat{\frak A}}^\oplus\hat x_A(\pi)\otimes I_{m(\pi)}d\mu(\pi).$$

In particular, as $\pi\mapsto\xi(\pi)$ is supported in 
$\{\tau\}$, by (\ref{eq_ZAB0}) we obtain
\begin{equation}\label{vectorequ}
\hat Z_A\xi= Z_A\xi=B_0\xi,\quad  \hat X_A\xi= X_A\xi=B_1\xi, \quad Z_A^2\xi=X_A^2\xi=\xi,
\end{equation}
and 
\begin{equation}\label{vectorequ2}
(\hat Z_A\hat X_A+\hat X_A\hat Z_A)\xi=0, \quad (B_0B_1+B_1B_0)\xi=0.
\end{equation}

We write $\{e_k^*\}_{k=0,1}$ for the dual basis of $\bb{C}^2$, 
considered as linear functionals on $\bb{C}^2$, given by $e_k^*(e_i)=\delta_{i,k}$. 
It will further be convenient we use the notation $e_k$ 
for the linear map $\mathbb C\to\mathbb C^2$, $1\mapsto e_k$, 
and note that $e_k$ is the adjoint $e_k^*$, $k = 0,1$. 
Let $V_{1,2} : H\mapsto \mathbb C^2\otimes H$ be the 
operator, given by 
\begin{equation}\label{eq_V_12de}
V_{1,2} = \frac{1}{2}(e_0\otimes (I+\hat Z_A) 
+ e_1\otimes \hat X_A(I-\hat Z_A)).
\end{equation}
Observe that, as $(1+\hat Z_A)/2$ is a projection and $\hat X_A$ is a unitary,  
for $\eta\in H$ we have 
\begin{eqnarray*}
\|V_{1,2}\eta\|^2
& = &
\frac{1}{4}\left(\left\|(I+\hat Z_A)\eta\right\|^2
+ \left\|\hat X_A(I-\hat Z_A)\eta\right\|^2\right)\\
& = &
\left\|\frac{1}{2}(I+\hat Z_A)\eta \right\|^2
+ \left\|\frac{1}{2}(I-\hat Z_A)\eta\right\|^2
=
\|\eta\|^2,
\end{eqnarray*}
that is, $V_{1,2}$ is an isometry.
Equipping  $\mathbb C^2\otimes H$ with the natural
$M_2(\mathbb C)\otimes \cl A$-$(M_2(\mathbb C)\otimes \cl B)^{o}$-bimodule action, we claim that $V_{1,2}$ is 
$\cl A$-$(M_2(\mathbb C)\otimes\cl A)$-local.

Set $P_{i,A} = \hat X_A^i(I + (-1)^i\hat Z_A)/2$. 
Then $P_{i,A}\in\cl A$. For $\eta\in \bb{C}^2\otimes H$ and $\psi\in H$, we have
$$\langle V_{1,2}^*\eta,\psi\rangle
= \langle \eta,V_{1,2}\psi\rangle
= \left\langle \eta,\sum_{i=0}^1(e_i\otimes P_{i,A})\psi\right\rangle
= \left\langle \sum_{i=0}^1(e_i^*\otimes P_{i,A}^*)\eta,\psi\right\rangle,$$
giving $V_{1,2}^*=\sum_{i=0}^1e_i^*\otimes P_{i,A}^*$. 
Thus, 
$$V_{1,2} T V_{1,2}^*=\sum_{i,j=0}^1e_je_i^*\otimes P_{j,A} T P_{i,A}^*, \ \ \ T\in \cl A,$$
showing that 
$V_{1,2}\cl AV_{1,2}^*\subseteq M_2(\mathbb C)\otimes\cl A$.
Using the fact that $\cl A\subseteq\cl B'$, 
we also have that 
$$(I\otimes R)V_{1,2} = V_{1,2}R, \ \ \ R\in \cl B.$$

Let now $V_{2,2}: \mathbb C^2\otimes H\to \mathbb C^2\otimes\mathbb C^2\otimes H$ be given by
\begin{equation}\label{eq_V_22de}
V_{2,2}=\sum_{i=0}^1 I\otimes e_i\otimes P_{i,B}, 
\end{equation}
where $P_{i,B}$ is defined similarly, replacing  $Z_A$ 
(resp. $X_A$) by $Z_B := B_0$ 
(resp. $X_B := B_1$). 
Using the previous arguments, we can see that $V_{2,2}$ is an $M_2(\mathbb C)\otimes\cl B$-$(M_2(\mathbb C)\otimes M_2(\mathbb C)\otimes\cl B$-local isometry and 
$$V_{2,2}V_{1,2}=\sum_{i,j=0}^1 e_i\otimes e_j\otimes P_{j,B}P_{i,A}.$$
Make natural modifications in 
(\ref{eq_V_12de}) and (\ref{eq_V_22de}), 
we further define local isometries $V_{1,1}:H\to\mathbb C^2\otimes H$
and $V_{2,1} : \mathbb C^2\otimes H \to \mathbb C^2\otimes \bb{C}^2 \otimes H$
by letting
$V_{1,1} = \sum_{j=0}^1 e_j\otimes P_{j,B}$ and 
$V_{2,1}=\sum_{i=0}^1 e_i\otimes I\otimes P_{i,A}$, 
noting that $V_{2,1}V_{1,1}=V_{2,2}V_{1,2}$.

We show that 
there exists a unit vector $\xi_{\rm aux}\in H$ such that
$V_{2,2}V_{1,2}\xi=\Omega_2\otimes\xi_{\rm aux}$, and that 
\begin{equation}\label{eq_V22tau}
V_{2,2}V_{1,2}\pi(u)\xi = \tau(u)\Omega_2\otimes\xi_{\rm aux}, 
\ \ \ u\in \cl S_{X,A}\otimes\cl S_{Y,B}.
\end{equation}
By (\ref{vectorequ}), $\hat Z_A\xi=Z_A\xi=Z_B\xi$, giving  
$$(I \mp \hat Z_B)(I \pm \hat Z_A)\xi = 0,$$
and hence 
\begin{equation}\label{eq_neqij}
P_{j,B}P_{i,A}\xi=0 \ \mbox{ if } i\ne j. 
\end{equation}
Next we observe that, by (\ref{vectorequ}) and (\ref{vectorequ2}), 
$$\hat X_A\hat Z_A\xi=-\hat Z_A\hat X_A\xi, 
\ X_BZ_B\xi=-Z_BX_B\xi \mbox{ and }  X_B\xi=X_A\xi=\hat X_A\xi,$$ and therefore, using the fact that $\cl A\subset \cl B'$, we have 
\begin{eqnarray*}
P_{1,B}P_{1,A}\xi
& = &
X_B\hat X_A(I-Z_B)(I - \hat Z_A)\xi
= (I + Z_B)(I + \hat Z_A)X_B\hat X_A\xi\\
& = &
(I + Z_B)(I + \hat Z_A)\xi
= P_{0,B}P_{0,A}\xi.
\end{eqnarray*}
Setting $\xi_{\rm aux}=\sqrt{2}P_{0,B}P_{0,A}\xi$, using (\ref{eq_neqij}) we thus have
$$V_{2,2}V_{1,2}\xi=(e_0\otimes e_0+e_1\otimes e_1)\otimes P_{0,B}P_{0,A}\xi = \Omega_2\otimes \xi_{\rm aux}.$$
Using the fact that $\hat Z_A\xi=Z_A\xi$ and $Z_A^2\xi=\xi$
(see relations (\ref{vectorequ})), we get
\begin{eqnarray*}
P_{j,B}P_{i,A}Z_A\xi
& = &
X_B^j\hat X_A^i(I + (-1)^jZ_B)(I + (-1)^i\hat Z_A)Z_A\xi\\
& = &
X_B^j\hat X_A^i(I + (-1)^jZ_B)(\hat Z_A+(-1)^iI)\xi
= (-1)^iP_{j,B}P_{i,A}\xi
\end{eqnarray*}
and hence, taking into account (\ref{eq_neqij}), we obtain 
$$V_{2,2}V_{1,2}Z_A\xi=(e_0\otimes e_0-e_1\otimes e_1)\otimes P_{0,B}P_{0,A}\xi
= (\sigma_z\otimes 1)\Omega_2\otimes\xi_{\rm aux}.$$
Since $\cl A\subseteq \cl B'$, we furthremore have
\begin{eqnarray*}
P_{j,B}P_{i,A}X_A\xi
& = &
X_B^j\hat X_A^i(I+(-1)^jZ_B)(I+(-1)^i\hat Z_A)X_A\xi\\
& = &
X_B^j\hat X_A^{i+1}(I+(-1)^jZ_B)(I+(-1)^{i+1})\hat Z_A)\xi\\
& = & 
P_{j,B}P_{i+1,A}\xi   
\end{eqnarray*}
and hence 
$$V_{2,2}V_{1,2}X_A\xi=(e_1\otimes e_0+e_0\otimes e_1)\otimes P_{0,B}P_{0,A}\xi=(\sigma_x\otimes I)\Omega_2\otimes\xi_{\rm aux}.$$
In a similar way we show that
$$V_{2,2}V_{1,2}Z_B\xi = (e_0\otimes e_0-e_1\otimes e_1)\otimes P_{0,B}P_{0,A}\xi=(I\otimes\sigma_z)\Omega_2\otimes\xi_{\rm aux}$$
and 
$$V_{2,2}V_{1,2}X_B\xi = (e_0\otimes e_1+e_1\otimes e_0)\otimes P_{0,B}P_{0,A}\xi=(I\otimes\sigma_x)\Omega_2\otimes\xi_{\rm aux}.$$
Equation (\ref{eq_V22tau}) is therefore established.

\smallskip

Finally, we can remove the condition of being quantum commuting PVM model.
It relies on the following statement which is similar to \cite[Proposition 5.5]{pszz}.
\begin{proposition}\label{PVM-POVM}
Let $X$ and $Y$ be finite sets and 
suppose that the quantum commuting correlation 
$p=\{p(a,b|x,y): x\in X, y\in Y, a,b\in\mathbb Z_2\}$ is an extreme point 
in $\cl C_{\rm qc}$. If $S$ is a quantum commuting model for $p$ then there exists a projective quantum commuting model $\wt{S}$ such that $S\preceq\wt{S}$.
\end{proposition}

\begin{proof}
    Let $S=(H, (E_{x,a})_{x\in X, a\in \mathbb Z_2}, (F_{y,b})_{y\in Y, b\in \mathbb Z_2}, \xi)$ be a quantum commuting model for $p$, that is, $p(a,b|x,y)=\langle E_{x,a}F_{y,b}\xi,\xi\rangle$ and $(E_{x,a})_{a\in\mathbb Z_2}$ and $(F_{y,b})_{b\in\mathbb Z_2}$ are POVMs for each $x\in X$ and $y\in Y$. Let $E_{x,0}=\int_{[0,1]}\lambda dE_x(\lambda)$ be the spectral decomposition of $E_{x,0}$, $x\in X$. Fix $x\in X$ and $0<c<1/2$, 
    and set $\Omega_c = [c,1]$.
    Write 
    $$E_{x,0} = cE_x(\Omega_c)+(1-c)T_{x,c},$$
    where $T_{x,c} = \frac{1}{1-c}\left(\int_{[0,1]}\lambda E_{x}(\lambda)-cE_x(\Omega_c)\right)$. 
    Then 
    $$T_{x,c}=\frac{1}{1-c}\left(\int_{\Omega_c}(\lambda-c)dE_x(\lambda)+\int_{\Omega_c^c}\lambda dE_x(\lambda)\right)\geq 0,$$
    and 
    $$I-T_{x,c} = \frac{1}{1-c}\left(\int_{\Omega_c}(1-\lambda)dE_x(\lambda)+\int_{\Omega_c^c}(1-c-\lambda) dE_x(\lambda)\right)\geq 0.$$
    It gives $p=cp_1+(1-c)p_2$, where $p_1$ and $p_2$ are 
    the quantum commuting correlations corresponding to the POVM,s where ($E_{x,0}$, $I-E_{x,0} $) is replaced by 
    ($E_x(\Omega_c)$, $1-E_x(\Omega_c)$) and ($T_{x,c}$, $I-T_{x,c}$), respectively.
    As $p$ is extreme, $p_1=p_2$ and hence $\langle E_x([c,1])\xi,\xi\rangle=\langle E_{x,0}\xi,\xi\rangle$. Letting $c\to 0$, we obtain that $\langle E_x((0,1]) \xi,\xi\rangle=\langle E_{x,0}\xi,\xi\rangle$. Let $P_{x,0}=E_x(\{1\})$ and note that if $P_{x,0}\ne 0$, it is the projection onto non-zero eigenspace of $E_{x,0}$  corresponding to the eigenvalue 1.  We have
    $$0=\langle(E_{x,0}-E_x((0,1])\xi,\xi\rangle=\int_{(0,1)}(\lambda-1)d\langle E_x(\lambda)\xi,\xi\rangle,$$
    showing that $E_x((0,1))\xi=0$ and hence $E_{x,0}\xi=P_{x,0}\xi$. Note that $P_{x,0}$ commute with each $F_{y,b}$. Similarly, we find projections  $Q_{y,0}$ such that $F_{y,0}\xi=Q_{y,0}\xi$ and commute with each $P_{x,0}$. 
    Set $$\wt{S}=(H, (P_{x,a})_{x\in X,a\in\mathbb Z_2}, (Q_{y,b})_{y\in Y,b\in\mathbb Z_2},\xi).$$ Then $S\preceq\wt{S}$. 
\end{proof}

\begin{corollary}
The correlation $p_{\tilde S}$
(where $\tilde{S}$ is the model given in (\ref{eq_optimmoSt})) 
is a self-test for the class of all quantum commuting models.
\end{corollary}

\begin{proof}
    As $p_{\tilde S}$ is the unique 
    optimal quantum commuting strategy for the CHSH game in the class $\cl C_{\rm qc}$ (see Remark \ref{optimal}) 
    assuming that $p_{\tilde S} = cp_1+(1-c)p_2$, $c\in(0,1)$, for
    some quantum commuting correlations $p_1$ and $p_2$, we obtain that $$\omega_{\rm qc}({\rm CHSH},p_1) =\omega_{\rm qc}({\rm CHSH},p_2)=\omega_{\rm qc}({\rm CHSH}, p_{\tilde S}).$$ 
    As the optimal strategy is unique, $p_1=p_2=p_{\tilde S}$ which  show that $p_{\tilde{\cl  S}}$ is extreme in $\cl C_{\rm qc}$. The statement now follow from Proposition \ref{PVM-POVM}.
    \end{proof}

\subsection{Clifford correlations}\label{ss_clifford}

Let $X$ be a finite set. 
We note that the C*-algebra $\cl A_{X,\bb{Z}_2}$ is identical with the 
universal unital C*-algebra generated by $|X|$ projections $e_x$, $x\in X$, via the
isomorphism given by letting 
$e_{x,0} = e_x$ and $e_{x,1} = 1 - e_x$, $x\in X$. 
Let 
$$\cl J_{\rm C} = \left\langle \left\{e_{x,0} e_{y,0} + e_{y,0} e_{x,0}  - 
e_{x,0} - e_{y,0} + \frac{1}{2}\cdot 1 
: x,y\in X, x\neq y \right\}\right\rangle$$
as a closed ideal of $\cl A_{X,\bb{Z}_2}$.
Recall, further, that the \emph{Clifford algebra} $\frak{C}_X$ over $X$ 
is the (unital) C*-algebra, generated by a family 
$\{u_x\}_{x\in X}$ of selfadjoint unitaries, satisfying the anticommutation relations 
$$u_x u_y + u_y u_x = 0, \ \ \ x,y\in X, \ x\neq y.$$

For the following fact, see, for example, \cite{tsirelson}. 

\begin{theorem}\label{th_Clifone}
If $|X|$ is even then, up to unitary equivalence, $\frak{C}_X$ has a unique 
irreducible representation $\pi_{\rm C}$, which is also 
faithful, and whose image is $M_{2^{\frac{n}{2}}}$.  
\end{theorem}

\begin{proposition}\label{p_cliff}
We have that $\frak{C}_X  = \cl A_{X,\bb{Z}_2}/\cl J_{\rm C}$,
up to a canonical *-isomorphism.
\end{proposition}

\begin{proof}
Let $p_x$ be the eigenprojection of $u_x$ corresponding to the eigenvalue $1$; 
thus, $u_x = 2p_x - 1$, $x\in X$. 
For $x\neq y$, we have that 
\begin{eqnarray*}
u_x u_y + u_y u_x
& = & 
(2p_x - 1)(2p_y - 1)  + (2p_y - 1)(2p_x - 1)\\
& = & 
4 p_x p_y - 2p_y - 2p_x + 1 + 4 p_y p_x - 2p_x - 2p_y + 1\\
& = & 
4 p_x p_y + 4 p_y p_x - 4p_x - 4p_y + 2.
\end{eqnarray*}
The claim follows from the fact that $\cl A_{X,\bb{Z}_2}$ is the universal C*-algebra
of $|X|$ selfadjoint unitaries. 
\end{proof}

We introduce two classes of 
no-signalling correlations, which will prove to 
be amenable to self-testing. 
For the first subclass, write $q_{\rm C} : \cl A_{X,\bb{Z}_2}\to \frak{C}_X$ for the quotient map arising from Proposition \ref{p_cliff}. 
Let $\frak{S}_{X} = q_{\rm C}(\cl S_{X,[2]})$; thus, 
$\frak{S}_{X} = {\rm span}\{1, u_x : x\in X\}$, 
an operator subsystem of $\frak{C}_X$. 
An NS correlation $p$ over $(X,X,\bb{Z}_2,\bb{Z}_2)$ will be called a 
\emph{Clifford correlation} if there exists a state 
$s : \frak{S}_{X}\otimes_{\min} \frak{S}_{X} \to \bb{C}$ 
such that 
$$p_s(a,b|x,y) := s(q_{\rm C}(e_{x,a})\otimes q_{\rm C}(f_{y,b})), \ \ \ x,y\in X, a,b\in \bb{Z}_2$$
(for clarity, we denote by $f_{y,b}$, $y\in X$, $b\in \bb{Z}_2$, the canonical generators of the second copy of 
$\cl A_{X,\bb{Z}_2}$). 
Let $\cl S_{\rm C}$ be the set of all states on 
$\cl A_{X,\bb{Z}_2}\otimes_{\min}\cl A_{X,\bb{Z}_2}$ 
that factor through $\frak{C}_X\otimes_{\min} \frak{C}_X$.

For the second class, let $\Sigma = \{e_{x,a}, f_{y,b} : x,y\in X, a,b\in \bb{Z}_2\}$, 
$\Sigma^*$ be the set of all finite words on the alphabet $\Sigma$, reduced 
under the relations $e_{x,a}^2 = e_{x,a}$, $f_{y,b}^2 = f_{y,b}$ and 
$e_{x,a}f_{y,b} = f_{y,b}e_{x,a}$, 
equipped with a canonical involution
(each element $w\in \Sigma^*$ can be equivalently 
considered as an element of 
$\cl A_{X,\bb{Z}_2}\otimes \cl A_{X,\bb{Z}_2}$, 
where $e_{x,a}$ (resp. $f_{y,b}$) are 
identified with the canonical generators of 
$\cl A_{X,\bb{Z}_2}$
(resp. $\cl A_{Y,\bb{Z}_2}$)). 
We note that the empty word $\varepsilon$ is considered to be an element of $\Sigma^*$, and let 
$\Sigma' = \Sigma\cup \{\varepsilon\}$. 
We identify a quantum commuting correlation $p = ((p(a,b|x,y))_{x,y,a,b}$ with the matrix 
$M^{(p)} = [m_{\alpha,\beta}]_{\alpha,\beta\in \Sigma'}$, where 
$$m_{\alpha,\beta} = 
\begin{cases}
1 & \text{if } \alpha = \beta = \varepsilon\\
p(a|x) & \text{if } \alpha = \varepsilon \ \& \ \beta = e_{x,a}, 
\mbox{ or }  \alpha = e_{x,a} \ \& \ \beta = \varepsilon\\
p(b|y) & \text{if } \alpha = \varepsilon \ \& \ \beta = f_{y,b}, 
\mbox{ or }  \alpha = f_{y,b} \ \& \ \beta = \varepsilon\\
p(a,b|x,y) & \text{if } \alpha = e_{x,a} \ \& \ \beta = f_{y,b}, 
\mbox{ or }  \alpha = f_{y,b} \ \& \ \beta = e_{x,a}.\\
\end{cases}
$$
Every linear functional $s : \cl A_{X,\bb{Z}_2}\otimes \cl A_{X,\bb{Z}_2}\to \bb{C}$ gives rise to the 
matrix $M^{(s)} = [s(\beta^*\alpha)]_{\alpha,\beta\in \Sigma^*}$; we call the 
matrices arising in this way \emph{admissible}.
We denote the set of all admissible positive semi-definite matrices over 
$\Sigma^*\times\Sigma^*$ 
(that is, admissible matrices $M$ whose every finite minor is positive semi-definite) by $\bb{A}$. 
According to the NPA hierarchy \cite{npa}, 
a no-signalling correlation $p = ((p(a,b|x,y))_{x,y,a,b}$ is of quantum commuting type 
if and only if the matrix $M^{(p)}$ over $\Sigma'\times \Sigma'$ can be 
completed to a matrix 
$M = (M_{\alpha,\beta})_{\alpha,\beta}$ that lies in the set $\bb{A}$. 

Let $w_{x,y} = e_{x,0}e_{y,0}$; we view $w_{x,y}$ as 
elements of $\Sigma^*$. 
We write 
$$\bb{A}_{\rm C} = \left\{[m_{\alpha,\beta}]_{\alpha,\beta\in \Sigma^*} \hspace{-0.1cm} \in \hspace{-0.04cm}  \bb{A} : 
m_{e_{x,0},e_{y,0}} - m_{w_{x,y},w_{y,x}} = \frac{1}{8}, \ \ x,y\in X\right\}.$$


\begin{theorem}\label{p_syncste}
Let $X$ be a finite set of even cardinality,  
$s : \frak{S}_{X} \otimes_{\min} \frak{S}_{X}
\to \bb{C}$ be a state, and let 
$\tilde{s} = 
s\circ (q_{\rm C}\otimes q_{\rm C})$. 
The following hold true:

\begin{itemize}
\item[(i)]
if $p_{\tilde{s}}$ is a synchronous Clifford correlation then 
$\tilde{s}$ is an abstract self-test for $\cl S_{\rm C}$;

\item[(ii)] if $p_{\tilde{s}}$ is a synchronous correlation
that admits a positive completion to an element of $\bb{A}_{\rm C}$ then 
$\tilde{s}$ is an abstract self-test for 
$S_{\rm sp}(\cl A_{X,\bb{Z}_2}\otimes_{\min} \cl A_{X,\bb{Z}_2})$.
\end{itemize}
\end{theorem}

\begin{proof}
(i) Suppose that $p_s$ is a synchronous Clifford correlation, and let, as in (i), 
$\phi : \frak{C}_X\otimes_{\min} \frak{C}_X \to \bb{C}$ be an extension of $s$.
Let $\tilde{u}_x = e_{x,0} - e_{x,1}$, $x\in X$. 
By the synchronicity of $p_s$, we have 
$$\tilde{s}(e_{x,1}\otimes e_{x,0}) = \tilde{s}(e_{x,0}\otimes e_{x,1}) = 0, \ \ \ x\in X.$$
Note that 
\begin{eqnarray*}
|X|1 
& = & 
\sum_{x\in X} (e_{x,0} + e_{x,1})\otimes (e_{x,0} + e_{x,1})\\
& = & 
\sum_{x\in X} e_{x,0} \otimes e_{x,0} + e_{x,1} \otimes e_{x,0} + 
e_{x,0} \otimes e_{x,1} + e_{x,1} \otimes e_{x,1},
\end{eqnarray*}
while
\begin{eqnarray*}
\sum_{x\in X} \tilde{u}_x\otimes \tilde{u}_x 
& = & 
\sum_{x\in X} (e_{x,0} - e_{x,1})\otimes (e_{x,0} - e_{x,1})\\
& = & 
\sum_{x\in X}  e_{x,0} \otimes e_{x,0} - e_{x,0} \otimes e_{x,1} 
- e_{x,1} \otimes e_{x,0} + e_{x,1} \otimes e_{x,1}.\\
\end{eqnarray*}
Thus, 
$$|X|1 - \sum_{x\in X} \tilde{u}_x\otimes \tilde{u}_x = 
2\sum_{x\in X} e_{x,0} \otimes e_{x,1} 
+ e_{x,1} \otimes e_{x,0}.$$
It follows that
\begin{equation}\label{eq_uxux}
|X|1 - \sum_{x\in X} u_x\otimes u_x \geq 0 
\end{equation}
and that 
\begin{equation}\label{eq_uxtildes}
\tilde{s}(|X|1 - \sum_{x\in X} \tilde{u}_x\otimes \tilde{u}_x) = 0.
\end{equation}
Write 
\begin{equation}\label{eq_repoftiph}
\phi(u) = \langle \pi(u)\xi,\xi\rangle, \ \ \ 
u\in \frak{C}_X \otimes_{\min} \frak{C}_X,
\end{equation}
arising from the GNS representation of $\phi$. 
By (\ref{eq_uxtildes}), 
$$\left\langle \pi\left(|X|1 - \sum_{x\in X} u_x\otimes u_x\right)\xi,\xi\right\rangle = 0$$
which, together with (\ref{eq_uxux}) implies that 
$\pi\left(|X|1 - \sum_{x\in X} u_x\otimes u_x\right)\xi = 0$. 
By \cite[Lemma 1.3]{tsirelson}, 
$\pi$ is irreducible, and since $|X|$ is even, 
Theorem \ref{th_Clifone} implies that, up to unitary equivalence, 
$\pi = \pi_{\rm C}$. 
Furthermore, from the proof of \cite[Lemma 1.3]{tsirelson}, 
the vector $\xi$ is determined uniquely up to a scalar multiple. 
The relation (\ref{eq_repoftiph}) shows that 
$\tilde{\phi}$ is 
the unique extension of $\tilde{s}$.

(ii) 
Let $\Theta : \cl A_{X,\bb{Z}_2}\otimes\cl A_{X,\bb{Z}_2}^o \to \cl A_{X,\bb{Z}_2}$ be the map, given by 
$\Theta(u\otimes v^o) = uv$, and 
let $\partial : \cl A_{X,\bb{Z}_2}\to \cl A_{X,\bb{Z}_2}^o$ be the *-isomorphism, given by 
$\partial(e_{x,a}) =  e_{x,a}^o$ \cite{kps}. 
By the synchronicity of $p_{\tilde{s}}$ and \cite{psstw}, there exists a tracial state $\tau : \cl A_{X,\bb{Z}_2}\to \bb{C}$, such that
$$\tilde{s}(w) = \tau\left((\Theta\circ \partial)(u)\right), \ \ \ w\in \cl S_{X,\bb{Z}_2}\otimes \cl S_{X,\bb{Z}_2}.$$
Write $r_x = e_{x,0}$ for brevity, and let 
$c = r_{x} r_{y} + r_{y} r_{x}  - r_{x} - r_{y} + \frac{1}{2}$, as an element of 
$\cl A_{X,\bb{Z}_2}$. 
A direct calculation shows that 
$$c^2 = r_xr_yr_xr_y + r_yr_xr_yr_x - r_xr_yr_x - r_yr_xr_y + \frac{1}{4}\cdot 1.$$
It follows that 
$$\tau(c^2) = 2 \tau(r_xr_yr_xr_y)  - 2 \tau(r_xr_y) + \frac{1}{4}\cdot 1.$$
Since $p_{\tilde{s}} \in \bb{A}_{\rm C}$, we have that $\tau(c^2) = 0$, and hence 
$\tau$ annihilates the ideal $\cl J_{\rm C}$ generated by $c$. 
By Proposition \ref{p_cliff}, 
$\tau$ induces a trace $\tilde{\tau} : \frak{C}_X\to \bb{C}$, 
and 
\begin{equation}\label{eq_s(u)}
s(u) = \tilde{\tau}\left(q_{\rm C}((\Theta\circ \partial)(u))\right), \ \ \ u\in \cl S_{X,\bb{Z}_2}\otimes \cl S_{X,\bb{Z}_2}.
\end{equation}
Let $\tilde{\phi} : \cl A_{X,\bb{Z}_2}\otimes_{\max} 
\cl A_{X,\bb{Z}_2}\to \bb{C}$ be an extension of $\tilde{s}$ to 
an element of $S(\cl A_{X,\bb{Z}_2}\otimes_{\max} 
\cl A_{X,\bb{Z}_2})$. Then 
$\tilde{\phi}$ annihilates $\cl J_{\rm C}\otimes \cl A_{X,\bb{Z}_2} 
+ \cl A_{X,\bb{Z}_2}\otimes \cl J_{\rm C}$ and, 
by the projectivity of the maximal tensor product, 
induces a state 
$\phi : \frak{C}_X\otimes_{\max} \frak{C}_X\to \bb{C}$. 
By the nuclearity of $\frak{C}_X$, we may consider $\phi$ as 
a state on $\frak{C}_X\otimes_{\min} \frak{C}_X$.
By (\ref{eq_s(u)}), $\phi$ extends $s$, and hence $p_{\tilde{s}}$ is a 
Clifford correlation. 
The claim now follows from (ii). 
\end{proof}


\subsection{Self-testing for probabilistic assignments}
\label{ss_probassicont}

Recall \cite{acin-etal} that a contextuality scenario is a hypergraph $\bb{G}$ 
with vertex set $X$ (whose elements are thought of as outcomes of an experiment)
and hyperedge set $E$ (whose elements are thought of as subsets of jointly measurable outcomes) with the property that 
every vertex is contained in at least one edge; we write $\bb{G} = (X,E)$. 
In this section, we show how our general setup offers a broad 
context for self-testing for contextuality scenarios, hosting, in 
particular, the Main Theorem of \cite{brv-etal}. 

A \textit{probabilistic assignment} of a contextuality scenario $\bb{G} = (X,E)$ 
is a function $p : V \rightarrow [0,1]$ such that
$$    \sum_{x \in e}p(x)=1, \; \text{ for every } \; e \in E.$$
The set of all probabilistic assignments of $\bb{G}$ will be denoted by $\cl G(\bb{G})$; 
we have that $ \Gg(\bb{G})$ is a convex, possibly  empty, subset of $ \bb{R}^{V}$.
We stress that the term \emph{probabilistic model} in the place of 
probabilistic assignment was used in \cite{acin-etal}, but we have chosen 
to employ the latter term because of 
clash of terminology with the established use of the 
term \emph{model} in self-testing \cite{mps, pszz}. 
A contextuality scenario that admits a probabilistic assignment will be called \textit{non-trivial};
all contextuality scenarios herein will be assumed to be non-trivial.

A \textit{positive operator representation} \textit{(POR)} of 
a contextuality scenario $\G=(V,E)$ on a 
Hilbert space $H$ is a collection $ (A_{x})_{x\in V} \subseteq \cl B(H)$ of 
positive operators, such that 
$$     \sum_{x\in e} A_{x} = I_{H}, \; \text{ for every } \; e\in E.$$
The POR $ (A_{x})_{x\in V}$ is called a 
\textit{projective representation} \textit{(PR)} of $\G=(V,E)$ if 
$A_x$ is a projection, $x\in V$.

It was shown in \cite{act} that, given a 
non-trivial contextuality scenario $\G = (V,E)$, 
there exists an operator system $\cl S_{\bb{G}}$ , generated by 
positive elements $e_x$, $x\in X$, such that the unital completely positive maps
$\phi : \cl S_{\bb{G}}\to \cl B(H)$ correspond to POR's $(A_x)_{x\in V}$ of $\G$ via the 
assignment $\phi(e_x) = A_x$, $x\in V$; in particular, the probabilistic models of $\G$
correspond in a one-to-one fashion to the states of $\cl S_{\G}$. 
In addition, in \cite{act}, the universal C*-cover $C^*_u(\cl S_{\bb{G}})$ 
was explicitly identified.

Let $C^*(\bb{G})$ be the universal unital C*-algebra, generated by 
projections $p_x$, $x\in V$, subject to the relations 
$\sum_{x\in e} p_x = 1$, $e\in E$, and let 
$$\cl T_{\bb{G}} = {\rm span}\{p_x : x\in V\},$$
considered as an operator subsystem of $C^*(\bb{G})$. 
By the definition of $C^*(\bb{G})$, 
the *-representations $\pi : C^*(\bb{G})\to \cl B(H)$ correspond to 
projective representations of $\bb{G}$ on the Hilbert space $H$, 
via the assignment $\pi(p_x) = P_x$, $x\in V$. 
Further \cite{act}, the unital completely positive maps 
$\phi : \cl T_{\bb{G}}\to \cl B(H)$ correspond, via the assignment $\phi(e_x) = A_x$, $x\in V$, precisely to POR's $(A_x)_{x\in V}$ of $\G$ that are 
\emph{dilatable} in that there exists a Hilbert space $K$, an isometry 
$V : H\to K$ and a PR $(P_x)_{x\in V}$ of $\bb{G}$ on $K$ such that 
$A_x = V^*P_x V$, $x\in V$.

Let $\G=(V,E)$ and $ \Hbb =(W,F)$ be contextuality scenarios.  
The pairs $(\cl S_{\bb{G}},\cl S_{\bb{H}})$ and 
$(\cl T_{\bb{G}},\cl T_{\bb{H}})$ 
determine finitary frameworks for self-testing. 
We point out that, due to existence of a canonical 
unital completely positive map 
$\iota_{\bb{G}} : \cl S_{\bb{G}} \to \cl T_{\bb{G}}$
(resp. $\iota_{\bb{G}} : \cl S_{\bb{H}} \to \cl T_{\bb{H}}$) \cite{act}, 
self-testing over the pair $(\cl S_{\bb{G}},\cl S_{\bb{H}})$ is more 
general than over that over $(\cl T_{\bb{G}},\cl T_{\bb{H}})$. 

A commuting operator model 
$S = (H,(A_x)_{x\in V},(B_y)_{y\in W}, \xi)$ of $(\cl S_{\bb{G}},\cl S_{\bb{H}})$
gives rise, via (\ref{eq_pij}), to a correlation $p_S$  which
is a probabilistic assignment 
for the \emph{product} contextuality scenario
$ \G \times \Hbb = (V\times W,E \dot{\times} F)$, 
where $E \dot{\times} F = \{e\times f : e\in E, f\in F\}$.
Further, $p = p_S$ is \emph{no-signalling} in that 
$$    \sum_{x\in e}p(x,y) = \sum_{x\in e'}p(x,y), \; \;  y \in W, \;  e,e' \in E,$$
and
$$    \sum_{y\in f}p(x,y) = \sum_{y\in f'}p(x,y), \; \;  x \in V, \;  f,f' \in F$$
(see \cite{acin-etal}). 
We let $\Gg_{\rm ns}(\G,\Hbb)$ be 
the set of all no-signalling probabilistic assignments of $\GH$
(and write $\Gg_{\rm ns} = \Gg_{\rm ns}(\G,\Hbb)$ in case 
no confusion arises).
We note that the finitary framework $(\cl S_{\bb{G}},\cl S_{\bb{H}})$
gives rise to self-testing that includes as a special case the 
POVM self-testing (see Section \ref{ss_POVMNS}); indeed, 
if $X$ and $A$ are finite sets, 
the \emph{Bell scenario} $\bb{B}_{X,A} = (V,E)$, defined by letting
$V = X\times A$ and 
$E = \{\{x\}\times A : x\in X\}$ has the property that 
$\cl S_{\bb{B}_{X,A}} = \cl T_{\bb{B}_{X,A}} = \cl S_{X,A}$ \cite{act}. 

We now assume that 
$\bb{H}_{\circ} = (W,F)$, where $|W| = 1$ and $F = \{W\}$; in this case,
$\cl S_{\bb{H}_{\circ}} = \cl T_{\bb{H}_{\circ}} = \bb{C}$ and, hence, 
$\cl T_{\bb{G}} \otimes_{\rm c} \cl T_{\bb{H}_{\circ}} = \cl T_{\bb{G}}$. 
We will see that, even in 
this case, self-testing over the pair $(\cl T_{\bb{G}},\cl T_{\bb{H}_{\circ}})$
may be non-trivial. 
We note that the projective models for a state $s\in S(\cl T_{\bb{G}})$ now have the simpler form 
$M = (H,(P_x)_{x\in V},\xi)$, where $H$ is a Hilbert space, $(P_x)_{x\in V}$ is a 
PR of $\bb{G}$, and $\xi\in H$ is a unit vector; the corresponding state 
$s_M\in S(\cl T_{\bb{G}})$
has the form $s_M(p_x) = \langle P_x\xi,\xi\rangle$, $x\in V$. 

Let $n$ be an odd number with $n\geq 3$ and $\bb{G}_n = (V_n,E_n)$, where 
$$V_n = \{x_i, x_{i,i+1} : i \in \bb{Z}_n\} \ \mbox{ and } \ 
E_n = \{\{x_i,x_{i+1},x_{i,i+1}\} : i\in \bb{Z}_n\}.$$
Let $\frak{M}^{(1)}$ be the set of all PR models 
$(A_x)_{x\in V_n}$ for which ${\rm rank} (A_{x_i}) = 1$
for all $i\in \bb{Z}_n$, 
$$\frak{M}^{(1)}_{\max} = \left\{M\in \frak{M}^{(1)} : 
\sum_{i\in \bb{Z}_n} s_M(p_{x_i}) = 
\max_{M'\in \frak{M}^{(1)}} \sum_{i\in \bb{Z}_n} s_{M'}(p_{x_i})\right\},
$$
and 
$\cl S^{(1)} = \{s_M : M\in \frak{M}^{(1)}_{\max}\}$.
The following result is mainly a rephrasing of 
\cite[Main Theorem]{brv-etal} 
in the setup of the present paper.

\begin{proposition}\label{p_consel}
Let $n$ be an odd number with $n \geq 3$. 
Then $\cl S^{(1)}$ admits a self-test. 
\end{proposition}

\begin{proof}
Each model in $\frak{M}^{(1)}$ has the form 
$(H, (P_i)_{i\in \bb{Z}_n}, (P_{i,i+1})_{i\in \bb{Z}_n},\xi)$, where ${\rm rank}(P_i) = 1$ for each $i\in \bb{Z}_n$, and
$$P_i + P_{i+1} + P_{i,i+1} = I_H, \ \ \ i\in \bb{Z}_n.$$
In particular, we have that 
\begin{equation}\label{eq_ineqI3}
P_i + P_{i+1}\leq I_H, \ \ \ i\in \bb{Z}_n.
\end{equation}
Conversely, if we are given rank one projections $(P_i)_{i\in \bb{Z}_n}$ on $H$ 
for which the inequalities (\ref{eq_ineqI3}) are fulfilled then, setting 
$P_{i,i+1} := I_H - P_i - P_{i+1}$, we obtain a model 
in $\frak{M}^{(1)}$. 
The claim now follows from 
\cite[Main Theorem]{brv-etal}. 
\end{proof}

\subsection{A self-test for full graph colourings}\label{s:qg}

In this subsection we consider self-testing for 
classical-to-quantum no-signalling (CQNS) correlations, 
exhibiting an example for the classical-to-quantum game 
of complete graph colouring. 

Let $d\in \bb{N}$ with $d\geq 2$, and
$$\cl B_{d^2,d} := \underbrace{M_d\ast_1 \cdots\ast_1 M_d}_{d^2 \hspace{0.05cm} \mbox{\tiny times}},$$ 
where the free product is amalgamated over the units.
Let $\{\epsilon_{x,a,a'} : a,a'\in [d]\}$ be the standard matrix units of 
the $x$-th copy of $M_d$ in $\cl B_{d^2,d}$. 
Further, let 
\begin{equation}\label{eq_Rd2d}
\cl R_{d^2,d} := {\rm span}\{\{\epsilon_{x,a,a'} : x\in [d]^2, a,a'\in [d]\},
\end{equation}
viewed as an operator subsystem of $\cl B_{d^2,d}$, 
thus obtaining a finitary context $(\cl R_{d^2,d},\cl R_{d^2,d})$. 
Each state $s : \cl R_{d^2,d}\otimes_{\rm c} \cl R_{d^2,d} \to \bb{C}$
gives rise to a trace preserving completely positive 
map $\Gamma_s : \cl D_d\otimes \cl D_d\to M_d\otimes M_d$, 
defined by 
$$\Gamma_s(\epsilon_{x,x}\otimes\epsilon_{y,y}) 
= \sum_{a,a'=1}^d \sum_{b,b'=1}^d s(e_{x,a,a'}\otimes e_{y,b'b}))\epsilon_{a,a'}\otimes\epsilon_{b,b'}, 
\ \ \ x,y\in [d^2],$$
which is a \emph{classical-to-quantum no-signalling (CQNS) correlation}
in that it 
satisfies a natural version of the no-signalling conditions 
(\ref{eq_NSQne1}) and (\ref{eq_NSQne2}) 
and which is further of \emph{quantum commuting type}, denoted $\cl{CQ}_{\rm qc}$
(see \cite[Section 7]{tt-QNS} for the precise definitions).

Let $\cl K_d$ be the complete classical graph on $d$ vertices, and 
$\cl Q_d$ be the complete quantum graph on $d$ vertices, that is, 
$\cl Q_d = \{\Omega_d\}^\perp$, as a 
subspace of the Hilbert space $\mathbb C^d\otimes\mathbb C^d$, 
where $\Omega_d = \frac{1}{\sqrt{d}} \sum_{a=1}^d e_a\otimes e_a$
is the maximally entangled unit vector in dimension $d$. 
The \emph{graph homomorphism game} $\cl K_{d^2} \mapsto \cl Q_d$
was defined in \cite{tt-QNS}. Its \emph{game algebra} is the 
universal C*-algebra ${\rm Hom}(\cl K_{d^2}, \cl Q_d)$, 
generated by elements $e_{x,a,a'}$, $x\in[d^2]$, $a,a'\in[d]$,
satisfying the relations
\begin{itemize}
\item[(i)] $e_{x,a,a'}e_{x,b',b}=\delta_{a',b'}e_{x,a,b}$, $\sum_{a=1}^d e_{x,a,a}=1$ for all $x\in[d^2]$;
\item[(ii)] $\sum_{a,b=1}^d e_{x,a,b}e_{y,b,a}=0$ if $x\ne y$.
\end{itemize}
We note that ${\rm Hom}(\cl K_{d^2}, \cl Q_d)$ is a
quotient of $\cl B_{d^2,d}$, realised via the map 
$\epsilon_{x,a,a'}\mapsto e_{x,a,a'}$. 

\begin{remark}\label{r_unipa}
\rm 
Suppose that $H$ is a finite dimensional Hilbert space and let 
$\pi : {\rm Hom}(\cl K_{d^2}, \cl Q_d)\to \cl B(H)$ be a 
unital *-representation.  
As the subalgebra generated by $e_{x,a,a'}$ is isomorphic to $M_d$, 
for each $x\in[d^2]$ there exists a (finite dimensional)
Hilbert space $K_x$ and unitaries 
$V_x : H\to  \mathbb C^d\otimes K_x$ such that 
$$\pi(e_{x,a,a'}) = V_x^*(\epsilon_{a,a'}\otimes I_{K_x})V_x,
\ \ \ x\in [d^2], a,a'\in [d].$$
Since $\pi$ is unital, ${\rm dim}(K_x)$ is constant across $x\in [d^2]$. 
Applying a further unitary, we may assume that 
there exists a Hilbert space $K$ with $K_x = K$ for all $x\in [d^2]$. 

By relation (ii) (see the paragraph before the formulation 
of Remark \ref{r_unipa}), 
$\sum_{a,b=1}^d \pi(e_{x,a,b})\pi(e_{y,b,a})=0$ whenever $x\ne y$.
It follows that 
$$\sum_{a,b=1}^d (\epsilon_{a,b}\otimes 1)V_xV_y^*(\epsilon_{b,a}\otimes 1) = 0, \ \ \ x,y\in [d^2], \ x\neq y,$$
and hence, taking partial trace $\Tr_{M_d}$ along $M_d$, we obtain
\begin{equation}\label{trace}
({\rm tr}_d\otimes{\rm id}_{\cl B(K)})(V_xV_y^*) = 0, \ \ \ 
x,y\in [d^2], \ x\neq y.
\end{equation}
\end{remark}

\begin{remark}\label{r_krantnd}
\rm 
Suppose that $\cl A$ is a finite dimensional C*-algebra
and $\pi : {\rm Hom}(\cl K_{d^2}, \cl Q_d)\to \cl A$ is a 
unital *-homomorphism. Then, up to a unital *-isomorphism, 
$\cl A = \oplus_{i=1}^k M_{d n_i}$ for some $k\in \bb{N}$
and some $n_i\in \bb{N}$, $i\in [k]$.
Indeed, since ${\rm Hom}(\cl K_{d^2}, \cl Q_d)$ is a quotient
of $\cl B_{d^2,d}$, the *-homomorpshism $\pi$ gives rise to a 
unital *-homomorphism $\tilde{\pi} : \cl B_{d^2,d}\to \cl A$, 
and therefore to a unital *-homomorphism $\tilde{\pi}_0 : M_d\to \cl A$.
Assume, without loss of generality, that 
$\cl A = \oplus_{i=1}^k M_{m_i}$ for some $k\in \bb{N}$
and some $m_i\in \bb{N}$, $i\in [k]$.
The projection ${\rm proj}_i : \cl A \to M_{m_i}$ is a unital 
*-homomorphism and hence ${\rm proj}_i \circ \tilde{\pi}_0 : 
M_d\to M_{m_i}$ is a unital *-homomorphism. 
It follows that $d | m_i$, $i\in [k]$.
\end{remark}

Following \cite{tt-QNS}, a \emph{perfect quantum commuting strategy} 
(resp. \emph{perfect quantum strategy})
for the 
graph homomorphism game $\cl K_{d^2}\rightarrow \cl Q_d$ 
is a CQNS correlation of quantum commuting (resp.  quantum) type
$\Gamma : \cl D_{d^2} \otimes \cl D_{d^2}\to M_d\otimes M_d$, such that 
$\Gamma(\epsilon_{x,x}\otimes \epsilon_{y,y})$ is supported in 
$\cl Q_d$ whenever $x\neq y$, and 
$\Gamma(\epsilon_{x,x}\otimes \epsilon_{x,x}) = \Omega_d\Omega_d^*$ for every 
$x\in [d^2]$. 
By \cite{bhtt-JFA, tt-QNS}, a quantum CQNS correlation 
$\Gamma : \cl D_{d^2} \otimes \cl D_{d^2}\to M_d\otimes M_d$ 
is a perfect quantum commuting strategy for $\cl K_{d^2}\rightarrow \cl Q_d$ 
if and only if there exist a 
tracial  C*-algebra  $(\cl A, \tau)$ 
and a unital *-homomorphism 
$\pi: {\rm Hom}(K_{d^2}, \cl Q_d)\to \cl A$, such that 
$$\Gamma(\epsilon_{x,x}\otimes\epsilon_{y,y}) 
= \sum_{a,a'=1}^d \sum_{b,b'=1}^d \tau(\pi(e_{x,a,a'}e_{y,b'b}))\epsilon_{a,a'}\otimes\epsilon_{b,b'}, 
\ \ \ x,y\in [d^2];$$
we write $\Gamma = \Gamma^{\pi,\tau}$.

Recall that a \emph{unitary error basis} in $M_d$ 
is a basis of $M_d$ with respect to the trace inner product that consists of 
unitaries, in other words, a collection $\{u_i\}_{i=1}^{d^2}$ of unitaries in $M_d$, such that 
${\rm tr}_d(u_i u_j) = \delta_{i,j}$, $i,j\in [d^2]$. 
In the rest of this subsection, we restrict to the case $d = 2$. 
Recall the set 
\begin{equation}\label{eq_Pauliag}
\cl P=\left\{\left(\begin{array}{cc} 1&0\\0&1\end{array}\right), \left(\begin{array}{cc} 0&1\\1&0\end{array}\right), \left(\begin{array}{cc} 1&0\\0&-1\end{array}\right), \left(\begin{array}{cc} 0&-i\\i&0\end{array}\right)\right\}
\end{equation}
of Pauli matrices in $M_2$; 
it will be convenient to temporarily denote them by  
$U_x$, $x\in[4]$, in the order they appear in (\ref{eq_Pauliag}). 
We note that, by \cite{kr}, 
if $\cl E$ is a unitary error basis in $M_2$ then
there exist unitary matrices $R$, $T\in M_2$ and constants $c_V$, $V\in \cl P$, such that
$\cl E = \{c_VRVT: V\in\cl P\}$.

Let 
$\pi_{\cl K_4} : \cl B_{d^2,d} \to M_2$ be the 
*-homomorphism, given by 
$\pi_{\cl K_4}(e_{x,a,a'}) = U_x^*\epsilon_{a,a'}U_x$, and 
$\Gamma_{\cl K_4} : \cl D_4\otimes \cl D_4\to M_2\otimes M_2$ 
be the CQNS correlation, given by  
$$\Gamma_{\cl K_4}(\epsilon_{x,x}\otimes\epsilon_{y,y})
= \sum_{a,a'=1}^2 \sum_{b,b'=1}^2 {\rm tr_2}
(\pi_{\cl K_4}(e_{x,a,a'})\pi_{\cl K_4}(e_{y,b',b}))\epsilon_{a,a'}\otimes\epsilon_{b,b'},$$
where $x,y\in [4]$, 
and note that $\Gamma_{\cl K_4}$ is a perfect quantum strategy 
for the game $\cl K_{4}\rightarrow \cl Q_2$.

\begin{theorem}\label{repr}
Let $(\cl A,\tau)$ be a tracial von Neumann algebra with a faithful trace $\tau$, 
and $\pi :  {\rm Hom}(\cl K_{4}, \cl Q_2)\to \cl A$ be 
a unital *-homomorphism, such that 
$\Gamma^{\pi,\tau} = \Gamma_{\cl K_4}$. 
Then there exist a tracial von Neumann algebra $(\cl N,\tau_{\cl N})$, a trace preserving $*$-isomorphism $\rho: \cl A\to M_2\otimes\cl N$ and a unitary $V\in M_2\otimes \cl N$ such that 

\begin{equation}\label{eq_multimp}
V^*\rho(\pi(e_{x,a,a'}))V = U_x^*\epsilon_{a,a'}U_x\otimes 1_{\cl N}, 
\ \ x\in [4], \ a,a'\in [2]. 
\end{equation}
\end{theorem}

\begin{proof}
Fix $x\in [4]$. Then $\{\pi(e_{x,a,a'})\}_{a,a'}$ is a system of matrix units in $\cl A$. Let $q=\pi(e_{x,1,1})$. Then for $m\in\cl A$ we have that
$m_{i,j}:=\pi(e_{x,1,i})m\pi(e_{x,j,1})$ is in $q\cl A q$ and the map
$$\rho:\cl A\to M_2\otimes q\cl A q, \ \ m\mapsto\sum_{i,j=1}^2\epsilon_{i,j}\otimes m_{i,j}$$
is a normal $*$-isomorphism.  Let $\cl N=q\cl A q$ and $\tau_{\cl N}$ be the restriction of $\tau$ to $q\cl A q$.
Then 
\begin{eqnarray*}
    (\Tr_2\otimes\tau_{\cl N})(\rho(m))=\sum_{i=1}^2\tau_{\cl N}(m_{i,i})=\sum_{i=1}^2\tau(m\pi(e_{x,i,1})\pi(e_{x,1,i}))=\tau(m).
\end{eqnarray*}
Thus $\rho$ is a trace preserving $*$-isomorphism and hence $\Gamma^{\pi,\tau}=\Gamma^{\rho\circ\pi,\Tr_2\otimes\tau_{\cl N}}$. For simplicity of notation identify now $\cl A$ and $M_2\otimes\cl N$. 
As $\{\pi(e_{x,a,a'})\}_{a,a'}$, $x\in [4]$, and $\{\epsilon_{a,a'}\otimes 1_{\cl N}\}_{a,a'}$ are systems of matrix units in $M_2\otimes\cl N$, by \cite[Lemma 2.1]{haagerup-musat}, there exists a unitary $V_x\in M_2\otimes\cl N$, $x\in [4]$, such that $$\pi(e_{x,a,a'})=V_x^*(\epsilon_{a,a'}\otimes 1_{\cl N})V_x.$$

A direct calculation shows that 
$$\Tr_2(U_1^*\epsilon_{a,a'}U_1U_3^*\epsilon_{b',b}U_3)=\left\{\begin{array}{ll}
1 & \text{if } a=b=a'=b',\\
-1 & \text{if } a=b\ne a'=b', \\
0& \text{otherwise.}
\end{array}\right.$$
By (\ref{trace}),  
\begin{equation}\label{trace2}
(\tr_2\otimes{\rm id}_{\cl N})(V_{x}V_{y}^*)=0, 
\ \ \ x\in [4], a,a'\in [2].
\end{equation}
Since $\Gamma^{\pi,\tau} = \Gamma_{\cl K_4}$, we have that 
\begin{equation}\label{eq_oneofaha}
(\Tr_{2}\otimes\tau_{\cl N})((\epsilon_{a,a}\otimes I_n)V_{1}V_{3}^*(\epsilon_{a,a}\otimes I_n)V_{3}V_{1}^*)=1, 
\ \ a\in [2], 
\end{equation}
and 
$$(\Tr_2\otimes\tau_{\cl N})((\epsilon_{a,a'}\otimes I_n)V_{1}V_{3}^*(\epsilon_{a',a}\otimes I_n)V_{3}V_{1}^*)=-1, 
\  a,a'\in [2], a\neq a'.$$
Writing $V_{3}V_{1}^*$ as a $2\times 2$-block matrix $$V_{3}V_{1}^*=\left(\begin{array}{cc}A_{3}&B_{3}\\C_{3}&D_{3}\end{array}\right),$$
with $A_3$, $B_3$, $C_3$ and $D_3$ in $\cl N$, 
by (\ref{eq_oneofaha}) 
we have that
$$\tau_{\cl N}(A_{3}^*A_{3})=\tau_{\cl N}(D_{3}^*D_{3})=1.$$ 
As $V_{3}V_{1}^*$ is unitary, 
$A_{3}^*A_{3}+C_{3}^*C_{3}=1$, giving 
$\tau_{\cl N}(C_{3}^*C_{3})=0$ and hence, since $\tau_{\cl N}$ is faithful, $C_{3}=0$. Similarly, $B_{3}=0$.  
From (\ref{trace2}) we get $A_{3}+D_{3}=0$, that is  $V_{3}V_{1}^*=\left(\begin{array}{cc}1&0\\0&-1\end{array}\right)\otimes u$, where $u$ is unitary in $\cl N$. 
Next, writing $V_{x}V_{1}^*$, $x=2,4$, in a block form
$$V_{x}V_{1}^*=\left(\begin{array}{cc}A_{x}&B_{x}\\C_{x}&D_{x}\end{array}\right),$$
applying (\ref{trace2}) to $V_{x}V_{1}^*$ and $V_{x}V_{3}^*=V_{x}V_{1}^*(V_{3}V_{1}^*)^*$, we get $A_{x}+D_{x}=0$ and 
$A_{x}u^*-D_{x}u^*=0$ and hence $A_{x}=D_{x}=0$ for $x=2,4$.

Next we observe that 
$$\Tr_2(U_1^*\epsilon_{a,a'}U_1U_2^*\epsilon_{b',b}U_2)=\left\{\begin{array}{ll}1, & a=a'\ne b=b',\\
1, & a=b'\ne a'=b\\
0,& \text{otherwise,}\end{array}\right.$$
implying $$\tau_{\cl N}(C_{2}^*B_{2})=\tau_{\cl N}(B_{2}^*C_{2})=1$$ and $$\tau_{\cl N}(B_{2}^*B_{2})=\tau_{\cl N}(C_{2}^*C_{2})=1.$$ By Cauchy-Schwartz, we obtain that $C_{2} = \theta B_{2}$ for a unimodular constant $\theta$; 
as $\tau_{\cl N}(C_{2}^*B_{2})=e^{-i\varphi}\tau_{\cl N}(B_{2}^*B_{2})$, we have that $\theta = 1$ and
$$V_{2}V_{1}^*=\left(\begin{array}{cc}0&1\\1&0\end{array}\right)\otimes v$$
for a unitary $v\in\cl N$. 
Condition (\ref{trace2}) applied to $V_{2}V_{4}^*=V_{2}V_{1}^*(V_{4}V_{1}^*)^*$ gives $vB_{4}^*+vC_{4}^*=0$ and hence $B_{4}=-C_{4}$. Therefore, 

$$V_{4}V_{1}^*=\left(\begin{array}{cc}0&-i\\i&0\end{array}\right)\otimes w$$
for a unitary $w\in\cl N$. 
With this at hand, we conclude that, if  
$V = V_{1}$, then 
$$V\pi(e_{x,a,a'})V^*=(V_{x}V_{1}^*)^*(\epsilon_{a,a'}\otimes 1_{\cl N})V_{x}V_{1}^*=U_x^*\epsilon_{a,a'}U_x\otimes 1_{\cl N},$$
$x\in [4]$, $a,a' \in [2]$, 
completing the proof. 
\end{proof}

It follows from \cite[Lemma 9.2]{tt-QNS} that there 
exists a *-isomorphism $\partial: \cl B_{4,2}\to\cl B_{4,2}^{\rm op}$, such that 
$\partial(e_{x,a,a'}) = e_{x,a',a}^{\rm op}$, 
$x\in [4]$, $a,a'\in [2]$. 
For
$w = e_{x_1,a_1,a_1'}\cdots e_{x_k,a_k,a_k'}$, for some $x_i\in [4]$, $a_i,a_i'\in [2]$, 
$i = 1,\dots,k$,
set 
$$\bar{w}:= \partial^{-1}(w^{\rm op}) = e_{x_k,a_k',a_k}\cdots e_{x_1,a_1',a_1}.$$

Recalling the definition (\ref{eq_Rd2d}), 
let $\cl S$ be the set of states of 
$C^*_u(\cl R_{4,2})\otimes_{\rm max} C^*_u(\cl R_{4,2})$
that factor through ${\rm Hom}(\cl K_{4}, \cl Q_2)\otimes_{\rm max} {\rm Hom}(\cl K_{4}, \cl Q_2)$. 
Let $\tilde{s}_{\cl K_4}\in\cl S$ be the state given by 
$$\tilde{s}_{\cl K_4}(u\otimes v)
= {\rm tr}_2(\pi_{\cl K_4}(q_{\cl K_4}(q_u(u)\overline{q_u(v)}))), \ \ \ u,v\in C^*_u(\cl R_{4,2}),$$
where $q_u : C^*_u(\cl R_{4,2})\to \cl B_{4,2}$ and 
$q_{\cl K_4} : \cl B_{4,2}\to {\rm Hom}(\cl K_{4}, \cl Q_2)$
are the quotient maps, 
and let $s_{\cl K_4}$ be the restriction of 
$\tilde{s}_{\cl K_4}$ to $\cl R_{4,2}\otimes_{\rm c} \cl R_{4,2}$.

\begin{corollary}\label{c:s_k}
    The state $s_{\cl K_4}$ is an abstract self-test for 
    the family $\cl S$. 
\end{corollary}

\begin{proof}
It suffices to show that $s_{\cl K_4}$ has a 
unique extension to a state 
$\phi : \cl B_{4,2}\otimes_{\max} \cl B_{4,2} \to \bb{C}$. Fix such an extension $\phi$ and note 
that the canonical correlation $\Gamma_{\phi} : \cl D_4\otimes \cl D_4\to M_2\otimes M_2$ associated with $\phi$
coincides with $\Gamma_{s_{\cl K_4}}$. 
By \cite[Theorem 3.2]{bhtt-JFA}, there exists a 
tracial von Neumann 
algebra $(\cl A,\tau)$ and a *-representation 
$\pi : \cl B_{4,2}\to \cl A$,  
such that 
$$\phi(u\otimes v) = \tau(\pi(u\bar{v})), \ \ \ 
u,v\in \cl B_{4,2}.$$
Further, $\pi$ canonically descends to a *-representation 
$\tilde{\pi} : {\rm Hom}(\cl K_{4}, \cl Q_2)\to \cl A$, 
such that, if 
$\tilde{\phi} : {\rm Hom}(\cl K_{4}, \cl Q_2)\otimes_{\max} {\rm Hom}(\cl K_{4}, \cl Q_2)\to \bb{C}$ is the 
canonical functional, arising from $\phi$ in view of the definition of the class $\cl S$, we have that 
$$\tilde{\phi}(u\otimes v) = \tau(\tilde{\pi}(u\bar{v})), \ \ \ 
u,v\in {\rm Hom}(\cl K_{4}, \cl Q_2).$$
Let $\rho : \cl A\to M_2\otimes \cl N$
be the *-homomorphism, arising from 
Theorem \ref{repr} (here $(\cl N,\tau_{\cl N})$
is a tracial von Neumann algebra). 
Then, for $w\in {\rm Hom}(\cl K_{4}, \cl Q_2)$, 
we have that 
$$\tilde{\phi}(w) = \tau(\tilde{\pi}(w)) = 
(\tau_2\otimes \tau_{\cl N})(\pi_{\cl K_4}(w)\otimes 1_{\cl N})
= \tau_2(\pi_{\cl K_4}(w)).$$
The proof is complete.
\end{proof}

Let $\cl C$ be the class of quantum  models 
$(H_A\otimes H_B, \varphi_A, \varphi_B,\xi)$ of CQNS correlations, where
$H_A$ and $H_B$ are finite-dimensional Hilbert spaces, 
$H_A\otimes H_B = \mbox{}_{\cl B(H_A)}H_{\cl B(H_B)^o}$ is a bipartite 
system defined by the canonical bimodule structure of $H_A\otimes H_B$, 
and $\varphi_A$ and $\varphi_B$ extend to representations $\pi_A$ and $\pi_B$ of $\cl A:=\cl B_{4,2}$ and $\cl B:=\cl B_{4,2}$, respectively.

\begin{proposition}\label{p:s_k}
The state $s_{\cl K_4}$ is a selftest for $\cl C$.
\end{proposition}

\begin{proof}
   We will apply Theorem \ref{th_rev}, see also \cite[Theorem 4.12]{pszz}. First we observe that $s_{\cl K_4}$ is an extreme point in 
   the dual $(\cl R_{4,2}\otimes_{\rm min}\cl R_{4,2})^{\rm d}$
   of $\cl R_{4,2}\otimes_{\rm min}\cl R_{4,2}$. In fact, if $s_{\cl K_4} = \sum_{i=1}^n\alpha_is_i$, where $\alpha_i\geq 0$, $\sum_{i=1}^k\alpha_i=1$, then the corresponding correlations $\Gamma_{s_i}$ are perfect strategies for quantum colouring of $\cl K_4$. 
Indeed, note that, if $J = \Omega_2\Omega_2^*$ is the maximally entangled state  in $\mathbb C^2\otimes\mathbb C^2$, then
$$0=\Tr(\Gamma_{s_{\cl K_4}}(\epsilon_{x,x}\otimes\epsilon_{x,x})J^\perp)=\sum_{i=1}^n\alpha_i\Tr(\Gamma_{s_i}(\epsilon_{x,x}\otimes\epsilon_{x,x})J^\perp)$$
   and hence $\Tr(\Gamma_{s_i}(\epsilon_{x,x}\otimes\epsilon_{x,x})J^\perp)$,
   being non-negative, must be zero. Similar arguments show that 
   $\tr(\Gamma_{s_i}(\epsilon_{x,x}\otimes\epsilon_{y,y})J)=0$ for all $i$ and all $x\ne y$.

   Therefore $s(e_{x,a,a'}\otimes e_{y,b,b'})=\sum_{i=1}^n\alpha_i\tau_i(\pi_i(e_{x,a,a'})\pi_i(e_{y,b',b}))$, where $\pi_i: {\rm Hom}(\cl K_{4}, \cl Q_2)\to \cl A_i$ are 
   $*$-homomorphisms into finite dimensional algebras $\cl A_i$ with trace $\tau_i$. 
   By Proposition \ref{repr},
$$\tau_i(\pi_i(e_{x,a,a'})\pi_i(e_{y,b',b})) = {\rm tr}_2(U_x^*\epsilon_{a,a'}U_xU_y^*\epsilon_{b',b}U_y)$$ for all $i$ and hence $s$ is extreme. 
The state $s_{\cl K_4}$ has also a unique extension to $\cl B_{4,2}\otimes\cl B_{4,2}$ and hence by Theorem \ref{th_rev} it is a self-test. 
\end{proof}


\subsection{Schur products}\label{s:schur}

Let $G$ be a finite group, let $\pi:G\to\cl{U}(H_\pi)$ and $\rho : G\to\cl{U}(H_\rho)$ be irreducible representations of $G$, and let $\psi\in H_\pi\ten H_\rho$ be a unit vector. Then
$$u(s,t):=\la (\pi(s)\ten\rho(t))\psi,\psi\ra, \ \ \ s,t\in G,$$
is a normalised positive definite function on $G\times G$ and the associated 
\emph{Schur multiplier} $\Theta(u):\cl{B}(\ell^2(G))\ten \cl{B}(\ell^2(G))\to \cl{B}(\ell^2(G))\ten \cl{B}(\ell^2(G))$ is the unital quantum channel satisfying
\begin{equation}\label{e:schur}\Theta(u)(e_{s,s'}\ten e_{t,t'})=u(s^{-1}s',t^{-1}t')e_{s,s'}\ten e_{t,t'};
\end{equation}
if further clarity is needed, we write $\Theta(u) = \Theta_{\pi,\rho,\psi}$. 
Letting $\widetilde{E}_{s,s',g,g'}=\delta_{s,g}\delta_{s',g'}\pi(s^{-1}s')$ and $\widetilde{F}_{t,t',h,h'}=\delta_{t,h}\delta_{t',h'}\rho(t^{-1}t')$, 
it is straightforward to verify that 
$\widetilde{E} := (\widetilde{E}_{s,s',g,g'})_{s,s',g,g'}$ and 
$\widetilde{F} := (\widetilde{F}_{t,t',h,h'})_{t,t',h,h'}$ are unitaries and hence 
the quadruple 
$S_{\pi,\rho,\psi} = (H_\pi\ten H_\rho, \widetilde{E},\widetilde{F},\psi)$ is a $\cl{Q}_{\rm q}$-model for $\Theta(u)$.

A \textit{unistochastic operator matrix (USOM)} 
is a stochastic operator matrix $E=(E_{x,x',a,a'})\in M_X\ten M_X\ten \cl{B}(H)$ for which there exists a unitary $U:H^X\to H^X$ such that $E_{x,x',a,a'}=U_{a,x}^*U_{a',x'}$. 
A $\cl{Q}_{\rm q}$-model will be called {\it unitary} if it is defined via USOM's $(E_{s,s',g,g'})_{s,s',g,g'}$ and $(F_{t,t',h,h'})_{s,s',g,g'}$ (see Section \ref{uni}). In particular, we have that $S := S_{\pi,\rho,\psi}$ is a unitary model.  
If $H$ and $K$ are Hilbert spaces, we say that a vector $\psi\in H\otimes K$ is 
\emph{marginally cyclic} if 
$$\overline{(\cl B(H)\otimes 1)\psi} = \cl B(H\otimes K) = \overline{(1\otimes \cl B(K))\psi}.$$
Let 
$$\frak{M}(u) = \{S_{\pi'\rho',\psi'} : 
\mbox{ full rank unitary model s.t. }
\Theta_{\pi'\rho',\psi'} = \Theta(u)\}.$$

\begin{theorem}\label{p:prop} 
If $\psi$ is marginally cyclic in $H_{\pi}\otimes H_{\rho}$ and 
\begin{equation}\label{e:ext}\mathrm{span}\{(\pi(s)\ten\rho(t))\psi\psi^*(\pi(s^{-1})\ten\rho(t^{-1})) : s,t\in G\}=\cl{B}(H_\pi\ten H_\rho)\end{equation}
then $(H_\pi\ten H_\rho, \widetilde{E},\widetilde{F},\psi)$ is a self-test for $\frak{M}(u)$.
\end{theorem}

\begin{proof}
Suppose that
$(H_A\ten H_B,(E_{s,s',g,g'})_{s,s',g,g'},(F_{t,t',h,h'})_{t,t',h,h'},\xi)$ is a full rank unitary model for $\Theta(u)$ in $\cl {Q}_{\rm q}$. This means that $E_{s,s',g,g'} = U_{g,s}^*U_{g',s'}$ and $F_{t,t',h,h'}=V_{h,t}^*V_{h',t'}$ for some block operator unitaries $(U_{g,s})_{g,s} :H_A^G\to H_A^G$ and 
$(V_{h,t})_{h,s} : H_B^G\to H_B^G$, the reduced densities 
$(\id\ten\tr)(\xi\xi^*)$ and $(\tr\ten\id)(\xi\xi^*)$ have full rank, and 
$$\la (E_{s,s',g,g'}\ten F_{t,t',h,h'}) \xi,\xi\ra=\delta_{s,g}\delta_{s',g'}\delta_{t,h}\delta_{t',h'}\la\pi(s^{-1}s')\ten\rho(t^{-1}t')\psi,\psi\ra.$$
In particular,
\begin{equation}\label{e:1}
\la (U_{s',s'}\ten V_{t',t'})\xi,
(U_{s,s}\ten V_{t,t}) \xi\ra = \la (\pi(s')\ten\rho(t'))\psi, (\pi(s)\ten\rho(t))\psi\ra
\end{equation}
for all $s,s',t,t'\in G$.
We now observe that the unitaries $(U_{g,s})_{g,s}$ and 
$(V_{h,t})_{h,t}$ are necessarily diagonal. 
Since (under the trace-duality convention)
$$\Theta(u)(X^{\rm t})^{\rm t} = (\id\ten\tr_{2,4})((U\ten V)(X\ten\xi\xi^*)(U^*\ten V^*)),$$ 
setting $\rho_\xi = (\id\ten\tr)(\xi\xi^*)$, we have 
$$\Theta(u)(\rho^{\rm t}\ten 1)^{\rm t} = 
(\id\ten \tr)U(\rho\ten\rho_{\xi})U^*.$$ 
Since the transformation 
$\rho \mapsto \Theta(u)(\rho^{\rm t}\ten 1)^{\rm t}$ 
is a $\cl{D}_G$-bimodule map, if 
$\{e_i\}_{i=1}^{m_A}$ is an orthonormal basis for $H_A$, 
it follows that the Kraus operator
$$(\id\ten e_i^*)U((\cdot)\ten\sqrt{\rho_{\xi}}e_j):\ell^2(G)\to\ell^2(G)$$
belongs to $\cl{D}_G'=\cl{D}_G$ for each $i,j \in [m_A]$. Hence,
$$U(1\ten\sqrt{\rho_{\xi}})(x\ten 1)=(x\ten 1)U(1\ten\sqrt{\rho_{\xi}}), \ \ \ x\in \cl{D}_G.$$
Since $\rho_{\xi}$, and hence $\sqrt{\rho_{\xi}}$, has full rank, it follows that $U(x\ten 1)=(x\ten 1)U$, so that $U\in\cl {D}_G\ten\cl B(H_A)$. Hence, $U_{s,s'}=\delta_{s,s'}U_{s,s}$ with $U_s:=U_{s,s}\in\cl U(H_A)$. 
Without loss of generality, we may also assume $U_e = 1$ (indeed, if not,
redefine $U_s$ as $U_e^*U_s$, noting that $E_{s,s',g,g'}=U_{g,s}^*U_{e}U_{e}^*U_{g',s'}$).
An analogous argument shows that $V_{t,t'}=\delta_{t,t'}V_{t,t}$ with $V_t:=V_{t,t}\in\cl U(H_B)$ and we may assume that $V_e=1$. 

Let $\cl A$ (resp. $\cl B$) 
be the (unital) $C^*$-subalgebra of $\cl B(H_A)$ (resp. $\cl B(H_B)$)
generated by $\{U_s\}_{s\in G}$ (resp. $\{V_t\}_{t\in G}$). 
By the finite-dimensionality of $\cl A$ and $\cl B$, there exist unitaries $W_A:H_A\to\bigoplus_{i=1}^{n_A} H_A^i\ten K_A^i$ and $W_B:H_B\to\bigoplus_{j=1}^{n_B}H_B^j\ten K_B^j$ such that 
$$W_A a W_A^*=\bigoplus_{i=1}^{n_A} \sigma_A^i(a)\ten I_{K_A^i}, \ \ \ a\in \cl A,$$
and 
$$W_B b W_B^*=\bigoplus_{j=1}^{n_B} \sigma_B^j(a)\ten I_{K_B^j}, \ \ \ b\in \cl B,$$
with $\sigma_A^i$ and $\sigma_B^j$ irreducible representations of $\cl A$ and $\cl B$, respectively. Fix orthonormal bases $\{e^i_k\}_{k=1}^{d_A^i}$ and $\{e^j_l\}_{l=1}^{d_B^j}$ of $K_A^i$ and $K_B^j$, respectively. It follows from (\ref{e:1}) that
\begin{align*}&\la (U_{s'}\ten V_{t'})\xi,(U_s\ten V_t)\xi\ra\\
&=\sum_{i=1}^{n_A}\sum_{j=1}^{n_B}\sum_{k=1}^{d_A^i}\sum_{l=1}^{d_B^j}p_{i,j,k,l}\la(\sigma_A^i(U_{s'})
\hspace{-0.1cm}\ten \hspace{-0.1cm}\sigma_B^j(V_{t'}))\xi_{i,j,k,l},(\sigma_A^i(U_{s})\hspace{-0.1cm}\ten\hspace{-0.1cm}\sigma_B^j(V_{t}))\xi_{i,j,k,l}\ra,
\end{align*}
where $\xi_{i,j,k,l}\in H_A^i\ten H_B^j$ is the normalisation of $(\id\ten e_k^i\ten\id\ten e_l^j)^*\xi$, and 
$$p_{i,j,k,l}=\norm{(\id\ten e_k^i\ten \id \ten e_l^j)^*\xi}^2\geq 0,$$ 
and note that $\sum_{i,j,k,l}p_{i,j,k,l}=1$. Combining (\ref{e:schur}) and (\ref{e:1}), we obtain a convex decomposition of the channel $\Theta(u)$, which is extreme within the set of completely positive trace preserving maps by (\ref{e:ext}) (see e.g. \cite[Theorem 3]{LS}). Indeed, 
setting $\psi_{(s,t)}:=(\pi(s)\otimes\varphi(t))\psi$, $s,t\in G$, 
in the notation and terminology of \cite{LS}, $\{\psi_{(s,t)}\}_{(s,t)\in G\times G}$ is a full set of vectors by (\ref{e:ext}) and $C_{(s,t),(s',t')}:=u(s^{-1}s',t^{-1}t')=\langle \psi_{(s',t')},\psi_{(s,t)}\rangle$.
Thus, for every $\lambda=(i,j,k,l)\in\Lambda=\{(i,j,k,l) : p_{i,j,k,l} > 0\}$, we have
\begin{align*}&\la(\pi(s')\ten\rho(t'))\psi,(\pi(s)\ten\rho(t))\psi\ra\\
&\hspace{1cm} =\la(\sigma_A^i(U_{s'})\ten \sigma_B^j(V_{t'}))\xi_{i,j,k,l},(\sigma_A^i(U_{s})\ten\sigma_B^j(V_{t}))\xi_{i,j,k,l}\ra.
\end{align*}
It follows that 
$$W_\lambda:H_\pi\ten H_\rho\ni (\pi(s)\ten\rho(t))\psi\to(\sigma_A^i(U_{s})\ten\sigma_B^j(V_{t}))\xi_{i,j,k,l}\in H_A^i\ten H_B^j$$
is a well-defined isometry for each $\lambda\in\Lambda$. Recalling that $\psi$ is marginally cyclic 
and that $\pi$ and $\rho$ are irreducible (so $\mathrm{span}\{\pi(s) : s\in G\}=\cl{B}(H_\pi)$ and $\mathrm{span}\{\rho(t) : t\in G\}=\cl{B}(H_\rho)$) for any $\eta\in H_\pi\ten H_{\rho}$ and $s\in G$, for suitable scalars $c_t$, 
$t\in G$, we have 
\begin{align*}W_\lambda(\pi(s)\ten1)\eta&=\sum_{t\in G}c_tW_\lambda(\pi(s)\ten\rho(t))\psi\\
&=\sum_{t\in G}c_t (\sigma_A^i(U_{s})\ten\sigma_B^j(V_{t}))\xi_{i,j,k,l}\\
&=\sum_{t\in G}c_t (\sigma_A^i(U_{s})\ten I_{H_B^j})(\sigma_A^i(U_{e})\ten\sigma_B^j(V_{t}))\xi_{i,j,k,l}\\
&=\sum_{t\in G}c_t (\sigma_A^i(U_{s})\ten I_{H_B^j})W_\lambda(1\ten\rho(t))\psi\\
&=(\sigma_A^i(U_{s})\ten I_{H_B^j})W_\lambda\eta.
\end{align*}
Thus, $W_\lambda(\pi(s)\ten1)=(\sigma_A^i(U_{s})\ten I_{H_B^j})W_\lambda$, $s\in G$. Similarly, $W_\lambda(1\ten\rho(t))=(I_{H_A^i}\ten \sigma_B^j(V_t))W_\lambda$, $t\in G$, so that
$$W_\lambda(\pi(s)\ten\rho(t))=(\sigma_A^i(U_{s})\ten \sigma_B^j(V_t))W_\lambda, \ \ \ s,t\in G.$$
Then the unitaries $\sigma_A^i(U_{s})\ten \sigma_B^j(V_t)$ are respectively mapped to the unitaries $\pi(s)\ten \rho(t)$ under the unital completely positive map 
$\Phi_{\lambda} : \cl{B}(H_A^i\ten H_B^j)\to\cl{B}(H_\pi\ten H_\rho)$, 
given by $\Phi_{\lambda}(T) = W_\lambda^* T W_\lambda$.
Thus, the unitaries $\sigma_A^i(U_{s})\ten \sigma_B^j(V_t)$ 
belong to the multiplicative domain $\cl{M}$ of $\Phi_{\lambda}$. Hence, by \cite[Proposition 1.5.7]{BO}, $\cl{M}$ contains the $C^*$-algebra generated by $\sigma_A^i(U_{s})\ten \sigma_B^j(V_t)$, which, by the irreducibility of $\sigma_A^i$ and $\sigma_B^j$ (and the definitions of $\cl A$ and $\cl B$), is equal to $\cl{B}(H_A^i\ten H_B^j)$. So $\Phi_\lambda:\cl{B}(H_A^i\ten H_B^j)\to\cl{B}(H_\pi\ten H_\rho)$ is a unital $*$-homomorphism, necessarily injective by simplicity of $\cl{B}(H_A^i\ten H_B^j)$ and surjective by irreducibility of $\pi$ and $\rho$. Hence, $W_\lambda$ is unitary.

Moreover, the local intertwining properties above show that 
the maps $\sigma_A : \cl A \to \cl{B}(H_\pi)\ten I_{H_\rho}$
and $\sigma_B : \cl B \to I_{H_\pi}\ten\cl{B}(H_\rho)$, given by 
$$\sigma_A(a) = W_\lambda^*(\sigma_A^i(a)\ten I_{H_B^J})W_\lambda
\ \mbox{ and } \ 
\sigma_B(b) = W_\lambda^*(I_{H_A^i}\ten\sigma_B^j(b))W_\lambda$$
are irreducible representations of $\cl A$ and $\cl B$, respectively, and are independent of $\lambda$. Clearly, $\sigma_A\ten\sigma_B$ is unitarily equivalent to each $\sigma_A^i\ten \sigma_B^j$, (via $W_\lambda$) so we must have local unitary equivalence: $\sigma_A\cong \sigma_A^i$ and $\sigma_B\cong\sigma_B^j$, say, via $W_A^i:H_A^i\to H_\pi$ and $W_B^j:H_B^j\to H_\vphi$. Then $W_A^i\ten W_B^j=\alpha_{i,j}W_\lambda^*$ for some phase $\alpha_{i,j}\in\mathbb{T}$.  

We now follow ideas from \cite[Theorem 4.12]{pszz}. Let 
$$\Lambda_A=\{i : p_{i,j,k,l}>0 \ \textnormal{for some} \ j,k,l\},$$ 
and define $\Lambda_B$ similarly. Fix unit vectors $\eta_A\in H_\pi$ and $\eta_B\in H_\rho$. For $i\notin\Lambda_A$, let 
$T_A^i:H_A^i\ten K_A^i \to H_\pi\ten H_A^i\ten K_A^i$ be given by
$T_A^i(\eta) = \eta_A\ten\eta$, and define 
$T_B^j:H_B^j\ten K_B^j\to H_\rho\ten H_B^j\ten K_B^j$, for $j\notin\Lambda_B$, similarly. Finally, set
$$H_A^{\rm aux}:=\bigg(\bigoplus_{i\in\Lambda_A}K_A^i\bigg)\bigoplus\bigg(\bigoplus_{i\notin\Lambda_A} H_i^A\ten K_A^i\bigg),$$
and define the isometry $T_A:H_A\to H_\pi\ten H_A^{\rm aux}$ by
$$T_A=\bigg(\bigg(\bigoplus_{i\in\Lambda_A}W_A^i\ten I_{K_A^i}\bigg)\oplus\bigg(\bigoplus_{i\notin\Lambda_A}T_A^i\bigg)\bigg)\circ W_A.$$
Define $H_B^{\rm aux}$ and $T_B:H_B\to H_\rho\ten H_B^{\rm aux}$ similarly, 
and let 
$$\xi^{\rm aux} = 
\bigoplus_{i=1}^{n_A} \bigoplus_{j=1}^{n_B}
\sum_{k=1}^{d_A^i} \sum_{l=1}^{d_B^j}
\alpha_{i,j} \sqrt{p_{i,j,k,l}}e_{k}^i\ten e_l^j\in H_A^{\rm aux}\ten H_B^{\rm aux}.$$ 
By construction, for any $a\in\cl A$ and $b\in\cl B$ we have
$$(T_A\ten T_B)(a\ten b)\xi=((\sigma_A(a)\ten \sigma_B(b))\psi)\ten\xi^{\rm aux}.$$
In particular,
\begin{align*}
& (T_A\ten T_B)(E_{s,s',g,g'}\ten F_{t,t',h,h'})\xi\\
& =\delta_{s,g}\delta_{s',g'}\delta_{t,h}\delta_{t',h'}T_A\ten T_B(U_s^*U_{s'}\ten V_t^*V_{t'})\xi\\
& =\delta_{s,g}\delta_{s',g'}\delta_{t,h}\delta_{t',h'}(\pi(s)^*\pi(s')\ten \rho(t)^*\rho(t'))\psi\ten\xi^{\rm aux}\\
& =\delta_{s,g}\delta_{s',g'}\delta_{t,h}\delta_{t',h'}(\pi(s^{-1}s')\ten \rho(t^{-1}t'))\psi\ten\xi^{\rm aux}.
\end{align*}
Thus, the model $(H_\pi\ten H_\rho, \{\widetilde{E}_{s,s',g,g'}\},\{\widetilde{F}_{t,t',h,h'}\},\psi)$ is a self-test, as claimed.
\end{proof}

We now exhibit a class of examples satisfying the hypotheses of Proposition \ref{p:prop}.

\begin{example}\label{ex:S_3} Let $G=S_3$, the symmetric group on 3 points. $S_3$ has a two dimensional irreducible representation $\pi:S_3\rightarrow \cl{U}(\mathbb{C}^2)$ given by
\begin{gather*}
    \pi(e)=\begin{bmatrix} 1 & 0\\ 0 & 1 \end{bmatrix},\hs\hs
    \pi(123)=\begin{bmatrix} e^{i2\pi/3} & 0\\ 0 & e^{-i2\pi/3} \end{bmatrix},\hs\hs
    \pi(132)=\begin{bmatrix} e^{-i2\pi/3} & 0\\ 0 & e^{i2\pi/3} \end{bmatrix},\\
    \pi(12)=\begin{bmatrix} 0 & 1\\ 1 & 0 \end{bmatrix},\hs\hs
    \pi(23)=\begin{bmatrix} 0 & e^{-i2\pi/3}\\ e^{i2\pi/3} & 0 \end{bmatrix},\hs\hs
    \pi(13)=\begin{bmatrix} 0 & e^{i2\pi/3}\\ e^{-i2\pi/3} & 0 \end{bmatrix}.
\end{gather*}
In the argument below, we will interpret the above matrices as rotations on the Bloch sphere. To that end, recall that the single qubit rotation operators
$$R_{x}(\theta)=e^{-i\frac{\theta}{2}\sigma_x}, \ \ R_{y}(\theta)=e^{-i\frac{\theta}{2}\sigma_y}, \ \ R_{z}(\theta)=e^{-i\frac{\theta}{2}\sigma_z},$$
induce rotations of angle $\theta\in\mathbb{R}$ about the $x$, $y$, and $z$ axes, respectively, where $\sigma_x$, 
$\sigma_y$ and $\sigma_z$ are the $2\times 2$ Pauli matrices.
Then
$$\pi(123)=R_z(2\pi/3), \ \ \pi(132)=R_z(4\pi/3), \ \ \pi(12)=X=-iR_x(\pi).$$
If $\{e_0,e_1\}$ denotes the standard basis of $\mathbb{C}^2$, for $\theta\in(0,\pi/2)\cup(\pi/2,\pi)$, we let 
$\{e_\theta,f_\theta\}$ denote the following $y$-rotated basis:
$$e_{\theta}=R_y(\theta)e_0=\begin{bmatrix}\cos\frac{\theta}{2}\\\sin\frac{\theta}{2}\end{bmatrix}, \ \ f_{\theta}=R_y(\theta)e_1=\begin{bmatrix}-\sin\frac{\theta}{2}\\\cos\frac{\theta}{2}\end{bmatrix}.$$
Let $\psi:=\alpha e_\theta\ten e_\theta+\beta f_{\theta}\ten f_{\theta}\in\mathbb{C}^2\ten\mathbb{C}^2$ for fixed $\alpha,\beta\in\mathbb{C}$ satisfying $|\alpha|^2+|\beta|^2=1$ and $|\alpha|^2\in(0,1/2)$. Then $\psi$ has full Schmidt rank so is marginally cyclic. We now show that 
$$V:=\mathrm{span}\{(\pi(s)\ten\pi(t))\psi\psi^*(\pi(s^{-1})\ten\pi(t^{-1}))
: s,t\in S_3\}=\cl{B}(\mathbb{C}^2\ten \mathbb{C}^2),$$
thereby obtaining a self-testing $\cl{Q}_{\rm q}$-model for the channel $\Theta(u)$, where $u(s,t) = \la(\pi(s)\ten\pi(t))\psi,\psi\ra$, by Proposition \ref{p:prop}. 

First, by irreducibility of $\pi$ and the orthogonality relation (\cite[Theorem III.1.1]{simon}), 
$$\sum_{s\in S_3}\pi(s)(\cdot)\pi(s)^*=\tr(\cdot)1,$$
so that 
$$\pi(s)\rho\pi(s)^*\ten 1, \ 1\ten\pi(t)\rho\pi(t)^*\in V, \ \ \ s,t\in G,$$
where $\rho$ is the reduced density matrix of $\psi$:
$$\rho=(\id\ten\tr)(\psi\psi^*) = |\alpha|^2e_\theta e_\theta^*+|\beta|^2f_\theta f_\theta^*=(\tr\ten\id)(\psi\psi^*).$$
Therefore, it suffices to show that \begin{equation}\label{e:2}\mathrm{span}_{\mathbb{R}}\{\pi(s)\rho\pi(s)^* : s\in S_3\}=M_2(\mathbb{C})_{sa}.\end{equation}
Recall that the space of trace-1 self-adjoint $2\times 2$ matrices is affinely isomorphic to $\mathbb{R}^3$ via 
$$(x,y,z)\Leftrightarrow\begin{bmatrix} 1/2+z & x-iy\\x+iy & 1/2-z\end{bmatrix},$$
with the convex subset of density operators corresponding to the unit ball. In particular, we can visualize $\rho$ as the vector $r_\rho=(x_\rho,y_\rho,z_\rho)\in\mathbb{R}^3_{\norm{\cdot}\leq 1}$, which lies on the straight line between the points $r_{e_\theta e_\theta^*}$ and $r_{f_\theta f_\theta^*}$ on the boundary sphere. Both these latter points lie on the great circle in the $xz$-plane, and, since $\theta \in(0,\pi/2)\cup(\pi/2,\pi)$, are neither on the $z$-axis, nor the $xy$-plane. (Recall that $r_{e_0e_0^*}$ and $r_{e_1e_1^*}$ are the North and South poles, respectively.) Thus, the points in the unit ball associated to the rotated states 
\begin{align*}\pi(123)\rho\pi(123)^*&=R_z(2\pi/3)\rho R_z(2\pi/3)^*,\\ 
\pi(132)\rho\pi(132)^*&=R_z(4\pi/3)\rho R_z(4\pi/3)^*,\\
\pi(12)\rho\pi(12)^*&=R_x(\pi)\rho R_x(\pi)^*
\end{align*}
are affinely independent in $\mathbb{R}^3$, so their affine hull yields all hermitian matrices of trace 1, and the equality (\ref{e:2}) follows.
\end{example}


\section{Connections with $C^*$-envelopes}\label{s_conncstarenv}

Recall that, if $\cl S$ is an operator system, 
its \textit{$C^*$-envelope} $C^*_e(\mc{S})$ is the unital C*-algebra, uniquely determined up to isomorphism by the following universal property: there is a unital complete order embedding $\iota:\cl{S}\to C^*_{e}(\cl{S})$ such that $C^*(\iota(\cl{S}))=C^*_{e}(\cl{S})$, and for any $C^*$-algebra 
$\cl A$ and unital complete order embedding $\vphi:\cl{S}\to \cl A$ with $C^*(\vphi(\cl{S})) = \cl A$, there is a surjective $*$-homomorphism $\pi : \cl A\to C^*_{e}(\cl{S})$ such that $\pi\circ\vphi=\iota$. 
In this penultimate section, we show that the examples in Sections \ref{ss_POVMNS} (PVM self-tests), \ref{s:qg} (quantum graph coloring) and \ref {s:schur} (Schur product channels) are all instances of a single phenomena: unique state extension across 
$$\mc{S}_A\ten_{\rm c} \mc{S}_B\subseteq C^*_e(\mc{S}_A)\ten_{\max} C^*_e(\mc{S}_B),$$
for pertinent operator systems $\mc{S}_A$ and $\mc{S}_B$ where the latter 
inclusion is valid. 

By contrast, general abstract self-testing concerns unique state extensions across
$$\mc{S}_A\ten_{\rm c} \mc{S}_B\subseteq C^*_u(\mc{S}_A)\ten_{\max} C^*_u(\mc{S}_B),$$
the tensor product of \textit{universal} C*-covers. 
Similar comparisons can be made when self-testing among models which factor through $C^*_e(\mc{S}_A)\ten_{\max} C^*_e(\mc{S}_B)$.

The observations that follow rely mainly on the work \cite{bks}, but along the way we establish new dilation results for stochastic operator matrices. 


\subsection{PVM's}

As mentioned in Remark \ref{r_POVMST}, PVM self-testing fits in our general framework, where the family of states $s : C^*_u(\cl S_{X,A})\otimes_{\max} C^*_u(\cl S_{Y,B})\to \C$ to be self-tested is restricted to 
those that factor through the quotient map
$$\pi_A\ten\pi_B:C^*_u(\cl S_{X,A})\otimes_{\max} C^*_u(\cl S_{Y,B})\to 
\cl A_{X,A}\otimes_{\max} \cl A_{Y,B}.$$
Note that the set of such states is precisely
$$(\pi_A\ten\pi_B)^*(\cl{S}(\cl A_{X,A}\otimes_{\max} \cl A_{Y,B}))\cong \cl{S}(\cl A_{X,A}\otimes_{\max} \cl A_{Y,B}).$$
It is well-known that
$$\cl A_{X,A}\cong \underbrace{\cl D_{A}\ast_1\cdots\ast_1 \cl D_{A}}_{|X| \mbox{ \tiny times}}\cong C^*_e(\cl S_{X,A}),$$
and similarly for $\cl A_{Y,B}$,  (for the latter isomorphism, see e.g. \cite[Corollary 2.9]{bks}). Thus, a state $f:\cl{S}_{X,A}\ten_c\cl{S}_{Y,B}\to\C$ has a unique extension to $$(\pi_A\ten\pi_B)^*(\cl{S}(C^*_e(\cl{S}_{X,A})\otimes_{\max} C^*_e(\cl{S}_{Y,B}))$$
if and only if it extends uniquely across
$$\cl{S}_{X,A}\ten_{\rm c}\cl{S}_{Y,B}\subseteq C^*_e(\cl{S}_{X,A})\otimes_{\max} C^*_e(\cl{S}_{Y,B}),$$
the above inclusion being valid by \cite[Lemma 3.10]{bks}.


\subsection{Semi-classical SOM's}
Extending the setup of Subsection \ref{s:qg}, let 
$$\cl B_{X,A} = \underbrace{M_A\ast_1\cdots\ast_1 M_A}_{|X| \mbox{ times}}.$$
For each $x\in X$, write $\{\epsilon_{x,a,a'} : a,a'\in A\}$ 
for the canonical matrix unit system of the $x$-th copy of $M_A$, and let
$$\cl R_{X,A} = {\rm span}\{\epsilon_{x,a,a'} : x\in X, a,a'\in A\},$$
considered as an operator subsystem of $\cl B_{X,A}$. By \cite[Corollary 2.9]{bks}, the $C^*$-algebra $\cl B_{X,A}$ is universal for unital $*$-homomorphisms $M_A\to \cl D_{X}\ten\cl A$, where $\cl A$ is a unital $C^*$-algebra (see \cite[Definition 2.2]{bks}), and the pair $(\cl R_{X,A},\cl B_{X,A})$ satisfies the hypotheses of \cite[Theorem 3.8]{bks}. Hence, $C^*_e(\cl R_{X,A})=\cl B_{X,A}$, and from \cite[Lemma 3.10]{bks} (or \cite[Lemma 2.8]{pt}),
$$\cl R_{X,A}\ten_{\rm c} \cl R_{Y,B}\subseteq C^*_e(\cl R_{X,A})\ten_{\rm max}C^*_e(\cl R_{Y,B}).$$
Following \cite[\S7.1]{tt-QNS}, an SOM $E\in M_X\ten M_A\ten \cl{B}(H)$ is called \textit{semi-classical} if $E=\sum_{x\in X}\epsilon_{x,x}\ten E_x$ with $E_x\in (M_A\ten\cl{B}(H))^+$ and $\tr_A E_x=I_H$, $x\in X$.
In \cite[Theorem 7.5]{tt-QNS} it was shown that semi-classical SOMs correspond to unital completely positive maps $\cl{R}_{X,A}\to\cl B(H)$. Moreover, the state space of $\cl{R}_{X,A}\ten_{\rm c}\cl{R}_{Y,A}$ is affinely isomorphic to the set of classical-to-quantum no-signalling correlations $\Gamma:\cl D_X\ten \cl D_Y\to M_A\ten M_B$ \cite[Theorem 7.7]{tt-QNS}. 

In light of the previous paragraphs, 
Corollary \ref{c:s_k} means that the state  $\tilde{s}_{\cl K_4}:\cl{R}_{4,2}\ten_{\rm c}\cl{R}_{4,2}\to \mathbb{C}$ given by 
$$\tilde{s}_{\cl K_4}(w)
= {\rm tr}_2((\pi_{\cl K_4}\cdot \pi_{\cl K_4})(w)), \ \ \ 
w\in \cl R_{4,2}\otimes_{\rm c} \cl R_{4,2},$$
has a unique extension to the finite-dimensional states of $C^*_e(\cl R_{4,2})\ten_{\rm max} C^*_e(\cl R_{4,2})$. Similarly, Proposition \ref{p:s_k} means that $\tilde{s}_{\cl K_4}$ is a self-test for the class of finite-dimensional semi-classical SOM models which factor through $C^*_e(\cl R_{4,2})\ten_{\rm max} C^*_e(\cl R_{4,2})$.


\subsection{Unistochastic Operator Matrices}\label{uni}

It is known that the entries of any SOM $E=(E_{x,x',a,a'})\in M_X\ten M_A\ten \cl{B}(H)$ can be represented as $E_{x,x',a,a'}=V_{a,x}^*V_{a',x'}$ for an isometry 
$V:H^X\to K^A$, where $V = (V_{a,x})_{a,x}$ \cite[Theorem 3.1]{tt-QNS}. When $X=A$, it is natural to study further 
the unistochastic SOM's (USOM's), that is, 
those SOM's for which $V$ can be taken unitary (see Section \ref{s:schur}). 
In this subsection, we show that any SOM can be dilated to a USOM, thereby establishing a matricial version of Naimark dilation between POVM's and PVM's. Along the way, we connect these notions to $C^*$-envelopes of pertinent operator systems, as done in the previously in this section. 


Let $\cl{B}_X$ denote the universal $C^*$-algebra generated by the elements $u_{a,x}$ $a,x\in X$ of a unitary matrix $u=(u_{a,x})_{a,x}$, commonly known as the Brown algebra \cite{brown}. It was shown in \cite[Proposition 2.3]{ghj} that if $S_1^X$ is the dual operator space $(M_X)^*$, 
then the map $e_{a,x}\mapsto u_{a,x}$, from 
$S^X_1$ into $\mathrm{span}\{u_{a,x}: a,x\in X\}\subseteq\cl{B}_X$
is a completely isometric isomorphism. 

Recall from Subsection \ref{ss_QNS} the universal TRO of a block operator isometry $v=(v_{a,x})_{a,x\in X}$ \cite{tt-QNS}, hereby denoted simply $\cl V_X$ as we assume 
that $X = A$ throughout this subsection; thus, $\cl V_X$ is universal for the relations
\begin{equation}\label{e:v}\sum_{a\in X}v_{a,x}^*v_{a,x'}=\delta_{x,x'}1, \ \ \ x,x'\in X,\end{equation}
in that every (concrete) block operator isometry $V = (v_{a,x})_{a,x}$, 
whose entries lie in $\cl B(H,K)$ for some Hilbert spaces $H$ and $K$,
gives rise to a unique 
ternary morphism $\theta_V : \cl V_X\to \cl B(H)$ such that $\theta(v_{a,x}) = V_{a,x}$, $x,a\in X$. 
Since the entries of $u$ satisfy the relations (\ref{e:v}), there 
exists a (unique) ternary morphism $\vphi:\cl{V}_X\to\cl{B}_X$ such that $\vphi(v_{a,x})=u_{a,x}$. On the other hand, since $v=(v_{a,x})_{a,x}$ belongs to the unit ball of 
$M_X(\cl V_X)$, the canonical completely isometric identification $M_X(\cl V_X)\cong \cl{CB}(S_1^X,\cl V_X)$
(see e.g.  \cite[Proposition 1.5.14]{BLM}),
yields a complete contraction $\vphi':\mathrm{span}\{u_{a,x}:a,x\in X\}\to\mathrm{span}\{v_{a,x}:a,x\in X\}$ satisfying $\vphi'(u_{a,x})=v_{a,x}$. Since ternary morphisms are completely contractive, it follows that the restriction
$$\kappa : \mathrm{span}\{u_{a,x}:a,x\in X\}\to \mathrm{span}\{v_{a,x}:a,x\in X\}$$
of $\vphi'$ is a complete isometry (see also \cite[Proposition 4.1]{cltt}). This first order isomorphism leads to the following dilation result.

\begin{proposition}\label{p:dilate} 
Let $H$ and $K$ be Hilbert spaces, for which
$(v_{a,x})_{a,x}\in M_X(\cl{B}(H,K))$. 
There exists a Hilbert space $L$, isometries $w_1:H\to L$, $w_2:K\to L$ and a unital *-homomorphism $\pi:\cl{B}_X\to\cl{B}(L)$ satisfying $$v_{a,x}=w_2^*\pi(u_{a,x})w_1 \ \ \ \textnormal{and}  \ \ w_2w_2^*\pi(u_{a,x})w_1=\pi(u_{a,x})w_1 \ \ \ x,a\in X.$$
\end{proposition}

\begin{proof} 
By the injectivity of $\cl{B}(H,K)$, we can extend $\kappa$ to a complete contraction $\widetilde{\kappa}:\cl{B}_X\rightarrow\cl{B}(H,K)$. By the 
Haagerup-Paulsen-Wittstock representation theorem for complete contractions, there exists a Hilbert space $L$, isometries $w_1:H\to L$, $w_2:K\to L$ and a unital *-homomorphism $\pi:\cl{B}_X\to\cl{B}(L)$ such that $\widetilde{\kappa}(a)=w_2^*\pi(a)w_1$, $a\in \cl{B}_X$. (We note that this 
can be derived by appling the proof of \cite[Theorem B.7]{BO} together with the rectangular version of Paulsen's off-diagonal trick \cite[Lemma 1.3.15]{BLM}
and refer to \cite[Theorem 2.1.12]{ekt} for a complete argument.) 
In particular, 
$$v_{a,x}=\kappa(u_{a,x})=w_2^*\pi(u_{a,x})w_1, \ \ \ x,a\in X.$$
Then $v=(w_2^*\ten I_X)\pi_X(u)(w_1\ten I_X)$, where $\pi_X:=(\pi\ten\id_{M_X})$. Since $$(w_2\ten I_X)v=(w_2w_2^*\ten I_X)\pi_X(u)(w_1\ten I_X)$$ 
with both $(w_2\ten I_X)v$ and $\pi_X(u)(w_1\ten I_X)$ isometries, we have that 
$\pi_X(u)(w_1\ten I_X)H\ten\C^X\subseteq\mathrm{ran}(w_2\ten I_X)$.
\end{proof}

\begin{corollary}\label{c:order2} There exists a complete order isomorphism 
$\vphi:\cl{T}_X\to\mathrm{span}\{u_{a,x}^*u_{a',x'}:x,x',a,a'\in X\}$ satisfying
\begin{equation}\label{e:formulae}\vphi(v_{a,x}^*v_{a',x'})=u_{a,x}^*u_{a',x'}, \ \ \  x,x',a,a'\in X.\end{equation}
\end{corollary}

\begin{proof} 
By the universal property of 
$v=(v_{a,x})_{a,x}$, there exists a (non-degenerate) TRO morphism $\cl{V}_X\to\cl{B}_X$ mapping $v_{a,x}$ to $u_{a,x}$, $x,a\in X$. 
The latter morphism induces a unital $*$-homomorphism $\vphi:\cl{C}_X\to \mc{B}_X$ mapping $v_{a,x}^*v_{a',x'}$ to  $u_{a,x}^*u_{a',x'}$, 
$x,x',a,a'\in X$, where $\cl{C}_X:=\cl C_{X,X}$.
Thus $\vphi|_{\cl{T}_X}$ is a unital completely positive map satisfying (\ref{e:formulae}).

Represent $v$ faithfully inside $M_X(\mc{B}(H,K))$ 
for some Hilbert spaces $H$ and $K$. 
By Proposition \ref{p:dilate}, there exists a Hilbert space $L$, isometries $w_1:H\to L$, $w_2:K\to L$ and a unital *-homomorphism $\pi:\mc{B}_X\to\mc{B}(L)$ satisfying $$v_{a,x}=w_2^*\pi(u_{a,x})w_1 \ \ \ \textnormal{and}  \ \ w_2w_2^*\pi(u_{a,x})w_1=\pi(u_{a,x})w_1, \ \ \ x,a\in X.$$
But then 
$$v_{a,x}^*v_{a',x'}=w_1^*\pi(u_{a,x}^*)w_2w_2^*\pi(u_{a',x'})w_1=w_1^*\pi(u_{a,x}^*u_{a',x'})w_1,$$
for all $x,x',a,a'\in X$.
It follows that
$w_1^*\pi(\cdot)w_1$ is a unital completely positive inverse to $\vphi$.
\end{proof}

Let $E$ be a unistochastic operator matrix acting on a 
Hilbert space $H$. When $H=\C$, and $E$ is diagonal in the sense that $E_{x,x',a,a'}=\delta_{x,x'}\delta_{a,a'}E_{x,x,a,a}$, we recover the usual notion of unistochastic matrices. A simple application of the previous results yields the following dilation property.

\begin{corollary} 
Let $H$ be a Hilbert space and 
$E=(E_{x,x',a,a'})\in M_X\ten M_X\ten \cl{B}(H)$ be a SOM. Then there 
exists a Hilbert space $K$, an isometry $W:H\to K$ and a 
block operator unitary $(U_{a,x})_{a,x\in X} \in M_X\otimes\cl B(K)$, 
such that $E_{x,x',a,a'}=W^*U_{a,x}^*U_{a',x'}W$.
\end{corollary}

\begin{proof} By the universal property of $\cl{T}_X$, there is unital completely positive map $\psi:\cl{T}_X\to\mc{B}(H)$ satisfying $\psi(e_{x,x',a,a'})=E_{x,x',a,a'}$. Let  $\vphi^{-1}:\mathrm{span}\{u_{a,x}^*u_{a',x'}: x,x',a,a'\in X\}\to \cl{T}_X$ be the complete order isomorphism from Corollary \ref{c:order2}. Extending the unital completely positive map $\psi\circ\vphi^{-1}$ to $\cl{B}_X\to\mc{B}(H)$ (by injectivity of $\mc{B}(H)$) and appealing to a Stinespring representation of the extension yields the desired conclusion.
\end{proof}

Using results from \cite{bks}, 
we now show that the $C^*$-envelope of $\cl{T}_X$ is the 
C*-subalgebra $C^*(\cl U)$ of $\cl{B}_X$, generated by the 
operator system
$$\cl U = {\rm span}\{u_{a,x}^*u_{a',x'}:x,x',a,a'\in X\}.$$

\begin{proposition}\label{p:env} 
Let $X$ be a finite set. Then
$C^*_{e}(\cl{T}_X)=C^*(\cl U)$.
\end{proposition}

\begin{proof} 
Throughout the proof, a unital $*$-homomorphism between unital $C^*$-algebras will simply be called a morphism. Let $\cl{A}$ be the unital $C^*$-algebra which is universal for morphisms $M_X\mapsto M_X\ten {\cl A}$
in the sense that 
there exists a morphism $\alpha:M_X\to M_X\ten \cl{A}$ 
such that, for any unital $C^*$-algebra $\cl{B}$ and morphism $\beta:M_X\to M_X\ten \cl{B}$, there is a unique morphism $\Lambda:\cl{A}\to\cl{B}$ such that $\beta=(\id\ten\Lambda)\alpha$ 
(see \cite[Theorem 2.3]{bks}). Viewing $\cl{A}\subseteq\mc{B}(H)$, there exists a unitary $w:\C^X\ten H\to\C^X\ten H$ such that $\alpha(T) = w^*(T \ten I_H)w$, $T\in M_X$. Since $\alpha$ has range in $M_X\ten\cl{A}$, we have that $w_{a,x}^*w_{a',x'}\in\cl{A}$ for all $x,x',a,a'\in X$. Moreover, $\cl{A}$ is generated by $\{(\rho\ten\id)\alpha(A) : \rho\in M_X^*, \ A\in M_X\}$ \cite[Theorem 2.3]{bks}, so $\cl{A}=C^*(\{w_{a,x}^*w_{a',x'}: x,x',a,a'\in X\})$. 

If $u=(u_{a,x})_{a,x}$ is the generating matrix for $\cl{B}_X$, 
$$M_X\ni T\mapsto u^*(T\ten 1)u\in M_X\ten \cl{B}_X$$ 
is a morphism, so the universal property of $\cl{A}$ supplies a morphism $\Lambda:\cl{A}\to\cl{B}_X$ satisfying $\Lambda(w_{a,x}^*w_{a',x'})=u_{a,x}^*u_{a',x'}$, $x,x',a,a'\in X$. Then $\Lambda(\cl{A})\subseteq 
C^*(\cl U)$, necessarily. 

By the universal property of $\cl{B}_X$, there is a morphism 
$\pi : \cl{B}_X\to\cl{B}(H)$ such that $\pi(u_{a,x}) = w_{a,x}$, $x,a\in X$. The restriction of $\pi$ to 
$C^*(\cl U)$ is an inverse to $\Lambda$, so that 
$C^*(\cl U)\cong\cl{A}$.
The desired conclusion then follows from \cite[Theorem 3.8]{bks}.
\end{proof}

It follows from the proof of Proposition \ref{p:env} that the operator system $\cl{T}_X$ satisfies the universal property of \cite[Theorem 3.3]{bks} for unital completely positive maps $M_X\to M_X\ten \cl{T}_X$, so that \cite[Corollary 3.11]{bks} implies $$\cl{T}_X\ten_{c}\cl{T}_Y\subseteq C^*_e(\cl{T}_X)\ten_{\max} C^*_e(\cl{T}_Y).$$
Therefore, the self-testing examples from Proposition \ref{p:prop} give rise to finite-dimensional states $f:\cl{T}_X\ten_{c} \cl{T}_Y\to\C$ which extend uniquely to finite-dimensional states of $C^*_e(\cl{T}_X)\ten_{\max} C^*_e(\cl{T}_Y)$.


\section{Questions}\label{s_questions}

The notion of an approximate dilation $\tilde{S}$ of a model $S$
over the pair $(\cl S_{X,A},\cl S_{Y,B})$ was defined in 
\cite{zhao} for the case of classes of quantum models, 
and can be easily extended to classes of quantum 
commuting models over
a finitary context $(\cl S_A,\cl S_B)$,
where $\cl S_A = {\rm span}\{f_i\}_{i=1}^k$ and 
$\cl S_B = {\rm span}\{g_j\}_{j=1}^l$:
given $\delta > 0$, a quantum commuting model 
$\tilde{S} = (\tilde{H},\tilde{\nph}_A,\tilde{\nph}_B\tilde{\xi})$ is said to $\delta$-dilate a quantum commuting 
model $S = (H,\nph_A,\nph_B,\xi)$ if there exist 
an auxiliary system $H_{\rm aux}$ and a unit vector 
$\xi_{\rm aux}\in H_{\rm aux}$, and 
a local isometry $V : H\to \tilde{H}\otimes H_{\rm aux}$ 
such that 
$$V\phi_A(f_i)\phi_B(g_j)\xi\sim^{\delta} 
\tilde{\phi}_A(f_i)\tilde{\phi}_B(g_j)\tilde{\xi}
\otimes\xi_{\rm aux}, \ \ \ i\in [k], j\in [l]. $$
We say that a correlation $p$ of quantum commuting type 
over $(\cl S_A,\cl S_B)$ is a robust 
quantum commuting self-test if there exists a model 
$\tilde{S}$ of $p$ such that for every $\epsilon > 0$
there exists $\delta > 0$ with the property that, whenever 
$p_S$ is a correlation arising from a quantum commuting model 
$S$ over $(\cl S_A,\cl S_B)$ with $\|p_S - p\|_1 < \epsilon$, 
we have that $\tilde{S}$ is a $\delta$-dilation of $S$. 
We have the following implications for 
a correlation $p$: 
$$p \mbox{ is a robust self-test} \ \Rightarrow \ 
p \mbox{ is a weak self-test} \ \Rightarrow \ 
p \mbox{ is a self-test}.
$$
These implications show that weak self-testing 
is a natural concept to study, but suggest the following 
initial question:

\begin{question}\label{q_weak}
\rm
Is every weak self-test a self-test?
\end{question}

If $S$ is a centrally supported quantum POVM model for $p\in\cl C_{\rm q}$ and $S_r$ is the reduced model of $S$, then $S\preceq S_r$ (see \cite[Lemma 4.2]{pszz}), which reduces the question of self-testing of $p$ to the study of full rank models of $p$. The latter was important for establishing a number of self-testing results and allowed the application of  representations of certain algebraic relations (see e.g. \cite{supic-bowles} for a self-test of the optimal strategy of CHSH game, and \cite{mps} for self-tests of some synchronous correlations). We therefore formulate the following:

\begin{question}
For general model $S=(_{\cl A}H_{\cl B}, \varphi_A, \varphi_B, \xi)$, 
is it true that $S\preceq S_r$?
\end{question}

We point out that, currently, 
we do not see how the construction of local isometries for the dilation $S\preceq S_r$ can be modified to give one in the non-tensor split case. 

In the definition of self-testing, an assumption is made on 
the existence of an auxiliary quantum system that ampliates the ideal 
model for the self-test in question. What types of auxiliary systems 
may arise is an interesting question in its own right. As an example, 
we formulate the following:

\begin{question} 
When applying our definition of self-testing to (quantum) no-signalling correlations of quantum type (e.g. for the classes $\mathcal{C}_{\rm q}$ and 
$\mathcal{Q}_{\rm q}$), do the auxiliary bimodules arising from the dilation 
pre-order automatically tensor factorise into quantum spacial systems? 
\end{question}



\medskip

\noindent 
{\bf Acknowledgements. } 
We are grateful to The Banff International Research Station for Mathematical Innovation and Discovery (BIRS) for hospitality and support during the initial phase of the project. The first author was partially supported by the NSERC Discovery Grant RGPIN-2023-05133. The second author was supported by NSF grants 
CCF-2115071 and DMS-2154459. The third author was supported by the Swedish Research Council project grant 2023-04555 and GS Magnusons Fond MF2022-0006.



\begin{thebibliography}{99}

\bibitem{acin-etal}
{\sc A.\,Ac\'in, T.\,Fritz, A.\,Leverrier and A.\,B.\,Sainz},
{\it A combinatorial approach to nonlocality and contextuality},
{\rm Comm. Math. Phys. 334 (2015), no. 2, 533-628}.

\bibitem{act}
{\sc M.\,Anoussis, A.\,Chatzinikolaou and I.\,G.\,Todorov},
{\it Operator systems, contextuality and non-locality},
{\rm preprint (2024), arXiv:2405.11152}.

\bibitem{arv}
{\sc W.\,B.\,Arveson}, 
{\it Notes on the unique extension property},
{\rm Unpublished note, 2003}.

\bibitem{brv-etal}
{\sc K.\,Bharti, M.\,Ray, A.\,Varvitsiotis, N.\,A.\,Warsi, A.\,Cabello and L.-C.\,Kwek}, 
{\it Robust self-testing of quantum systems via noncon- textuality inequalities}, 
{\rm Phys. Rev. Letters 122 (2019), no. 25, 250403}.

\bibitem{BLM}
{\sc D.\,P.\,Blecher and C.\,Le Merdy},
{\em Operator Algebras and Their Modules: An operator space approach},
{\rm Oxford University Press, 2004}.

\bibitem{bks}
{\sc A.\,Bochniak, P.\,Kasprzak and P.\,M.\,So\l tan},
{\it Quantum correlations on quantum spaces}, {\rm  Int. Math. Res. Not. IMRN 14 (2023), 12400-12440}.


\bibitem{bsca-1}
{\sc J.\,Bowles, I.\,Supi\'c, D.\,Cavalcanti and 
A.\,Ac\'in}, 
{\it Device-independent entanglement certification of all entangled states}, 
{\rm Phys. Rev. Lett. 121 (2018), no. 18, 180503}.


\bibitem{bsca-2}
{\sc J.\,Bowles, I.\,\v Supi\'c, D.\,Cavalcanti and 
A.\,Ac\'in}, 
{\it Self-testing of Pauli observables for device-independent entanglement certification}, 
{\rm Phys. Rev. A 98 2018, no. 4, 042336}.



\bibitem{bhtt-JFA}
{\sc M.\,Brannan, S.\,J.\,Harris, I.\,G.\,Todorov and L.\,Turowska}, 
{\it Synchronicity for quantum non-local games},
{\rm J. Funct. Anal. 284 (2023), no. 2, Paper No. 109738, 54 pp.}

\bibitem{bhtt-Adv}
{\sc M.\,Brannan, S.\,J.\,Harris, I.\,G.\,Todorov and L.\,Turowska}, 
{\it Quantum no-signalling bicorrelations},
{\rm Adv. Math. 449 (2024), Paper No. 109732, 81 pp.}


\bibitem{broadbent}
{\sc A.\,Broadbent, A.\,Mehta and Y.\,Zhao}, 
{\it Quantum delegation with an off-the-shelf device},
{\rm Leibniz International Proceedings in Informatics (2023), arXiv:2304.03448}.

\bibitem{brown}
{\sc L.\,G.\,Brown}, 
{\it Ext of certain free product $C^*$-algebras},
{\rm J. Operator Theory 6 (1981),
no. 1, 135-141}.

\bibitem{BO}
{\sc N.\,P.\,Brown and N.\,Ozawa},
{\it $C^*$-algebras and finite-dimensional approximations},
{\rm  American Mathematical Society, 2008}.


\bibitem{chsh}
{\sc J.\,F.\,Clauser, M.\,A.\,Horne, A.\,Shimony and R.\,A.\,Holt},
{\it Proposed experiment to test local hidden-variable theories},
{\rm Phys. Rev. Letters 23 (1969), 880-884}.

\bibitem{cgjv}
{\sc A.\,Coladangelo, A.\,B.\,Grilo, S.\,Jeffery and T.\,Vidick}, 
{\it Verifier-on-a-leash: New schemes for verifiable delegated quantum computation, with quasilinear resources}, 
{\rm Th. Comp. 20 (2024), no. 1, 1-87}.



\bibitem{cckl}
{\sc A.\,Conlon, J.\,Crann, D.\,W.\,Kribs and R.\,H.\,Levene},
{\it Quantum teleportation in the commuting operator framework},
{\rm Ann. Henri Poincar\'{e} 24 (2023), no. 5, 1779-1821}.

\bibitem{cklt} 
{\sc J.\,Crann, D.\,W.\,Kribs, R.\,H.\,Levene and I.\,G.\,Todorov}, 
{\it State convertibility in the von Neumann algebra framework}, 
{\rm Comm. Math. Phys. 378 (2020), no. 2, 1123-1156}.

\bibitem{cltt}
{\sc J.\,Crann, R.\,H.\,Levene, I.\,G.\, Todorov and L.\,Turowska},
{\it Values of cooperative quantum games},
{\rm preprint (2023), arXiv:2310.17735}.

\bibitem{davk}
{\sc K.\,R.\,Davidson and M.\,Kennedy}, 
{\it The Choquet boundary of an operator system},
{\rm Duke Math. J. 164 (2015), no. 15, 2989-3004}.


\bibitem{dixmier}
{\sc J.\,Dixmier},
{\it C*-algebras}, 
{\rm North Holland, 2011}.

\bibitem{dixmier_von_neumann}
{\sc J.\,Dixmier},
{\it Von Neumann algebras}, 
{\rm Elsevier, 2011}.


\bibitem{dw}
{\sc R.\,Duan and A.\,Winter},
{\it No-signalling assisted zero-error capacity of quantum channels and an
information theoretic interpretation of the Lov\'{a}sz number},
{\rm IEEE Trans. Inf. Theory 62 (2016), no. 2, 891-914}.

\bibitem{EL} 
{\sc E.\,G.\,Effros and C.\,E.\,Lance},
{\it Tensor products of operator algebras},
{\rm Adv. Math. 25 (1977), no. 1, 1-34}.

\bibitem{ekt} 
{\sc G.\,K.\,Eleftherakis, E.\,T.\,A.\,Kakariadis and I.\,G.\,Todorov},
{\it Symmetrisations of operator spaces},
{\rm preprint (2025), arXiv:2503.15192}. 

\bibitem{ghj}
{\sc L.\,Gao, S.\,Harris and M.\,Junge},
{\it Quantum teleportation and super-dense coding in operator algebras}.
{\rm Int. Math. Res. Not. IMRN (2021), no. 12, 9146-9179}.

\bibitem{haagerup-musat}
{\sc U.\,Haagerup and M.\,Musat},
\textit{Factorization and dilation problems for completely positive maps on von Neumann algebras},
{\rm Comm. Math. Phys. 303 (2011), no. 2, 555-594}.


\bibitem{jnvwy}
{\sc Z.\,Ji, A.\,Natarajan, T.\,Vidick, J.\,Wright and H.\,Yuen},
{\it MIP*=RE},
{\rm preprint (2020), arXiv:2001.04383}.



\bibitem{js} 
{\sc V.\,F.\,R.\,Jones and V.\,S.\,Sunder},
{\it Introduction to subfactors},
{\rm Cambridge University Press, 1997}.


\bibitem{kadison-ringrose}
{\sc R.\,V.\,Kadison and J.\,R.\,Ringrose},
{\it Fundamentals of the theory of operator algebras II}, 
{\rm American Mathematical Society, 1997}.

\bibitem{kar}
{\sc P.\,N.\,Kar},
{\it Robust self-testing for synchronous correlations and games},
{\rm preprint (2025), arXiv:2503.23500}.


\bibitem{KPTT11}
{\sc A.\,S.\,Kavruk, V.\,I.\,Paulsen, I.\,G.\,Todorov and M.\,Tomforde},
{\it Tensor products of operator systems},
{\rm J.\ Funct.\ Anal.\ 261 (2011), no.\ 2, 267-299}.

\bibitem{kps}
{\sc S.-J.\,Kim, V.\,I.\,Paulsen and C.\,Schafhauser},
{\it A synchronous game for binary constraint systems},
{\rm J. Math. Phys. 59 (2018), no. 3, 032201, 17 pp.}


\bibitem{kita}
{\sc Y.\,Kitajima},
{\it Local operations and completely positive maps in algebraic quantum field theory},
{\rm In: M. Ozawa, J. Butterfield, H. Halvorson, M. R\'{e}dei, Y. Kitajima, Y, F. Buscemi (eds) Reality and Measurement in Algebraic Quantum Theory, NWW 2015, Springer Proceedings in Mathematics \& Statistics, vol. 261, Springer, Singapore, 2018}.

\bibitem{kr}
{\sc A.\,Klappenecker and M.\,R\"{o}tteler},
{\it On the monomiality of nice error bases},
{\rm IEEE Trans. Inf. Theory 51 (2005), no. 3, 1084-1089}.

\bibitem{LS}
{\sc L.\,J.\,Landau and R.\,F.\,Streater},
{\it On Birkoff's theorem for doubly stochastic completely positive maps of matrix algebras},
{\rm Lin. Alg. App. 193, 107-127 (1993)}.

\bibitem{lsww}
{\sc L.\,van Luijk, A.\,Stottmeister, R.\,F.\,Werner and  H.\,Wilming},
{\it Pure state entanglement and von Neumann algebras},
{\rm preprint (2024), arXiv:2409.17739}.


\bibitem{lmprsstw}
{\sc M.\,Lupini, L.\,Man\v{c}inska, V.\,I.\,Paulsen, D.\,E.\,Roberson, G.\,Scarpa, S.\,Severini, I.\,G.\,Todorov and A.\,Winter},
{\it Perfect strategies for non-local games},
{\rm Math. Phys. Anal. Geom. 23 (2020), no. 1, Paper No. 7, 31 pp.}



\bibitem{mps}
{\sc L.\,Man\v{c}inska, J.\,Prakash and C.\,Schafhauser},
{\it Constant-sized robust self-tests for states and measurements of unbounded dimension},
{\rm Comm. Math. Phys.   
405 (2024), no. 9, Paper No. 221, 36 pp.}

\bibitem{myao}
{\sc D.\,Mayers and A.\,Yao},
{\it Self testing quantum apparatus},
{\rm Quantum Inf. Comp. 4 (2004), no. 4, 273-286}.


\bibitem{mpw}
{\sc A.\,Mehta, C.\,Paddock and L.\,Wooltorton},
{\it Self-testing in the compiled setting via tilted-CHSH inequalities},
{\rm preprint (2024), arXiv:2406.04986}.


\bibitem{mckyc}
{\sc M.\,McKague, T.\,H.\,Yang and V.\,Scarani}, 
{\it Robust self-testing of the singlet}, 
{\rm J. Phys. A 45 (2012), no. 45, 455304}.


\bibitem{npa}
{\sc M.\,Navascu\'es, S.\,Pironio and A.\,Ac\'in}, 
{\it A convergent hierarchy of semdeﬁnite programscharacterizing the set of quantum correlations}, 
{\rm New J. Phys. 10 (2008), 073013}.

\bibitem{os}
{\sc V.\,Ostrovskyi and Yu.\,Samoilenko},
{\it Introduction to the theory of representations of finitely presented *-algebras. I. Representations by bounded operators},
{\rm Harwood Academic Publishers, 1999}.

\bibitem{pszz}
{\sc C.\,Paddock, W.\,Slofstra, Y.\,Zhao and Y.\,Zhou},
{\it An operator-algebraic formulation of self-testing},
{\rm Ann. Henri Poincare 25 (2024), no. 10, 4283-4319}.

\bibitem{Pa}
{\sc V.\,I.\,Paulsen},
{\it Completely bounded maps and operator algebras},
{\rm Cambridge University Press, 2002}.


\bibitem{psstw}
{\sc V.\,I.\,Paulsen, S.\,Severini, D.\,Stahlke, I.\,G.\,Todorov and A.\,Winter}, 
{\it Estimating quantum chromatic numbers}, 
{\rm J. Funct. Anal. 270 (2016), no. 6, 2188-2222}.


\bibitem{pt}
{\sc V.\,I.\,Paulsen and I.\,G.\,Todorov}, 
{\it Quantum chromatic numbers via operator systems},
{\rm Q. J. Math. 66 (2015), no. 2, 677-692}.


\bibitem{pabgms}
{\sc S.\,Pironio, A.\,Ac\'in, N.\,Brunner, N.\,Gisin, S.\,Massar and V.\,Scarani}, 
{\it Device-independent quantum key distribution secure against collective
attacks}, 
{\rm New J. Phys. 11 (2009), no. 4, 045021}.

\bibitem{ruv}
{\sc B.\,W.\,Reichardt, F.\,Unger and U.\,Vazirani}, 
{\it Classical command of quantum systems}, 
{\rm Nature 496 (2013), 456-460}.


\bibitem{seitz}
{\sc G.\,Seitz},
{\it Finite groups having only one irreducible representation of degree greater than one},
{\rm Proc. Amer. Math. Soc. 19 (1968), 459-461}.

\bibitem{simon}
{\sc B.\,Simon},
{\it Representations of finite and compact groups}, 
{\rm American Mathematical Society, 1996}.

\bibitem{supic-bowles}
{\sc I.\,S\v{u}pic and J.\,Bowles},
{\it Self-testing of quantum systems: a review},
{\rm Quantum 4, 337 (2020)}.

\bibitem{takesaki}
{\sc M.\,Takesaki},
{\it Theory of Operator Algebras I},
{\rm Springer-Verlag, 2001}.

\bibitem{takesaki2} 
{\sc M.\,Takesaki} ,
{\it Theory of Operator Algebras II},
{\rm Springer-Verlag, 2003}.

\bibitem{tt-QNS} 
{\sc I.\,G.\,Todorov and L.\,Turowska},
{\it Quantum no-signalling correlations and non-local games},
{\rm Comm. Math. Phys. 405 (2024), no. 6, 
Paper No. 141, 65 pp.}


\bibitem{tsirelson}
{\sc B.\,S.\,Tsirelson},
{\it Quantum analogues of Bell's inequalities. The case of two spatially divided domains},
{\rm Zap. Nauchn. Sem. Leningrad. Otdel. Mat. Inst. Steklov. (LOMI)   142 (1985), 174-194, 200}.

\bibitem{vw}
{\sc R.\,Verch and R.\,F.\,Werner},
{\it Distillability and positivity of partial transposes in general quantum field systems},
{\rm Rev. Math. Phys. 17 (2005), no. 5, 545--576}.

\bibitem{zhao}
{\sc Y.\,Zhao}, 
{\it Robust self-testing for nonlocal games with robust game algebras}, 
{\rm preprint (2024), arXiv:2411.03259}.



\end{thebibliography}
\end{document}